\documentclass[reqno]{amsart}

\usepackage{array,multirow}
\usepackage{amsmath}
\usepackage{amssymb}
\usepackage{amsfonts}
\usepackage{graphicx}
\usepackage{bm}
\usepackage{enumerate}

\newtheorem{theorem}{Theorem}
\newtheorem{lemma}{Lemma}

\newtheorem{corollary}{Corollary}
\newtheorem{definition}{Definition}

\numberwithin{equation}{section}

\newcounter{mnote}
\newcommand{\IF}{I\!\!F}         

\newcommand{\IL}{I\!\!L}

\newcommand{\Span}{\mbox{{\rm Span}}}

\newcommand{\Ri}{{}^{\mbox{\tiny \rm 3}}\!R} 
\newcommand{\hRi}{{}^{\mbox{\tiny \rm 3}}\!\hat{R}} 

\newcommand{\Remark}{\noindent{\bf Remark.~}}

\newcommand{\LRI}{\quad\Leftrightarrow\quad}

\newcommand{\RI}{\quad\Rightarrow\quad}

\newcommand{\leqs}{\leqslant}      
     
\newcommand{\geqs}{\geqslant}      
\newcommand{\ngeqs}{\ngeqslant}      

\newcommand{\tiL}{\mbox{{\tiny $L$}}}

\newcommand{\tiR}{\mbox{{\tiny $R$}}}
\newcommand{\tiX}{\mbox{{\tiny $X$}}}
\newcommand{\tiY}{\mbox{{\tiny $Y$}}}

\newcommand{\tiIL}{\mbox{{\tiny $\IL$}}}

\newcommand{\tiwedge}{\mbox{{\tiny $\wedge$}}}
\newcommand{\tivee}{\mbox{{\tiny $\vee$}}}

\newcommand{\N}{{\mathbb N}}       
\newcommand{\R}{{\mathbb R}}       

\newcommand{\cD}{{\mathcal D}}

\newcommand{\cL}{{\mathcal L}}
\newcommand{\cM}{{\mathcal M}}

\newcommand{\cO}{{\mathcal O}}

\newcommand{\cY}{{\mathcal Y}}

\newcommand{\ttK}{{\tt K}}

\newcommand{\ttk}{{\tt k}}

\newcommand{\bomega}{\boldsymbol{\omega}}

\def\mathbi#1{\textbf{\em #1}}

\newcommand{\biC}{\mathbi{C}}

\newcommand{\biL}{\mathbi{L}}

\newcommand{\biW}{\mathbi{W\,}}

\newcommand{\bib}{\mathbi{b \!\!}}

\newcommand{\bif}{\mathbi{f\,}}

\newcommand{\bij}{\mathbi{j\,}}

\newcommand{\biw}{\mathbi{w}}

\newcommand{\tbL}{\textbf{L}}

\newcommand{\tbW}{\textbf{W\,}}

\newcommand{\tbb}{\textbf{b}}

\newcommand{\tbf}{\textbf{f\,}}

\newcommand{\tbj}{\textbf{j}}

\newcommand{\tbv}{\textbf{v}}
\newcommand{\tbw}{\textbf{w}}

\newcommand{\hh}{\hat h}
\newcommand{\hj}{\hat \jmath}
\newcommand{\hk}{\hat k}
\newcommand{\hl}{\hat l}

\newcommand{\hD}{\hat D}

\newcommand{\hrho}{\hat \rho}
\newcommand{\htau}{\hat \tau}

\newcommand{\un}{\underline}       

\usepackage{color}

\begin{document}

\title[Rough Solutions of the Einstein Constraints on Closed Manifolds]
      {Rough Solutions of the Einstein Constraints
       on Closed Manifolds without Near-CMC Conditions}

\author[M. Holst]{Michael Holst}
\email{mholst@math.ucsd.edu}
\thanks{MH was supported in part by NSF Awards~0715146, 0411723, 
and 0511766, and DOE Awards DE-FG02-05ER25707 and DE-FG02-04ER25620.}

\author[G. Nagy]{Gabriel Nagy}
\email{gnagy@math.ucsd.edu}
\thanks{GN was supported in part by NSF Awards~0715146 and 0411723.}

\author[G. Tsogtgerel]{Gantumur Tsogtgerel}
\email{gantumur@math.ucsd.edu}
\thanks{GT was supported in part by NSF Awards~0715146 and 0411723.}

\address{Department of Mathematics\\
         University of California San Diego\\ 
         La Jolla CA 92093}

\date{\today}


\begin{abstract}
We consider the conformal decomposition of Einstein's constraint
equations introduced by Lichnerowicz and York, on a closed manifold.
We establish existence of non-CMC weak solutions using a combination of
{\em a priori} estimates for the individual Hamiltonian and momentum
constraints, barrier constructions and fixed-point techniques for the 
Hamiltonian constraint, Riesz-Schauder theory for the momentum constraint,
together with a topological fixed-point argument for the coupled system.
Although we present general existence results for non-CMC weak solutions 
when the rescaled background metric is in any of the three Yamabe classes,
an important new feature of the results we present for the positive Yamabe 
class is the absence of the near-CMC assumption,
if the freely specifiable part of the data given by the
traceless-transverse part of the rescaled extrinsic curvature 
and the matter fields are sufficiently small, and if
the energy density of matter is not identically zero.
In this case, the mean extrinsic curvature can be taken to be an 
arbitrary smooth function without restrictions on the size of its 
spatial derivatives, so that it can be arbitrarily far from constant,
giving what is apparently the first existence results for non-CMC 
solutions without the near-CMC assumption.
Using a coupled topological fixed-point argument that avoids
near-CMC conditions, we establish existence 
of coupled non-CMC weak solutions with
(positive) conformal factor $\phi \in W^{s,p}$,
where $p \in (1,\infty)$ and $s(p) \in (1+3/p,\infty)$.
In the CMC case, the regularity can be reduced to
$p \in (1,\infty)$ and $s(p) \in (3/p, \infty) \cap [1,\infty)$.
In the case of $s=2$, we reproduce the CMC existence results of 
Choquet-Bruhat~\cite{yCB04}, and
in the case $p=2$, we reproduce the CMC existence results of
Maxwell~\cite{dM05}, but with a proof that
goes through the same analysis framework that we use to 
obtain the non-CMC results.
The non-CMC results on closed manifolds here extend the 
1996 non-CMC result of Isenberg and Moncrief in three ways:
(1) the near-CMC assumption is removed in the case of the
    positive Yamabe class;
(2) regularity is extended down to the maximum allowed 
    by the background metric and the matter;
and
(3) the result holds for all three Yamabe classes.
This last extension was also accomplished recently by 
Allen, Clausen and Isenberg, although their result is restricted 
to the near-CMC case and to smoother background metrics and data.
\end{abstract}

\maketitle

\clearpage

{\tiny
\tableofcontents
}

\vspace*{-1.0cm}

\section{Introduction}
   \label{sec:intro}

In this article, we give an analysis of the coupled Hamiltonian and
momentum constraints in the Einstein equations on a 3-dimensional
closed manifold. We consider the equations with matter sources
satisfying an energy condition implied by the dominant energy
condition in the 4-dimensional spacetime; the unknowns are a
Riemannian three-metric and a two-index symmetric tensor. The
equations form an under-determined system; therefore, we focus
entirely on a standard reformulation used in both mathematical and
numerical general relativity, called the conformal method,
introduced by Lichnerowicz and York~\cite{aL44,jY71,jY72}. The
conformal method assumes that the unknown metric is known up to a
scalar field called a conformal factor, and also assumes that the
trace and a term proportional to the trace-free divergence-free part
of the two-index symmetric tensor is known, leaving as unknown a
term proportional to the traceless symmetrized derivative of a
vector. Therefore, the new unknowns are a scalar and a vector field,
transforming the original under-determined system for a metric and a
symmetric tensor into a (potentially) well-posed elliptic system for
a scalar and a vector field.  See~\cite{rBjI04} for a recent review
article.

The question of existence of solutions to the Lichnerowicz-York
conformally rescaled Einstein's constraint equations, for an
arbitrarily prescribed mean extrinsic curvature, has remained an open
problem for more than thirty years.  The rescaled equations,
which are a coupled nonlinear elliptic system consisting of the scalar
Hamiltonian constraint coupled to the vector momentum constraint, have
been studied almost exclusively in the setting of constant mean
extrinsic curvature, known as the CMC case.  In the CMC case the
equations decouple, and it has long been known how to establish
existence of solutions. The case of CMC data on closed (compact
without boundary) manifolds was completely resolved by several authors
over the last twenty years, with the last remaining sub-cases resolved
and all the CMC sub-cases on closed manifolds summarized by Isenberg
in~\cite{jI95}. Over the last ten years, other CMC cases on different
types of manifolds containing various kinds of matter fields were
studied and partially or completely resolved; see the survey~\cite{rBjI04}.
We take a moment to point out just some of the quite substantial number
of works in this area, including:
the original work on the Lichnerowicz equation
\cite{aL44};
the development of the conformal method
\cite{jY71,jY72,jY73,jY74};
the initial solution theory for the Hamiltonian constraint
\cite{nOMjY73,nOMjY74,nOMjY74a};
the thin sandwich alternative to the conformal method
\cite{rBgF93,cMkTjW70};
the complete classification of CMC initial data
\cite{jI95}
and the few known non-CMC results
\cite{jIvM96,jIjP97,yCBjIjY00};
various technical results on transverse-traceless tensors
and the conformal Killing operator
\cite{rB96,rBnOM96};
the more recent development of the conformal thin sandwich formulation
\cite{jY99};
initial data for black holes
\cite{rB00,jBjY80};
initial data for Kerr-like black holes
\cite{sD99,sD00b};
initial data with trapped surface boundaries
\cite{sD04,dM05b};
rough solution theory for CMC initial data
\cite{dM05,dM06,yCB04};
and the gluing approach to generating initial data
\cite{jC00}.
A survey of many of these results appears in~\cite{rBjI04}.

On the other hand, the question of existence of solutions to the Einstein
constraint equations for non-constant mean extrinsic curvature 
(the ``non-CMC case'') has remained largely unanswered, with progress
made only in the case that the mean extrinsic curvature is nearly
constant (the ``near-CMC case''), in the sense that the
size of its spatial derivatives is sufficiently small.
The near-CMC condition leaves the constraint equations
coupled, but ensures the coupling is weak.  In~\cite{jIvM96}, Isenberg
and Moncrief established the first existence (and uniqueness) result
in the near-CMC case, for background metric having negative Ricci
scalar.  Their result was based on a fixed-point argument, together
with the use of iteration barriers (sub- and super-solutions) 
which were shown to be bounded above and below by fixed positive constants,
independent of the iteration.
We note that
both the fixed-point argument and the global barrier construction
in~\cite{jIvM96} rely critically on the near-CMC assumption.  All
subsequent non-CMC existence results are based on the
framework in~\cite{jIvM96} and are thus limited to the near-CMC case
(see the survey~\cite{rBjI04}, the non-existence results
in~\cite{jInOM04}, and also the newer existence results
in~\cite{pAaCjI07} for non-negative Yamabe classes).

This article presents (together with the brief overview in~\cite{mHgNgT07a})
the first non-CMC existence results for the
Einstein constraints that do not require the near-CMC assumption.
Two recent advances make this possible: A new topological
fixed-point argument (established here and in~\cite{mHgNgT08b})
and a new global super-solution construction for the Hamiltonian 
constraint (established here and in~\cite{mHgNgT07a})
that are both free of near-CMC conditions.
These two results allow us to establish existence of
non-CMC solutions for conformal background metrics in the positive
Yamabe class, with the freely specifiable part of the data given by
the traceless-transverse part of the rescaled extrinsic curvature and
the matter fields sufficiently small, and with the matter energy
density not identically zero.
Our results here and in~\cite{mHgNgT08b,mHgNgT07a} can be viewed
as reducing the remaining open questions of existence of 
non-CMC (weak and strong) solutions without near-CMC conditions to 
two more basic and clearly stated open problems:
(1) Existence of near-CMC-free global {\em super}-solutions for the
Hamiltonian constraint equation when the background metric 
is in the non-positive Yamabe classes and for large data;
and
(2) existence of near-CMC-free global {\em sub}-solutions for the
Hamiltonian constraint equation when the background metric 
is in the positive Yamabe class in vacuum (without matter).
We will make some further comments about this later in the paper.

Our results in this article, which can be viewed as pushing forward the
rough solutions program that was initiated by Maxwell in~\cite{dM05,dM06}
(see also~\cite{yCB04}), further extend the known solution theory for the 
Einstein constraint equations on closed manifolds in several directions:
\begin{enumerate}[{\it(i)}]
\item {\bf\em Far-from-CMC Weak Solutions:}
      We establish the first existence results (Theorem~\ref{T:main1}) 
      for the coupled Einstein constraints in the non-CMC setting without 
      the near-CMC condition.
      In particular, if the rescaled background metric is in the 
      positive Yamabe class, if the freely specifiable part of the data 
      given by the traceless-transverse part of the rescaled extrinsic 
      curvature and the matter fields are sufficiently small, and if
      the energy density of matter is not identically zero, then we show
      existence of non-CMC solutions with mean extrinsic curvature
      arbitrarily far from constant.
      Two advances in the analysis of the Einstein constraint equations 
      make this result possible: A topological
      fixed-point argument (Theorems~\ref{T:FIXPT1} and~\ref{T:FIXPT2}) 
      based on compactness arguments rather than $k$-contractions 
      that is free of near-CMC conditions, and constructions of 
      global barriers for the Hamiltonian constraint that are
      similarly free of the near-CMC condition
      (Lemmas~\ref{L:HC-Sp}, \ref{L:HC-GSb},
              \ref{L:global-super}, \ref{L:global-sub},
              and~\ref{L:global-sub-Y-}).
\item {\bf\em Near-CMC Weak Solutions:}
      We establish existence results (Theorem~\ref{T:main2})
      for non-CMC solutions to the coupled constraints under the
      near-CMC condition in the setting of weaker (rougher) solutions spaces
      and for more general physical scenarios than appeared previously
      in~\cite{jIvM96,pAaCjI07}.
      In particular, we establish existence of
      weak solutions to the coupled Hamiltonian and momentum constraints on
      closed manifolds for all three Yamabe classes,
      with (positive) conformal factor in 
      $\phi \in W^{s,p}$
      where 
      $p \in (1,\infty)$
      and
      $s(p) \in (1+3/p,\infty)$.
      These results are based on combining barriers, 
      {\em a priori} estimates, and other results for the individual
      constraints together with a new type of topological fixed-point 
      argument (Theorems~\ref{T:FIXPT1} and~\ref{T:FIXPT2}),
      and are established in the presence of a weak background metric
      and data meeting very low regularity requirements.

\item {\bf\em CMC Weak Solutions:}
      In the CMC case, we establish existence (Theorem~\ref{T:main3}) 
      of weak solutions to the un-coupled Hamiltonian and momentum constraints 
      on closed manifolds for all three Yamabe classes,
      with (positive) conformal factor $\phi \in W^{s,p}$ where
      $p \in (1,\infty)$ and $s(p) \in (3/p, \infty) \cap [1,\infty)$.
      In the case of $s=2$, we reproduce the CMC existence results of 
      Choquet-Bruhat~\cite{yCB04}, and
      in the case $p=2$, we reproduce the CMC existence results of
      Maxwell~\cite{dM05}, but with a different proof;
      our CMC proof goes through the same analysis framework that we use to 
      obtain the non-CMC results (Theorems~\ref{T:FIXPT1} and~\ref{T:FIXPT2}).
      Again, these results established in the presence of a weak background 
      metric and with data meeting very low regularity requirements.

\item {\bf\em Barrier Constructions:}
      We give constructions
      (Lemmas~\ref{L:global-super} and~\ref{L:global-sub})
      of weak global sub- and super-solutions (barriers)
      for the Hamiltonian constraint equation which are free of the 
      near-CMC condition.  The constructions require the assumption that 
      the freely specifiable part of the data given by the 
      traceless-transverse part of the rescaled extrinsic curvature 
      and the matter fields are sufficiently small 
      (required for the super-solution construction
       in Lemma~\ref{L:global-super})
      and if the energy density of matter is not identically zero
      (required for the sub-solution
       in construction Lemma~\ref{L:global-sub}, although we note this 
      can be relaxed using the technique in~\cite{pAaCjI07}).
      While near-CMC-free sub-solutions are common in the literature,
      our near-CMC-free super-solution constructions
      appear to be the first such results of this type.

\item {\bf\em Supporting Technical Tools:}
      We assemble a number of new supporting technical results in the
      body of the paper and in several appendices, including: 
      topological fixed-point arguments designed for the Einstein
      constraints;
      construction and properties of general Sobolev classes $W^{s,p}$
      and elliptic operators on closed manifolds with weak metrics;
      the development of a very weak solution theory for the momentum 
      constraint; {\em a priori} $L^{\infty}$-estimates for 
      weak $W^{1,2}$-solutions to the Hamiltonian constraint; 
      Yamabe classification of non-smooth metrics in general 
      Sobolev classes $W^{s,p}$; and an analysis of the connection
      between conformal rescaling and the near-CMC condition.

\end{enumerate}

The results in this paper imply that the weakest differentiable solutions of
the Einstein constraint equations we have found correspond to CMC
and non-CMC hypersurfaces with physical spatial metric $h_{ab}$ satisfying
\begin{equation}
\textstyle
h_{ab} \in W^{s,p}(\cM),
\quad
\quad
p \in (1,\infty),
\quad
\quad
s(p) \in (1+\frac{3}{p},\infty).
\end{equation}
The curvature of such metrics can be computed in a distributional sense,
following~\cite{rGjT87}.
In the CMC case, the regularity can be reduced to
\begin{equation}
\textstyle
h_{ab} \in W^{s,p}(\cM),
\quad
\quad
p \in (1,\infty),
\quad
\quad
s(p) \in (\frac{3}{p}, \infty) \cap [1,\infty).
\end{equation}
In the case $s=2$, 
we reproduce the CMC existence results of Choquet-Bruhat~\cite{yCB04},
and in the case $p=2$,
we reproduce the CMC existence results of Maxwell~\cite{dM05},
but with a different proof;
our CMC proof goes through the same analysis framework that we use to 
obtain the non-CMC results (Theorems~\ref{T:FIXPT1} and~\ref{T:FIXPT2}).
In this paper we do not include uniqueness statements on CMC solutions,
or necessary and sufficient conditions for the existence of CMC solutions;
however, we expect that the techniques used in the above mentioned works 
can be adapted to this setting without difficulty.

There are several related motivations for
establishing the extensions outlined above.
First, as outlined in~\cite{rBjI04}, new results for the non-CMC case,
beyond the case analyzed in~\cite{jIvM96,pAaCjI07},
are of great interest in both mathematical and numerical relativity.
Non-CMC results that are free of the near-CMC assumption are of particular
interest, since the existence of solutions in this case has been an open
question for more than thirty years.
Second, there is currently substantial research activity in rough solutions
to the Einstein evolution equations, which rest on rough/weak solution
results for the initial data~\cite{sKiR01}.
Third, the approximation theory for Petrov-Galerkin-type methods
(including finite element, wavelet, spectral, and other methods) for the
constraints and similar systems previously developed in~\cite{mH01a}
establishes convergence of numerical solutions in very general physical
situations, but rests on assumptions about the solution theory; 
the results in the present paper and in~\cite{mHgNgT08b},
help to complete this approximation theory framework.
Similarly, very recent results on convergence of adaptive methods
for the constraints in~\cite{HoTs07a,HoTs07b} rest in large part
on the collection of results here and in~\cite{mH01a,mHgNgT08b}.

An extended outline of the paper is as follows.

In~\S\ref{sec:constraints},
we summarize the conformal decomposition of Einstein's constraint
equations introduced by Lichnerowicz and York, on a closed manifold.
We describe the classical strong formulation of the resulting coupled
elliptic system, and then define weak formulations of the constraint
equations that will allow us to develop solution theories for the
constraints in the spaces with the weakest possible regularity.

After setting up the basic notation, we give an overview of our 
main results in~\S\ref{sec:main}, summarized in three existence 
theorems (Theorems~\ref{T:main1}, \ref{T:main2}, and~\ref{T:main3}) for 
weak far-from-CMC, near-CMC, and CMC solutions to the coupled constraints,
extending the known solution theory in several distinct ways as described 
above.
We outline the two recent advances in the analysis of the
Einstein constraint equations that make these results possible.
The first advance is an abstract coupled topological fixed-point result
(Theorems~\ref{T:FIXPT1} and~\ref{T:FIXPT2}), 
the proof of which is based directly on 
compactness rather than on $k$-contractions.
This gives an analysis framework for weak solutions to the constraint
equations that is fundamentally free of the near-CMC assumption;
the near-CMC assumption then only potentially arises in the
construction of global barriers as part of the overall fixed-point argument.
A result of this type also makes possible the new non-CMC results for 
the case of compact manifolds with boundary appearing in~\cite{mHgNgT08b}.
The second new advance is the construction of global super-solutions
for the Hamiltonian constraint that are also free of the near-CMC condition;
we give an overview of the main ideas in the constructions, which are
then derived rigorously in~\S\ref{sec:barriers}.

In~\S\ref{sec:individual} we then develop the necessary results
for the individual constraint equations in order to
complete an existence argument for the coupled system
based on the abstract fixed-point argument in 
Theorems~\ref{T:FIXPT1} and~\ref{T:FIXPT2}.
In particular, in~\S\ref{sec:momentum},
we first develop some basic technical results for the momentum constraint
operator under weak assumptions on the problem data, including
existence of weak solutions to the momentum constraint, given the
conformal factor as data.
In~\S\ref{sec:Hamiltonian}, we assume the existence of 
barriers (weak sub- and super-solutions) to the Hamiltonian constraint 
equation forming a nonempty positive bounded interval,
and then derive several properties of the Hamiltonian constraint that are 
needed in the analysis of the coupled system.
The results are established under weak assumptions on the problem data,
and for any Yamabe class.

Using order relations on appropriate Banach spaces, we then derive
several such compatible weak global sub- and super-solutions
in~\S\ref{sec:barriers},
based both on constants and on more complex non-constant constructions.
While the sub-solutions are similar to those found previously in
the literature, some of the super-solutions are new.
In particular, we give two super-solution constructions that do
not require the near-CMC condition.
The first is constant, and requires that the scalar curvature be
strictly globally positive.
The second is based on a scaled solution to a Yamabe-type problem,
and is valid for any background metric in the positive Yamabe class.

In~\S\ref{sec:proof}, we establish the main results by giving the
proofs of Theorems~\ref{T:main1}, \ref{T:main2}, and~\ref{T:main3}.
In particular, using the topological fixed-point argument in
Theorem~\ref{T:FIXPT2}, we combine the global barrier constructions in
\S\ref{sec:barriers} with the individual constraint results in 
\S\ref{sec:individual} to establish existence of weak non-CMC solutions.
We summarize our results in \S\ref{sec:summary}.
For ease of exposition, various supporting technical results are 
given in several appendices as follows:
Appendix~\S\ref{sec:fixpt} -- topological fixed-point arguments;
Appendix~\S\ref{sec:OBS} -- ordered Banach spaces;
Appendix~\S\ref{sec:monotone} -- monotone increasing maps;
Appendix~\S\ref{sec:Sobolev} -- construction of fractional order Sobolev spaces of sections of vector bundles over closed manifolds;
Appendix~\S\ref{sec:killing} -- {\em a priori} estimates for elliptic operators;
Appendix~\S\ref{sec:maxprinciple} -- maximum principles on closed manifolds;
Appendix~\S\ref{sec:yamabe} -- Yamabe classification of weak metrics;
Appendix~\S\ref{sec:conf-inv} -- conformal covariance of the Hamiltonian constraint;
and
Appendix~\S\ref{sec:rescaling} -- conformal rescaling and the near-CMC condition.

\section{Preliminary material}
   \label{sec:constraints}

\subsection{Notation and conventions}
   \label{subsec:notation}
Let $\cM$ be an $n$-dimensional smooth closed manifold.
We denote by $\pi : E \to \cM$ (or simply $E \to \cM$, or just $E$) 
a smooth vector bundle over $\cM$,
where the manifold $\cM$ is called the base space,
$E$ is called the total space,
and $\pi$ is the bundle projection such that
for any $x \in \cM$, $E_x = \pi^{-1}(x)$ is the fiber over $x$,
which is a vector space of (fiber) dimension $m_x$.
If all fibers $E_x$ have dimension $m_x=m$, we say the fiber
dimension of $E$ is $m$.
The manifold $\cM$ itself can be considered as the
vector bundle $E = \cM \times \{0\}$ with fiber dimension $m=0$.
A section of the trivial vector bundle $E = \cM \times \mathbb{R}$ with fiber
dimension $m=1$ is simply a scalar function on $\cM$.
Our primary interest is the case where 
$$
E= T^r_{s}\cM = \underbrace{T\cM\otimes\ldots\otimes T\cM}_{r\textrm{ times}}\otimes\underbrace{T^*\cM\otimes\ldots\otimes T^*\cM}_{s\textrm{ times}},
$$
the $(r,s)$-tensor bundle with contravariant 
order $r$ and covariant order $s$, giving fiber dimension $m=n(r+s)$,
where $T\cM$ is the tangent bundle,
and $T^*\cM$ is the co-tangent bundle of $\cM$.
A $C^k$ section of $\pi$ (or of $E$) is a $C^k$ map
$\gamma : \cM \to E$ such that for each $x \in \cM$,
$\pi(\gamma(x)) = x$.
These $C^k$ sections form real Banach spaces $C^k(E)$ which
arise naturally in the global linear analysis of
partial differential equations on manifolds.

Let $h_{ab}\in C^{\infty}(T^{0}_{2}\cM)$ be a smooth Riemannian metric on $\cM$, 
(where by convention Latin indices denote abstract indices 
as e.g. in~\cite{Wald84}), meaning that it is
a symmetric, positive definite, covariant, smooth two-index 
tensor field on $\cM$.
The combination $(\cM,h_{ab})$ is referred to as a (smooth)
Riemannian manifold; we will relax the smoothness requirement 
on $h_{ab}$ below.
For each $x \in \cM$, the metric $h_{ab}(x)$ defines a positive definite inner product 
on the tangent space $T_x\cM$ at $x$. Denote by $h^{ab}$
the inverse of $h_{ab}$, that is, $h_{ac}h^{bc}
=\delta_a{}^b$, where $\delta_a{}^b: T_x\cM\to T_x\cM$ is the
identity map. We use the convention that repeated indices, one
upper-index and one sub-index, denote contraction. 
Indices on tensors will be raised and lowered with $h^{ab}$ and
$h_{ab}$, respectively. For example, given the tensor $u^{ab}{}_c$ we
denote $u_{abc}=h_{aa_1}h_{bb_1}\,u^{a_1b_1}{}_{c}$, and $u^{abc}
=h^{cc_1}\,u^{ab}{}_{c_1}$; notice that the order of the indices is
important in the case that the tensor $u_{abc}$ or $u^{abc}$ is not
symmetric. We say that a tensor is of type $m$ iff it can be transformed 
into a tensor $u_{a_1\cdots a_m}$ by lowering appropriate indices
(its vector bundle then has fiber dimension $mn$).

We now give a brief overview of $L^p$ and Sobolev spaces of sections
of vector bundles over closed manifolds in order to introduce the 
notation used throughout the paper.
An overview of the construction of fractional order Sobolev spaces
of sections of vector bundles can be found in Appendix~\ref{sec:Sobolev}, 
based on Besov spaces and partitions of unity.
The case of the sections of the trivial bundle of scalars can 
also be found in~\cite{Hebey96}, and the case of tensors 
can also be found in~\cite{Palais65}.
Let $\nabla_a$ be
the Levi-Civita connection associated with the metric $h_{ab}$, that
is, the unique torsion-free connection satisfying
$\nabla_ah_{bc}=0$. Let $R_{abc}{}^d$ be the Riemann tensor of the
connection $\nabla_a$, where the sign convention used in this
article is $(\nabla_a\nabla_b -\nabla_b\nabla_a)v_c = R_{abc}{}^d
v_d$. Denote by $R_{ab} := R_{acb}{}^c$ the Ricci tensor and by $R
:=R_{ab}h^{ab}$ the Ricci scalar curvature of this connection.

Integration on $\cM$ can be defined with the volume form
associated with the metric $h_{ab}$.
Given an arbitrary tensor $u^{a_1\cdots
a_r}{}_{b_1\cdots b_s}$ of type $m=r+s$, we define a
real-valued function measuring its magnitude at any point $x \in \cM$ as
\begin{equation}
\label{tensor-magnitude}
|u| := (u^{a_1\cdots b_s}u_{a_1\cdots b_s})^{1/2}.
\end{equation}
A norm of an arbitrary tensor field 
$u^{a_1\cdots a_r}{}_{b_1\cdots b_s}$ on $\cM$ can then be
defined for any $1\leqs p < \infty$ and for $p=\infty$ respectively
using~(\ref{tensor-magnitude}) as follows,
\begin{equation}
\label{N-Lp-norm}
\|u\|_p := \left(\int_{\cM} |u|^p\,dx\right)^{1/p},\qquad
\|u\|_{\infty}:= \mbox{ess}\,\sup_{x\in\cM}|u|.
\end{equation}
One way to construct the {\bf Lebesgue spaces} $L^p(T^{r}_{s}\cM)$ of
sections of the $(r,s)$-tensor bundle, for $1\leqs
p\leqs\infty$, is through the completion
of $C^{\infty}(T^{r}_{s}\cM)$ with respect to the {\bf $L^p$-norm} \eqref{N-Lp-norm}.
The $L^p$ spaces are Banach spaces, and the case $p=2$
is a Hilbert space with the inner product and norm given by
\begin{equation}
\label{N-L2-inner} (u,v):= \int_{\cM} u_{a_1\cdots a_m}v^{a_1\cdots
a_m}\, dx,\qquad \|u\|:=\sqrt{(u,u)}=\|u\|_2.
\end{equation}
Denote covariant derivatives of tensor fields as
$\nabla^{k}u^{a_1\cdots a_m} 
:=\nabla_{b_1}\cdots\nabla_{b_k}u^{a_1\cdots a_m}$,
where $k$ denotes the total number of derivatives
represented by the tensor indices $(b_1,\ldots,b_k)$.
Another norm on $C^{\infty}(T^{r}_{s}\cM)$ is given for
any non-negative integer $k$ and for any $1\leqs p \leqs \infty$ as follows,
\begin{equation}
\label{N-Wkp-norm}
\|u\|_{k,p} := \sum_{l=0}^k \,\|\nabla^{l}u\|_p.
\end{equation}
The {\bf Sobolev spaces} $W^{k,p}(T^r_{s}\cM)$ of sections
of the $(r,s)$-tensor bundle can be defined as
the completion of $C^{\infty}(T^{r}_{s}\cM)$ with respect to the 
$W^{k,p}$-norm \eqref{N-Wkp-norm}.
The Sobolev spaces $W^{k,p}$ are Banach
spaces, and the case $p=2$ is a
Hilbert space.
We have $L^p=W^{0,p}$ and $\|s\|_p =\|s\|_{0,p}$.
See Appendix~\ref{sec:Sobolev} for a more careful construction
that includes real order Sobolev spaces of sections of vector bundles.

Let $C^\infty_{+}$ be the set of nonnegative smooth (scalar) functions on $\cM$.
Then we can define order cone
\begin{equation}
\label{E:wkp-cone}
W^{s,p}_{+} := \bigl\{ \phi \in W^{s,p} :
\langle\phi,\varphi\rangle\geqs 0 \quad \forall \, \varphi\in C^{\infty}_{+}
\, \bigr \},
\end{equation}
with respect to which the Sobolev spaces $W^{s,p}=W^{s,p}(\cM)$ are ordered Banach spaces.
Here $\langle\cdot,\cdot\rangle$ is the unique extension of $L^2$-inner product to a bilinear form $W^{s,p}\otimes W^{-s,p'}\to\R$, with $\frac1{p'}+\frac1p=1$.
The order relation is then $\phi\geqs\psi$ iff $\phi-\psi\in
W^{s,p}_{+}$.
We note that this order cone is normal only for $s=0$.
See Appendix~\ref{sec:OBS}, where we review the main properties of
ordered Banach spaces.

\subsection{The Einstein constraint equations}

We give a quick overview of the Einstein constraint equations in
general relativity, and then define weak formulations that
are fundamental to both solution theory and the development
of approximation theory.
Analogous material for the case of compact manifolds with
boundary can be found in~\cite{mHgNgT08b}.

Let $(M,g_{\mu\nu})$ be a 4-dimensional spacetime, that is, $M$ is a
4-dimensional, smooth manifold, and $g_{\mu\nu}$ is a smooth,
Lorentzian metric on $M$ with signature $(-,+,+,+)$. Let
$\nabla_{\mu}$ be the Levi-Civita connection associated with the
metric $g_{\mu\nu}$.  The Einstein equation is
\[
G_{\mu\nu} = \kappa T_{\mu\nu},
\]
where $G_{\mu\nu} = R_{\mu\nu} - \frac{1}{2}R\,g_{\mu\nu}$ is the
Einstein tensor, $T_{\mu\nu}$ is the stress-energy tensor, and
$\kappa = 8\pi G/c^4$, with $G$ the gravitation constant and $c$ the
speed of light. The Ricci tensor is $R_{\mu\nu} =
R_{\mu\sigma\nu}{}^{\sigma}$ and $R= R_{\mu\nu}g^{\mu\nu}$ is the
Ricci scalar, where $g^{\mu\nu}$ is the inverse of $g_{\mu\nu}$,
that is $g_{\mu\sigma} g^{\sigma\nu} =\delta_{\mu}{}^{\nu}$. The
Riemann tensor is defined by
$R_{\mu\nu\sigma}{}^{\rho} w_{\rho} =\big(\nabla_{\mu}\nabla_{\nu}
-\nabla_{\nu}\nabla_{\mu}\bigr) w_{\sigma}$, where $w_{\mu}$ is any
1-form on $M$. The stress energy tensor $T_{\mu\nu}$ is assumed to
be symmetric and to satisfy the condition
$\nabla_{\mu}T^{\mu\nu} = 0$ and the {\bf dominant energy
condition}, that is, the vector $-T^{\mu\nu}v_{\nu}$ is timelike and
future-directed, where $v^{\mu}$ is any timelike and future-directed
vector field. In this section Greek indices $\mu$, $\nu$, $\sigma$,
$\rho$ denote abstract spacetime indices, that is, tensorial
character on the 4-dimensional manifold $M$. They are raised and
lowered with $g^{\mu\nu}$ and $g_{\mu\nu}$, respectively. 
Latin indices $a$, $b$, $c$, $d$ will denote tensorial
character on a 3-dimensional manifold.

The map $t:M\to \R$ is a {\bf time function} iff the function $t$ is
differentiable and the vector field $-\nabla^{\mu}t$ is a timelike,
future-directed vector field on $M$. Introduce the hypersurface $\cM
:=\{ x\in M : t(x)=0\}$, and denote by $n_{\mu}$ the unit 1-form
orthogonal to $\cM$. By definition of $\cM$ the form $n_{\mu}$ can
be expressed as $n_{\mu} = -\alpha\,\nabla_{\mu}t$, where $\alpha$,
called the lapse function, is the positive function such that
$n_{\mu} n_{\nu}\, g^{\mu\nu} = -1$. Let $\hh_{\mu\nu}$ and
$\hk_{\mu\nu}$ be the first and second fundamental forms of $\cM$,
that is,
\[
\hh_{\mu\nu} := g_{\mu\nu}- n_{\mu} n_{\nu},\qquad
\hk_{\mu\nu} := -\hh_{\mu}{}^{\sigma} \nabla_{\sigma} n_{\nu}.
\]
The Einstein constraint equations on $\cM$ are given by
\[
\bigl( G_{\mu\nu} -\kappa T_{\mu\nu}\bigr) \, n^{\nu} =0.
\]
A well known calculation allows us to express these equations
involving tensors on $M$ as equations involving {\em intrinsic}
tensors on $\cM$. The result is the following equations,
\begin{align}
\label{CE-def-H}
\hRi + \hk^2 - \hk_{ab}\hk^{ab} - 2\kappa \hrho &=0,\\
\label{CE-def-M}
\hD^a\hk - \hD_b\hk^{ab} + \kappa \hj^a &= 0,
\end{align}
where tensors $\hh_{ab}$, $\hk_{ab}$, $\hj_a$ and $\hrho$ on a
3-dimensional manifold are the pull-backs on $\cM$ of the tensors
$\hh_{\mu\nu}$, $\hk_{\mu\nu}$, $\hj_{\mu}$ and $\hrho$ on the
4-dimensional manifold $M$. We have introduced the energy density
$\hrho := n_{\mu} n_{\mu} T^{\mu\nu}$ and the momentum current
density $\hj_{\mu} := -\hh_{\mu\nu} n_{\sigma} T^{\nu\sigma}$. We
have denoted by $\hD_{a}$ the Levi-Civita connection associated to
$\hh_{ab}$, so $(\cM,\hh_{ab})$ is a 3-dimensional Riemannian
manifold, with $\hh_{ab}$ having signature $(+,+,+)$, and we use the
notation $\hh^{ab}$ for the inverse of the metric $\hh_{ab}$.
Indices have been raised and lowered with $\hh^{ab}$ and $\hh_{ab}$,
respectively. We have also denoted by $\hRi$ the Ricci scalar curvature of the metric $\hh_{ab}$. Finally, recall that the
constraint Eqs.~(\ref{CE-def-H})-(\ref{CE-def-M}) are indeed
equations on $\hh_{ab}$ and $\hk_{ab}$ due to the matter fields
satisfying the energy condition $-\hrho^2 +\hj_a\hj^a \leqs 0$
(with strict inequality holding at points on $\cM$ where $\hrho \neq 0$;
see~\cite{Wald84}), which is
implied by the dominant energy condition on the stress-energy tensor
$T^{\mu\nu}$ in spacetime.

\subsection{Conformal transverse traceless decomposition}
Let $\phi$ denote a positive scalar field on $\cM$, and decompose the
extrinsic curvature tensor $\hk_{ab} = \hl_{ab} + \frac{1}{3}\hh_{ab} \htau$,
where $\htau := \hk_{ab}\hh^{ab}$ is the trace and then $\hl_{ab}$ is
the traceless part of the extrinsic curvature tensor. Then, introduce
the following conformal re-scaling:
\begin{equation}
\label{CE-def-mf}
\begin{aligned}
\hh_{ab} &=: \phi^4 \, h_{ab},&
\hl^{ab} &=: \phi^{-10} \,l^{ab},&
\htau &=: \tau,\\
\hj^a &=: \phi^{-10}\; j^a,&
\hrho &=: \phi^{-8}\, \rho.
\end{aligned}
\end{equation}
We have introduced the Riemannian metric $h_{ab}$ on the 3-dimensional
manifold $\cM$, which determines the Levi-Civita connection $D_a$, and
so we have that $D_a h_{bc}=0$. We have also introduced the symmetric,
traceless tensor $l_{ab}$, and the non-physical matter sources $j^a$
and $\rho$. The different powers of the conformal re-scaling above are
carefully chosen so that the constraint
Eqs.~(\ref{CE-def-H})-(\ref{CE-def-M}) transform into the following
equations
\begin{gather}
\label{CE-cr1H}\textstyle
-8 \Delta \phi + \Ri \phi + \frac{2}{3}\tau^2 \phi^5
- l_{ab}l^{ab} \phi^{-7} -2\kappa \rho\phi^{-3} =0,\\
\label{CE-cr1M}\textstyle
-D_bl^{ab} + \frac{2}{3} \phi^6 D^a \tau +\kappa j^a =0,
\end{gather}
where in equation above, and from now on, indices of unhatted fields
are raised and lowered with $h^{ab}$ and $h_{ab}$ respectively. We
have also introduced the {\bf Laplace-Beltrami operator}
with respect to the metric $h_{ab}$, acting on smooth scalar fields;
it is defined as follows
\begin{equation}
\label{E:laplace-beltrami}
\Delta \phi:= h^{ab}D_aD_b\phi.
\end{equation}
Eqs.~(\ref{CE-cr1H})-(\ref{CE-cr1M}) can be obtained by a
straightforward albeit long computation. In order to perform this
calculation it is useful to recall that both $\hD_a$ and $D_a$ are
connections on the manifold $\cM$, and so they differ on a tensor
field $C_{ab}{}^c$, which can be computed explicitly in terms of
$\phi$, and has the form
\[
C_{ab}{}^c = 4 \delta_{(a}{}^cD_{b)} \ln(\phi)
- 2 h_{ab}h^{cd}D_d \ln(\phi).
\]
We remark that the power four on the re-scaling of the metric
$\hh_{ab}$ and $\cM$ being 3-dimensional imply that $\hRi =\phi^{-5}
(\Ri\phi - 8\Delta\phi)$, or in other words, that $\phi$ satisfies
the {\bf Yamabe-type problem}:
\begin{equation}
  \label{E:yamabe}
-8\Delta \phi+\Ri \phi - \hRi \phi^5 = 0, \quad \phi > 0,
\end{equation}
where $\hRi$ represents the scalar curvature corresponding to the physical metric $\hh_{ab} = \phi^4 h_{ab}$.
Note that for any other power in the re-scaling, terms proportional 
to $h^{ab}(D_a\phi)(D_b\phi)/\phi^2$ appear in the transformation.
The set of all metrics on a closed manifold can be classified into the
three disjoint Yamabe classes $\cY^{+}(\cM)$, $\cY^{0}(\cM)$, and
$\cY^{-}(\cM)$, corresponding to whether one can conformally transform the metric into a metric with strictly positive,
zero, or strictly negative scalar curvature, respectively, cf.~\cite{jLtP87} (See also Appendix \ref{sec:yamabe}).
We note that the {\bf Yamabe problem} is to determine,
for a given metric $h_{ab}$, whether there exists a conformal
transformation $\phi$ solving~(\ref{E:yamabe}) such that $\hRi = \mathrm{const}$.
Arguments similar to those above for $\phi$ force the power negative ten 
on the re-scaling of the tensor $\hl^{ab}$ and $\hj^a$, so terms proportional
to $(D_a\phi)/\phi$ cancel out in \eqref{CE-cr1M}. Finally, the
ratio between the conformal re-scaling powers of $\hrho$ and $\hj^a$
is chosen such that the inequality 
$-\rho^2 + h_{ab} j^aj^b \leqs 0$
implies the inequality 
$-\hrho^2 + \hh_{ab}\hj^a\hj^b \leqs 0$.
For a complete discussion of all possible choices of re-scaling
powers, see Appendix~\ref{sec:rescaling}.

There is one more step to convert the original constraint equation
\eqref{CE-def-H}-\eqref{CE-def-M} into a determined elliptic
system of equations. This step is the following: Decompose the
symmetric, traceless tensor $l_{ab}$ into a divergence-free part
$\sigma_{ab}$, and the symmetrized and traceless gradient of a vector,
that is, $l^{ab} =: \sigma^{ab} + (\cL w)^{ab}$, where
$D_a\sigma^{ab}=0$ and we have introduced the {\bf conformal Killing
operator} $\cL$ acting on smooth vector fields and defined as follows
\begin{equation}
\label{CF-def-CK}\textstyle
(\cL w)^{ab} := D^a w^b + D^b w^a - \frac{2}{3}(D_c w^c) h^{ab}.
\end{equation}
Therefore, the constraint Eqs.~(\ref{CE-def-H})-(\ref{CE-def-M}) are
transformed by the conformal re-scaling into the following equations
\begin{gather}
\label{CE-cr2H}\textstyle
\hspace*{-0.15cm}
- 8 \Delta \phi + \Ri \phi 
+ \frac{2}{3}\tau^2 \phi^5
- [\sigma_{ab}+(\cL w)_{ab}] [\sigma^{ab}+(\cL w)^{ab}]\phi^{-7} 
- 2\kappa \rho \phi^{-3} =0,\\
\label{CE-cr2M}\textstyle
-D_b(\cL w)^{ab} + \frac{2}{3} \phi^6 D^a \tau +\kappa j^a =0.
\end{gather}
In the next section we interpret these equations above as partial
differential equations for the scalar field $\phi$ and the vector
field $w^a$, while the rest of the fields are considered as given
fields. Given a solution $\phi$ and $w^a$ of
Eqs.~(\ref{CE-cr2H})-(\ref{CE-cr2M}), the physical metric $\hh_{ab}$
and extrinsic curvature $\hk^{ab}$ of the hypersurface $\cM$ are given
by
\[\textstyle
\hh_{ab} = \phi^4 h_{ab},\qquad
\hk^{ab} = \phi^{-10}
[\sigma^{ab} + (\cL w)^{ab}] + \frac{1}{3}\, 
\phi^{-4} \tau h^{ab},
\]
while the matter fields are given by Eq~(\ref{CE-def-mf}).

From this point forward, for simplicity we will denote the Levi-Civita connection
of the metric $h_{ab}$ on the 3-dimensional manifold $\cM$
as $\nabla_a$ rather than $D_a$, and the Ricci scalar of
$h_{ab}$ will be denoted by $R$ instead of $\Ri$.
Let $(\cM, h)$ be a 3-dimensional Riemannian manifold, where $\cM$ is
a smooth, compact manifold without boundary, and
$h\in C^{\infty}(T^0_2\cM)$ is a positive definite metric. 
With the shorthands $C^{\infty}=C^\infty(\cM\times\R)$ and $\biC^{\infty}=C^\infty(T\cM)$,
let
$L: C^{\infty}\to C^{\infty}$ and $\IL :\biC^{\infty}\to\biC^{\infty}$
be the operators with
actions on $\phi\in C^{\infty}$ and $\biw\in\biC^{\infty}$ given by
\begin{align}
\label{CF-def-L}
L\phi &:= -\Delta \phi,\\
\label{CF-def-IL}
(\IL \biw)^a &:= -\nabla_b (\cL \biw)^{ab},
\end{align}
where $\Delta$ denotes the Laplace-Beltrami operator 
defined in~\eqref{E:laplace-beltrami}, 
and where $\cL$ denotes the conformal Killing operator 
defined in~\eqref{CF-def-CK}.
We will also use the index-free notation $\IL\biw$ and $\cL\biw$.

The freely specifiable functions of the problem are a scalar function
$\tau$, interpreted as the trace of the physical extrinsic curvature;
a symmetric, traceless, and divergence-free, contravariant, two index 
tensor $\sigma$; the non-physical energy density $\rho$ and the non-physical
momentum current density vector $\bij$ subject to the requirement
$-\rho^2 +\bij\cdot\bij \leqs 0$.
The term non-physical refers here to
a conformal rescaled field, while physical refers to a conformally
non-rescaled term. The requirement on $\rho$ and $\bij$ mentioned
above and the particular conformal rescaling used in the
semi-decoupled decomposition imply that the same inequality is
satisfied by the physical energy and momentum current densities. This
is a necessary condition (although not sufficient) in order that the
matter sources in spacetime satisfy the dominant energy
condition. The definition of various energy conditions can be found
in~\cite[page 219]{Wald84}. Introduce the non-linear operators $F:
C^{\infty}\times\biC^{\infty}\to C^{\infty}$ and $\IF
:C^{\infty}\to\biC^{\infty}$ given by
\begin{equation*}
F(\phi,\biw) = a_{\tau} \phi^5 + a_{\tiR} \phi - a_{\rho} \phi^{-3}
- a_{w}\phi^{-7},
\quad\textrm{and}\quad
\IF (\phi) = \bib_{\tau} \, \phi^6 + \bib_{j},
\end{equation*}
where the coefficient functions are defined as follows
\begin{equation}
\label{CF-def-coeff2}
\begin{aligned}
a_{\tau} &:= \textstyle\frac{1}{12}\tau^2,&
a_{\tiR} &:= \textstyle\frac{1}{8}R,&
a_{\rho} &:= \textstyle\frac{\kappa}{4} \rho,\\
a_{\biw} &:= \textstyle\textstyle\frac{1}{8}(\sigma +\cL\biw)_{ab}(\sigma + \cL \biw)^{ab},&
b_{\tau}^a &:= \textstyle\frac{2}{3}\nabla^a \tau,&
b_{j}^a &:= \kappa j^a.
\end{aligned}
\end{equation}
Notice that the scalar coefficients $a_{\tau}$, $a_{w}$, and
$a_{\rho}$ are non-negative, while there is no sign restriction on
$a_{\tiR}$.

With these notations, the {\bf classical formulation} (or the strong formulation) of the
coupled Einstein constraint equations reads as: Given
the freely specifiable smooth functions $\tau$, $\sigma$, $\rho$, and
$\bij$ in $\cM$, find a scalar field $\phi$ and a vector field
$\biw$ in $\cM$ solution of the system
\begin{equation}
\label{CF-LY}
L \phi + F(\phi,\biw) = 0
\qquad\textrm{and}\qquad
\IL \biw  + \IF(\phi) =0
\qquad\textrm{in }\cM.
\end{equation}

\subsection{Formulation in Sobolev spaces}
\label{sub:weak}

We now outline a formulation of the
Einstein constraint equations that involves the weakest
regularity of the equation coefficients such that the equation
itself is well-defined. 
So in particular, the operators $L$ and $\IL$ are no longer differential operators sending smooth sections to smooth sections.
We shall employ Sobolev spaces to quantify smoothness, cf. Appendix \ref{sec:Sobolev}.

Let $(\cM, h)$ be a 3-dimensional Riemannian
manifold, where $\cM$ is a smooth, compact manifold without
boundary, and with $p\in(\frac32,\infty)$ and $s\in(\frac3p,\infty)\cap[1,2]$, $h\in W^{s,p}(T^0_2\cM)$ is a positive definite metric.
Note that the restriction $s\leqs2$ is only apparent, since $W^{t,p}\hookrightarrow W^{2,p}$ for any $t>2$.
In the formulation of the constraint equations we need to distinguish the cases $s>2$ and $s\leqs2$ at least notation-wise, and we choose to present in this subsection the case $s\leqs2$ because this is the case that is considered in the core existence theory; the higher regularity is obtained by a standard bootstrapping technique.
The general case is discussed in Sections \ref{sec:individual} and \ref{sec:proof}.
Let us define $r=r(s,p)=\frac{3p}{3+(2-s)p}$, so that the continuous embedding
$L^{r}\hookrightarrow W^{s-2,p}$ holds.
Introduce the operators
\begin{align*}
A_{\tiL}:W^{s,p}\to W^{s-2,p},\qquad\textrm{and}\qquad
A_{\tiIL}:\biW^{1,2r}\to \biW^{-1,2r},
\end{align*}
as the unique extensions of the operators $L$ and $\IL$
in Eqs.~(\ref{CF-def-L}) and~(\ref{CF-def-IL}), respectively, 
cf. Lemma \ref{l:bdd-operator} in Appendix \ref{sec:killing}.
The boldface letters denote spaces of sections of the tangent bundle $T\cM$, e.g., $\biW^{1,2r}=W^{1,2r}(T\cM)$.

Fix the source functions
\begin{equation}
\label{WF-coeff}
\tau \in L^{2r},\quad
\rho \in W^{s-2,p}_{+},\quad
\sigma \in L^{2r},\quad
\bij \in \biW^{-1,2r},
\end{equation}
where $\sigma$ is symmetric, traceless and divergence-free in weak
sense, the latter meaning that $\langle\sigma,\cL\bomega\rangle=0$ for all
$\bomega\in\biW^{1,(2r)'}$.
Here $\frac1{(2r)'}+\frac1{2r}=1$, 
and $\langle\cdot,\cdot\rangle$ denotes the extension of the $L^2$-inner product to $\biW^{-1,2r}\otimes\biW^{1,(2r)'}$.
We say that the matter fields $\rho$ and $\bij$ satisfy the {\bf energy condition}
iff there exist sequences $\{\rho_n\}\subset C^\infty$ and
$\{\bij_n\}\subset\biC^\infty$, respectively converging to $\rho$ and $\bij$ in the appropriate topology, such that
\[
\rho_n^2 - \bij_n\cdot\bij_n\geqs0.
\]
Given any function $\tau\in L^{2r}$ we have $\bib_{\tau}\equiv\frac23\nabla\tau\in\biW^{-1,2r}$.
The assumptions $\tau\in L^{2r}$ and $\sigma\in L^{2r}$ imply that for every
$\biw\in\biW^{1,2r}$ the functions $a_{\tau}$ and $a_{\biw}$ belong
to $L^{r}$. For example, to see that $a_{\biw}\in L^{r}$, we proceed
as
$$
\|a_{\biw}\|_{r}
=\|\sigma+\cL\biw\|_{2r}
\leqs2\left(\|\sigma\|_{2r}^2+\|\cL\biw\|_{2r}^2\right)
\leqs2\left(\|\sigma\|_{2r}^2+c_{\cL}^2\|\biw\|_{1,2r}^2\right),
$$
where we used the boundedness $\|\cL\biw\|_{2r}\leqs c_{\cL}\|\biw\|_{1,2r}$.
The assumption on the background metric implies that $a_{\tiR}\in W^{s-2,p}$. 

Given any two functions $u,v \in L^{\infty}$, and $t\geqs0$ and $q\in[1,\infty]$, define the interval
\[
[u,v]_{t,q} := \{ \phi \in W^{t,q} : u \leqs \phi \leqs v \} \subset W^{t,q},
\]
see Lemma~\ref{L:wsp-interval} on page \pageref{L:wsp-interval}.
We equip $[u,v]_{t,q}$ with the subspace topology of $W^{t,q}$.
We will write $[u,v]_q$ for $[u,v]_{0,q}$, and $[u,v]$ for $[u,v]_{\infty}$.
Now, assuming that $\phi_{-},\phi_{+}\in W^{{s},p}$ and $0<\phi_{-}\leqs\phi_{+}<\infty$, we introduce the non-linear operators
\begin{gather*}
f : [\phi_{-},\phi_{+}]_{s,p} \times \biW^{1,2r}
\to W^{s-2,p},
\qquad\textrm{and}\qquad
\bif : [\phi_{-},\phi_{+}]_{s,p} \to \biW^{-1,2r},
\end{gather*}
by
\begin{equation*}
f(\phi,\biw)
=  a_{\tau}\phi^5 + a_{\tiR}\phi
- a_{\rho}\phi^{-3} - a_{\biw}\phi^{-7},
\quad\textrm{and}\quad
\bif(\phi)
=  \bib_{\tau}\phi^6 + \bib_{j},
\end{equation*}
where the pointwise multiplication by an element of $W^{{s},p}$ defines a bounded linear map in
$W^{s-2,p}$ and in $\biW^{-1,2r}$, cf. Corollary \ref{c:alg}(a) in Appendix \ref{sec:Sobolev}.

Now, we can formulate the Einstein constraint equations in terms of the above defined operators: Find elements $\phi\in [\phi_{-},\phi_{+}]_{s,p}$ and $\biw\in\biW^{1,2r}$ solutions of
\begin{align}
\label{WF-LYs1}
A_{\tiL}\phi + f(\phi,\biw) &= 0,\\
\label{WF-LYm1}
A_{\tiIL}\biw + \bif(\phi) &= 0.
\end{align}

In the following, often we treat the two equations separately.
The {\bf Hamiltonian constraint equation} is the following: Given a function $\biw\in\biW^{1,2r}$, find an element $\phi\in [\phi_{-},\phi_{+}]_{s,p}$ solution of
\begin{equation}
\label{WF-LYs2}
A_{\tiL}\phi + f(\phi,\biw) = 0.
\end{equation}
When the Hamiltonian constraint equation is under consideration, the
function $\biw$ is referred to as the {\em source}. To indicate the
dependence of the solution $\phi$ on the source $\biw$, sometimes we
write $\phi=\phi_{\biw}$. Let us define the {\bf momentum constraint
equation}: Given $\phi\in W^{s,p}$ with $\phi>0$, find an element
$\biw\in\biW^{1,2r}$ solution of
\begin{equation}
\label{WF-LYm2}
A_{\tiIL}\biw + \bif(\phi) = 0.
\end{equation}
When the momentum constraint equation is under consideration, the
function $\phi$ is referred to as the {\em source}. To indicate the
dependence of the solution $\biw$ on the source $\phi$, sometimes we
write $\biw=\biw_\phi$.

\section{Overview of the main results}
   \label{sec:main}

In this section, we state our three main theorems 
(Theorems~\ref{T:main1}, \ref{T:main2}, and~\ref{T:main3} below) 
on the existence of far-from-CMC, near-CMC, and CMC
solutions to the Einstein constraint equations, and give an outline
of the overall structure of the argument that we build in the paper.
The proofs of the main results appear in~\S\ref{sec:proof} toward
the end of the paper, after we develop a number of supporting
results in the body of the paper. After we give an overview of the
basic abstract structure of the coupled nonlinear constraint problem,
we prove two abstract topological
fixed-point theorems (Theorems~\ref{T:FIXPT1} and~\ref{T:FIXPT2})
that are the basis for our analysis of the coupled system; 
these arguments are also the basis for our results in~\cite{mHgNgT08b} 
on existence of non-CMC solutions to the Einstein constraints on
compact manifolds with boundary.
After proving these abstract results, we give an overview
of the technical results that must be established in the remainder
of the paper in order to use the abstract results.

Before stating the main theorems, let us make precise what we mean by
near-CMC condition in this article.
We say that the extrinsic mean curvature $\tau$ satisfies the
{\bf near-CMC condition} when the following inequality is satisfied
\begin{equation}\label{nearcmc}
\|\nabla\tau\|_{z}<\Gamma\inf_{\cM}|\tau|,
\end{equation}
where the constant $\Gamma=\frac{\sqrt3}{2C}$ if
$\rho,\sigma^2\in L^\infty$, and
$\Gamma=\frac{\sqrt3}{2C}(\frac{\min uv}{\max uv})^6$ otherwise, with
the constant $C>0$ as in Corollary \ref{T:MC-E-aw} and the
continuous functions $u,v>0$ are as defined in \eqref{e:rhosigmacurva} or in \eqref{e:rhosigmacurvb} on page \pageref{e:rhosigmacurva}.  Here
$C$ depends only on the Riemannian manifold $(\cM,h_{ab})$, and not
mentioning  $(\cM,h_{ab})$, $u$ and $v$ depend only on $\rho$,
$\sigma^2$, and $\tau$. 
It is important to note that we always have $0<\frac{\min uv}{\max uv}\leqs1$, so that in any case the condition \eqref{nearcmc} is at least as strong as 
the same condition with $\Gamma$ taken to be equal to $\frac{\sqrt3}{2C}$.
The condition depends on the value of $z$, and that will be inserted through the context.

Recall that the three Yamabe classes $\cY^+(\cM)$,  $\cY^-(\cM)$ and $\cY^0(\cM)$ are defined after Eq. \eqref{E:yamabe}. 
See Appendix \ref{sec:yamabe} for more details.

\subsection{Theorem~\ref{T:main1}: Far-CMC weak solutions}
Here is the first of our three main results.
This result does not involve the near-CMC condition, which is one of
the main contributions of this paper.
The result is developed in the presence of a weak 
background metric $h_{ab} \in W^{s,p}$, for
$p \in (1,\infty)$ and $s \in (1+\frac{3}{p},\infty)$,
with the weakest possible assumptions on the data that
allows for avoiding the near-CMC condition.

\begin{theorem}
{\bf (Far-CMC $W^{s,p}$ solutions, 
      $p \in (1,\infty)$, $s \in (1+\frac{3}{p},\infty)$)}
\label{T:main1}
Let $(\cM,h_{ab})$ be a 3-dimensional closed Riemannian manifold.
Let $h_{ab}\in W^{s,p}$ admit no conformal Killing field
and be in $\cY^{+}(\cM)$,
where
$p \in (1,\infty)$ and $s \in (1+\frac{3}{p},\infty)$ are given.
Select $q$ and $e$ to satisfy:
\begin{itemize}
\item $\frac1q \in (0,1)\cap(0,\frac{s-1}{3})\cap[\frac{3-p}{3p},\frac{3+p}{3p}]$,
\item $e \in (1 + \frac{3}{q}, \infty)\cap[s-1,s]\cap[\frac3q+s-\frac3p-1,\frac{3}{q}+s-\frac{3}{p}]$.
\end{itemize}
Assume that the data satisfies:
\begin{itemize}
\item $\tau\in W^{e-1,q}$ if $e\geqs2$, and $\tau \in W^{1,z}$ otherwise, with $z = \frac{3q}{3 + \max\{0,2-e\}q}$,
\item $\sigma \in W^{e-1,q}$,
      ~with $\|\sigma^2\|_{\infty}$ sufficiently small,
\item $\rho \in W^{s-2,p}_+\cap L^{\infty} \setminus \{0\}$,
      ~with $\|\rho\|_{\infty}$ sufficiently small,
\item $\tbj \in \tbW^{e-2,q}$, ~with $\|\tbj\|_{e-2,q}$ sufficiently small.
\end{itemize}
Then there exist $\phi\in W^{s,p}$ with $\phi>0$ and
$\tbw\in\tbW^{e,q}$ solving the Einstein constraint
equations.
\end{theorem}

\begin{proof}
The proof will be given in~\S\ref{sec:proof}.
\end{proof}

\begin{figure}[htb]
\begin{center}
\includegraphics[width=0.9\textwidth]{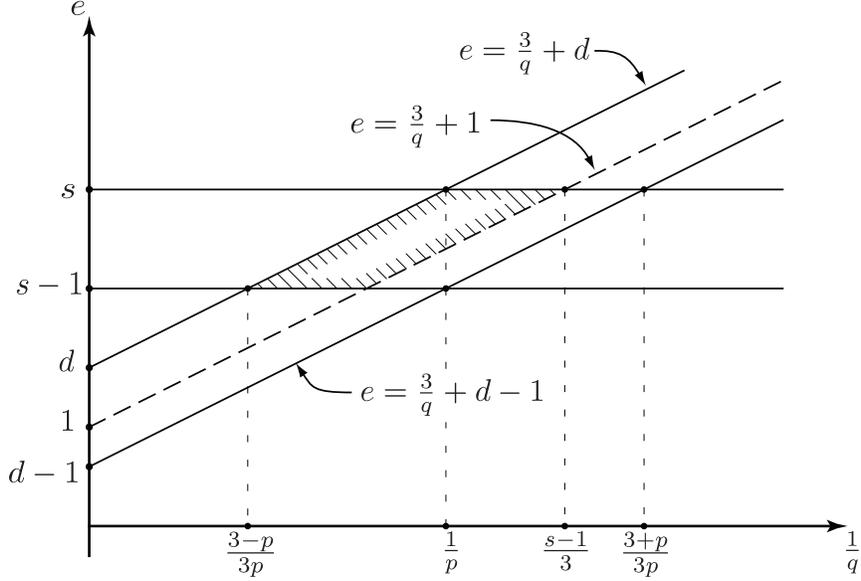}
\caption{Range of $e$ and $q$ in Theorems \ref{T:main1} and \ref{T:main2}, with $d=s-\frac3p>1$.}
\label{f:ncmc}
\end{center}
\end{figure}

\Remark
The above result avoids the near-CMC condition \eqref{nearcmc}; however, one should be aware of the various smallness conditions involved in the above theorem.
More precisely, the mean curvature $\tau$ can be chosen to be an arbitrary function from a suitable function space,
and {\em afterwards}, one has to choose $\sigma$, $\rho$, and $\bij$ satisfying
smallness conditions that depend on the chosen $\tau$.
Nevertheless, the novelty of this result is that $\tau$ can be specified freely, whereas the condition \eqref{nearcmc} is not satisfied for arbitrary $\tau$.

\subsection{Theorem~\ref{T:main2}: Near-CMC weak solutions}
Here is the second of our three main results; this result requires the 
near-CMC condition, but still extends the known near-CMC results to 
situations with weaker assumptions on metric and on the data.
In particular, the result is developed in the presence of a weak 
background metric $h_{ab} \in W^{s,p}$, for
$p \in (1,\infty)$ and $s \in (1+\frac{3}{p},\infty)$,
and with the weakest possible assumptions on the data.

\begin{theorem}
{\bf (Near-CMC $W^{s,p}$ solutions, 
      $p \in (1,\infty)$, $s \in (1+\frac{3}{p},\infty)$)}
\label{T:main2}
Let $(\cM,h_{ab})$ be a 3-dimensional closed Riemannian manifold.
Let $h_{ab}\in W^{s,p}$ admit no conformal Killing field,
where $p \in (1,\infty)$ and $s \in (1+\frac{3}{p},\infty)$
are given.
Select $q$, $e$ and $z$ to satisfy:
\begin{itemize}
\item[$\bullet$] $\frac1q \in (0,1)\cap(0,\frac{s-1}{3})\cap[\frac{3-p}{3p},\frac{3+p}{3p}]$,
\item[$\bullet$] $e \in (1 + \frac{3}{q}, \infty)\cap[s-1,s]\cap[\frac3q+s-\frac3p-1,\frac{3}{q}+s-\frac{3}{p}]$.
\item[$\bullet$] $z = \frac{3q}{3 + \max\{0,2-e\}q}$.
\end{itemize}
Assume that $\tau$ satisfies the near-CMC condition \eqref{nearcmc} with $z$ as above, and
that the data satisfies:
\begin{itemize}
\item[$\bullet$] $\tau\in W^{e-1,q}$ if $e>2$, and $\tau \in W^{1,z}$ if $e\leqs2$,
\item[$\bullet$] $\sigma \in W^{e-1,q}$,
\item[$\bullet$] $\rho \in W^{s-2,p}_+$,
\item[$\bullet$] $\tbj \in \tbW^{e-2,q}$.
\end{itemize}
In addition, let one of the following sets of conditions hold:
\begin{itemize}
\item[(a)]
 $h_{ab}$ is in $\cY^{-}(\cM)$;
    the metric $h_{ab}$ is conformally equivalent to a metric with scalar curvature $(-\tau^2)$;
\item[(b)]
 $h_{ab}$ is in $\cY^{0}(\cM)$ or in $\cY^{+}(\cM)$;
    either $\rho\not\equiv0$ and $\tau\not\equiv0$
    or $\tau\in L^{\infty}$ and $\inf_{\cM}\sigma^2$ is sufficiently large.
\end{itemize}
Then there exist $\phi\in W^{s,p}$ with $\phi>0$ and
$\tbw\in\tbW^{e,q}$ solving the Einstein constraint
equations.
\end{theorem}
\begin{proof}
The proof will be given in~\S\ref{sec:proof}.
\end{proof}

\subsection{Theorem~\ref{T:main3}: CMC weak solutions}

Here is the last of our three main results;
it covers specifically the CMC case, and allows for lower regularity
of the background metric than the non-CMC case.
In particular, the result is developed with a weak 
background metric $h_{ab} \in W^{s,p}$, for
$p \in (1,\infty)$ and $s \in (\frac{3}{p},\infty) \cap [1,\infty)$.
In the case of $s=2$, we reproduce the CMC existence results of 
Choquet-Bruhat~\cite{yCB04}, and
in the case $p=2$, we reproduce the CMC existence results of
Maxwell~\cite{dM05}, but with a different proof;
our CMC proof goes through the same analysis framework that we use to 
obtain the non-CMC results (Theorems~\ref{T:FIXPT1} and~\ref{T:FIXPT2}).
In the following theorem we do not include the trivial case $h_{ab}\in\cY^{0}$ and $\tau=\sigma=\rho=0$.

\begin{theorem}
{\bf (CMC $W^{s,p}$ solutions, 
      $p \in (1,\infty)$, $s \in (\frac{3}{p},\infty) \cap [1,\infty)$)}
\label{T:main3}
Let $(\cM,h_{ab})$ be a 3-dimensional closed Riemannian manifold.
Let $h_{ab}\in W^{s,p}$ admit no conformal Killing field,
where $p \in (1,\infty)$ and $s \in (\frac{3}{p},\infty) \cap [1,\infty)$
are given.
With $d:=s-\frac3p$, select $q$ and $e$ to satisfy:
\begin{itemize}
\item[$\bullet$] $\frac1q \in (0,1)\cap[\frac{3-p}{3p},\frac{3+p}{3p}]\cap[\frac{1-d}{3},\frac{3+sp}{6p})$,
\item[$\bullet$] $e \in [1,\infty)\cap[s-1,s]\cap[\frac3q+d-1, \frac{3}{q}+d]\cap(\frac3q+\frac{d}2,\infty)$.
\end{itemize}
Assume $\tau = \mathrm{const}$ (CMC) and that the data satisfies:
\begin{itemize}
\item[$\bullet$] $\sigma \in W^{e-1,q}$,
\item[$\bullet$] $\rho \in W^{s-2,p}_+$,
\item[$\bullet$] $\tbj \in \tbW^{e-2,q}$.
\end{itemize}
In addition, let one of the following sets of conditions hold:
\begin{itemize}
\item[(a)]
 $h_{ab}$ is in $\cY^{-}(\cM)$; $\tau\neq0$;
\item[(b)]
 $h_{ab}$ is in $\cY^{+}(\cM)$; $\rho\neq0$ or $\sigma\neq0$;
\item[(c)]
 $h_{ab}$ is in $\cY^{0}(\cM)$; $\tau\neq0$; $\rho\neq0$ or $\sigma\neq0$;
\item[(d)]
 $h_{ab}$ is in $\cY^{0}(\cM)$; $\tau=\rho=\sigma=0$; $\tbj=0$.
\end{itemize}
Then there exist $\phi\in W^{s,p}$ with $\phi>0$ and
$\tbw\in\tbW^{e,q}$ solving the Einstein constraint
equations.
\end{theorem}

\begin{proof}
The proof will be given in~\S\ref{sec:proof}.
\end{proof}

\begin{figure}[htb]
\begin{center}
\includegraphics[width=0.9\textwidth]{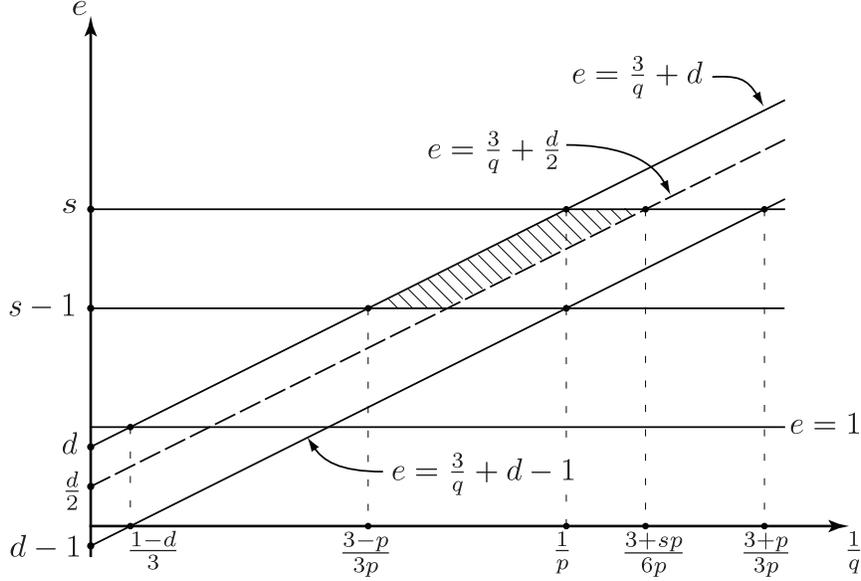}
\caption{Range of $e$ and $q$ in Theorem \ref{T:main3}. Recall that $d=s-\frac3p>0$.}
\label{f:cmc}
\end{center}
\end{figure}

\subsection{A coupled topological fixed-point argument}
In Theorems~\ref{T:FIXPT1} and~\ref{T:FIXPT2} below 
(see also~\cite{mHgNgT08b})
we give some abstract fixed-point results which form the basic
framework for our analysis of the coupled constraints.
These topological fixed-point theorems will be the main tool by which
we shall establish Theorems~\ref{T:main1}, \ref{T:main2}, 
and~\ref{T:main3} above.
They have the important feature that the required properties of the
abstract fixed-point operators $S$ and $T$ appearing in 
Theorems~\ref{T:FIXPT1} and~\ref{T:FIXPT2} 
below can be established in the case of the Einstein constraints without
using the near-CMC condition; this is not the case for fixed-point
arguments for the constraints based on $k$-contractions
(cf.~\cite{jIvM96,pAaCjI07}) which require near-CMC conditions.
The bulk of the paper then involves establishing the
required properties of $S$ and $T$ without using the 
near-CMC condition, and finding suitable global barriers 
$\phi_-$ and $\phi_+$ for defining the required set $U$ 
that are similarly free of the near-CMC condition (when possible).

We now set up the basic abstract framework.
Let $X$ and $Y$ be Banach spaces,
let $f:X \times Y \to X^*$ and $\bif:X \to Y^*$
be (generally nonlinear) operators,
let $A_{\tiIL}:Y \to Y^*$ be a linear invertible operator,
and
let $A_{\tiL}:X \to X^*$ be a linear invertible operator satisfying
the maximum principle, meaning that 
$A_{\tiL}u \leqs A_{\tiL}v \Rightarrow u \leqs v$.
The order structure on $X$ for interpreting the maximum principle
will be inherited from an ordered Banach space $Z$ 
(see Appendices~\ref{sec:OBS}, \ref{sec:monotone}, and~\ref{sec:maxprinciple}, 
and also cf.~\cite{Zeidler-I})
through the compact embedding $X \hookrightarrow Z$,
which will also make available compactness arguments.

The coupled Hamiltonian and momentum constraints can be viewed 
abstractly as coupled operator equations of the form:
\begin{eqnarray}
A_{\tiL} \phi + f(\phi,w) & = & 0, 
\label{E:ham-abstract} \\
A_{\tiIL} w + \bif(\phi) & = & 0,
\label{E:mom-abstract}
\end{eqnarray}
or equivalently as the coupled fixed-point equations
\begin{eqnarray}
\phi & = & T(\phi,w),
\label{E:ham-abstract-fixpt} \\
w & = & S(\phi),
\label{E:mom-abstract-fixpt}
\end{eqnarray}
for appropriately defined fixed-point
maps $T : X \times Y \to X$ and $S : X \to Y$.
The obvious choice for $S$ is the
{\em Picard map} for~(\ref{E:mom-abstract})
\begin{equation}
  \label{E:mom-abstract-picard}
S(\phi) = -A_{\tiIL} ^{-1} \bif(\phi),
\end{equation}
which also happens to be the solution map for~(\ref{E:mom-abstract}).
On the other hand, there are a number of distinct possibilities
for $T$, ranging from the solution map for~(\ref{E:ham-abstract}),
to the {\em Picard map} for~(\ref{E:ham-abstract}),
which inverts only the linear part of the operator in~(\ref{E:ham-abstract}):
\begin{equation}
  \label{E:ham-abstract-picard}
T(\phi,w) = -A_{\tiL}^{-1}f(\phi,w).
\end{equation}

Assume now that $T$ is as in~(\ref{E:ham-abstract-picard}),
and (for fixed $w \in Y$) that
$\phi_-$ and $\phi_+$ are sub- and super-solutions
of the semi-linear operator equation~(\ref{E:ham-abstract})
in the sense that
$$
A_{\tiL}\phi_- + f(\phi_-,w) \leqs 0,
\quad \quad
A_{\tiL}\phi_+ + f(\phi_+,w) \geqs 0.
$$
The assumptions on $A_{\tiL}$
imply (see Lemma~\ref{L:T-sub-super} in Appendix~\ref{sec:monotone})
that for fixed $w \in Y$, 
$\phi_-$ and $\phi_+$ are also sub- and super-solutions of 
the equivalent fixed-point equation:
$$
\phi_- \leqs T(\phi_-,w),
\quad \quad
\phi_+ \geqs T(\phi_+,w).
$$
For developing results on fixed-point iterations in ordered
Banach spaces, it is convenient to work with maps which are
monotone increasing in $\phi$, for fixed $w \in Y$:
$$
\phi_1 \leqs \phi_2 \quad \Longrightarrow \quad T(\phi_1,w) \leqs T(\phi_2,w).
$$
The map $T$ that arises as the Picard map for a semi-linear 
problem will generally not be monotone increasing; 
however, if there exists a continuous linear monotone increasing
map $J : X \rightarrow X^{*}$,
then one can always introduce a positive shift $s$ into the 
operator equation
$$
A_{\tiL}^s \phi + f^s(\phi,w) = 0,
$$
with $A_{\tiL}^s = A_{\tiL} + sJ$ and $f^s(\phi,w) = f(\phi,w) - sJ \phi$.
(Throughout this paper, the spaces we encounter for $X$ typically
fit into a Gelfand triple $X\hookrightarrow H \hookrightarrow X^{*}$,
where the ``pivot'' space $H$ is Hilbert space, and the continuous
map between $X$ and $X^*$ is a composition of the two inclusion maps.)
Since $s > 0$ the shifted operator $A_{\tiL}^s$ 
retains the maximum principle property of $A_{\tiL}$, and if $s$ is 
chosen sufficiently large then $f^s$ is monotone decreasing in 
$\phi$. 
Under the additional condition on $J$ and $s$ that $A_{\tiL}^s$ is 
invertible (see also \cite{mHgNgT08b}), the shifted Picard map 
$$
T^s(\phi,w) = -(A_{\tiL}^s)^{-1} f^s(\phi,w)
$$
is now monotone increasing in $\phi$.

We now give two abstract existence results for systems of the 
form~(\ref{E:ham-abstract-fixpt})--(\ref{E:mom-abstract-fixpt}).
\begin{theorem}
{\bf (Coupled Fixed-Point Principle A)}
\label{T:FIXPT1}
Let $X$ and $Y$ be Banach spaces, 
and let $Z$ be a Banach space 
with compact embedding $X \hookrightarrow Z$.
Let $U \subset Z$ be a non-empty, convex, closed, bounded subset,
and let
\[
S:U \to \mathcal{R}(S) \subset Y,
\quad
\quad
T:U \times \mathcal{R}(S) \to U \cap X,
\]
be continuous maps.
Then there exist 
$\phi \in U \cap X$
and 
$w\in\mathcal{R}(S)$
such that
\[
\phi=T(\phi,w) \quad \textrm{and} \quad w=S(\phi).
\]
\end{theorem}
\begin{proof}
The proof will be through a standard variation of
the Schauder Fixed-Point Theorem, reviewed as
Theorem~\ref{T:SCHAUDER-B} in Appendix~\ref{sec:fixpt}.
The proof is divided into several steps.

{\em Step 1: Construction of a non-empty, 
             convex, closed, bounded subset $U \subset Z$.}
By assumption we have that $U \subset Z$ is
non-empty,
convex (involving the vector space structure of $Z$),
closed (involving the topology on $Z$),
and 
bounded (involving the metric given by the norm on $Z$).

{\em Step 2: Continuity of a mapping 
             $G:U \subset Z \to U \cap X \subset X$.}
Define the composite operator
\[
G := T \circ S : U \subset Z \to U \cap X \subset X.
\]
The mapping $G$ is continuous, since it is a composition 
of the continuous operators
$S:U \subset Z \to \mathcal{R}(S) \subset Y$
and
$T:U \subset Z \times \mathcal{R}(S) \to U \cap X \subset X$.

{\em Step 3: Compactness of a mapping $F:U \subset Z \to U \subset Z$.}
The compact embedding assumption $X \hookrightarrow Z$ implies that 
the canonical injection operator $i:X \to Z$ is compact.
Since the composition of compact and continuous operators is compact, 
we have the composition
$F := i\circ G : U \subset Z \to U \subset Z$
is compact.

{\em Step 4: Invoking the Schauder Theorem.}
Therefore, by a standard variant of the Schauder Theorem
(see Theorem~\ref{T:SCHAUDER-B} in Appendix~\ref{sec:fixpt}),
there exists a fixed-point $\phi \in U$ 
such that $\phi=F(\phi)=T(\phi,S(\phi))$.
Since $\mathcal{R}(T) = U \cap X$, in fact $\phi \in U \cap X$.
We now take $w=S(\phi) \subset \mathcal{R}(S)$ and we have the result.
\end{proof}

The assumption in Theorem~\ref{T:FIXPT1} that the mapping
$T$ is invariant on the non-empty, closed, convex, bounded
subset $U$ can be established
using {\em a priori} estimates if $T$ is the solution 
mapping, but if there are multiple fixed-points then continuity of
$T$ will not hold.  Fixed-point theory for set-valued maps could
still potentially be used (cf.~\cite{Zeidler-I}).
On the other hand, if $T$ is chosen to be the
Picard map, then it is typically easier to establish continuity
of $T$ even with multiple fixed-points, but more difficult to establish 
the invariance property without additional conditions on $T$.
In our setting, we wish to allow for non-uniqueness in the
Hamiltonian constraint (for example see~\cite{mHgNgT08b} for possible
non-uniqueness in the case of compact manifolds with boundary), so will 
generally focus on the Picard map for the Hamiltonian constraint in our 
fixed-point framework for the coupled constraints.
The following special case of Theorem~\ref{T:FIXPT1} gives some
simple sufficient conditions on $T$ to establish the invariance 
using barriers in an ordered Banach space (for a review of ordered 
Banach spaces, see Appendix~\ref{sec:OBS} or~\cite{Zeidler-I}).

\begin{theorem}
{\bf (Coupled Fixed-Point Principle B)}
\label{T:FIXPT2}
Let $X$ and $Y$ be Banach spaces, 
and let $Z$ be a real ordered Banach space
having the compact embedding $X \hookrightarrow Z$.
Let $[\phi_-,\phi_+] \subset Z$ be a nonempty interval which
is closed in the topology of $Z$,
and set $U = [\phi_-,\phi_+] \cap \overline{B}_M \subset Z$
where $\overline{B}_M$ is the closed ball
of finite radius $M>0$ in $Z$ about the origin.
Assume $U$ is nonempty, and let the maps
\[
S:U \to \mathcal{R}(S) \subset Y,
\quad
\quad
T:U \times \mathcal{R}(S) \to U \cap X,
\]
be continuous maps.
Then there exist 
$\phi\in U \cap X$
and 
$w\in\mathcal{R}(S)$
such that
\[
\phi=T(\phi,w)\quad\textrm{and}\quad w=S(\phi).
\]
\end{theorem} 
\begin{proof}
By choosing the set $U$ to be the non-empty intersection of the interval
$[\phi_-,\phi_+]$ with a bounded set in $Z$, we have $U$ bounded in $Z$.
We also have that $U$ is convex with respect to the vector space structure
of $Z$, since it is the intersection of two convex sets $[\phi_-, \phi_+]$
and $\overline{B}_M$.
Since $U$ is the intersection of the interval $[\phi_-,\phi_+]$, which 
by assumption is closed in the topology of $Z$, with the 
closed ball $\overline{B}_M$ in $Z$, $U$ is also closed.
In summary, we have that $U$ is non-empty as a subset of $Z$, 
closed in the topology of $Z$, 
convex with respect to the vector space structure of $Z$, 
and bounded with respect to the metric (via normed) space structure of $Z$.
Therefore, the assumptions of Theorem~\ref{T:FIXPT1} hold and the result 
then follows.
\end{proof}

We make some final remarks about Theorems~\ref{T:FIXPT1} and~\ref{T:FIXPT2}.
If the ordered Banach space $Z$ in Theorem~\ref{T:FIXPT2}
had a {\em normal} order cone, then the closed interval
$[\phi_-,\phi_+]$ would automatically be bounded in the norm
of $Z$ (see Lemma~\ref{L:normal} in Appendix~\ref{sec:OBS} or~\cite{Zeidler-I}
for this result).
The interval by itself is also non-empty and closed by assumption,
and trivially convex (see Appendix~\ref{sec:OBS}), so that 
Theorem~\ref{T:FIXPT2} would follow immediately from Theorem~\ref{T:FIXPT1}
by simply taking $U = [\phi_-,\phi_+]$.
Second, the closed ball $\overline{B}_M$ in Theorem~\ref{T:FIXPT2}
can be replaced with any non-empty, convex, closed, and bounded subset 
of $Z$ having non-trivial intersection with the interval $[\phi_-,\phi_+]$.
Third, in the case that $T$ in Theorem~\ref{T:FIXPT2} arises as the Picard
map~(\ref{E:ham-abstract-picard}) of the 
semi-linear problem~(\ref{E:ham-abstract}), we can always
ensure that $T$ is invariant on $U$ in Theorem~\ref{T:FIXPT2}
by: (1) obtaining sub- and super-solutions to the semi-linear operator
equation and using these for $\phi_-$ and $\phi_+$, since these will
also be sub- and super-solutions for the fixed-point equation involving
the Picard map;
(2) introducing a shift into the nonlinearity to ensure $T$ is monotone
increasing;
and
(3) obtaining {\em a priori} norm bounds on Picard iterates.
As noted earlier, (1) and (2) will ensure
\begin{equation}
\label{E:cone-bounds}
\phi_- \leqs T(\phi_-,w) \leqs T(\phi,w) \leqs T(\phi_+,w) \leqs \phi_+,
\end{equation}
for all $\phi \in [\phi_-,\phi_+]$, and $w \in \mathcal{R}(S)$,
whereas (3) ensures that
\begin{equation}
\label{E:norm-bounds}
\| T(\phi,w) \|_X \leqs M,
\quad
\forall \phi \in [\phi_-,\phi_+],
\quad
\forall w \in \mathcal{R}(S),
\end{equation}
which together ensure $T:U \times \mathcal{R}(S) \to U \cap X$,
where $U = [\phi_-,\phi_+] \cap \overline{B}_M \subset Z$.
Again, if $Z$ has a normal order cone structure, then 
ensuring~(\ref{E:cone-bounds}) holds will
automatically guarantee that~(\ref{E:norm-bounds}) also holds,
so it is not necessary to establish~(\ref{E:norm-bounds}) separately
in the case of a normal order cone.

Finally, note that Theorem~\ref{T:FIXPT2} also allows one to choose the
solution map (or any other fixed-point map) for $T$ together with 
{\em a priori} order cone and norm estimates to ensure the 
conditions~(\ref{E:cone-bounds}) and~(\ref{E:norm-bounds}) hold
(as long as continuity for $T$ can be shown).
Even if {\em a priori} order-cone estimates cannot be shown to hold 
directly for this choice of $T$,
as long as the map can be ``bracketed'' in the interval $[\phi_-,\phi_+]$
by two auxiliary monotone increasing maps, then it can be shown
that~(\ref{E:cone-bounds}) holds.
This allows one to use the Picard map even if it is not monotone increasing, 
without having to introduce the shift into the Picard map.

The overall argument we use to prove the non-CMC results in
Theorems~\ref{T:main1}, \ref{T:main2}, and~\ref{T:main3}
using Theorems~\ref{T:FIXPT1} and~\ref{T:FIXPT2}
involves the following steps:
\begin{itemize}
\item[{\em Step 1:}]{\em The choice of function spaces}.
      We will choose the spaces for use of Theorem~\ref{T:FIXPT2} as follows:
\begin{itemize}
\item $X=W^{s,p}$, with $p\in(1,\infty)$, and $s(p)\in(1+\frac{3}{p},\infty)$.
      In the CMC case in Theorem~\ref{T:main3}, we can lower $s$ to
       $s(p)\in(\frac{3}{p},\infty) \cap [1,\infty)$.
\item $Y=\biW^{e,q}$, 
      with
      $e$ and $q$ as given in the theorem statements.
\item $Z=W^{\tilde{s},p}$, $\tilde{s} \in (\frac{3}{p},s)$,
      so that $X=W^{s,p}\hookrightarrow W^{\tilde{s},p} = Z$ is compact.
\item $U=[\phi_-,\phi_+]_{\tilde{s},p} \cap \overline{B}_M
      \subset W^{\tilde{s},p} = Z$,
      with $\phi_-$ and $\phi_+$ global barriers
      (sub- and super-solutions, respectively) for
      the Hamiltonian constraint equation which satisfy the compatibility
      condition: $0 < \phi_- \leqs \phi_+ < \infty$.
\end{itemize}
\item[{\em Step 2:}]{\em Construction of the mapping $S$}.
      Assuming the existence of ``global'' weak sub- and super-solutions
      $\phi_-$ and $\phi_+$, and assuming the fixed function
      $\phi \in U = [\phi_-,\phi_+]_{\tilde{s},p} \cap \overline{B}_M
      \subset W^{\tilde{s},p} = Z$ is taken as data in the
      momentum constraint, we establish continuity and related properties
      of the momentum constraint solution map
      $S : U \to \mathcal{R}(S) \subset \biW^{e,q} = Y$.
      (\S\ref{sec:momentum})
\item[{\em Step 3:}]{\em Construction of the mapping $T$}.
      Again existence of ``global'' weak sub- and super-solutions
      $\phi_-$ and $\phi_+$, 
      with fixed $w \in \mathcal{R}(S) \subset \biW^{e,q} = Y$
      taken as data in the Hamiltonian constraint, we establish
      continuity and related properties of the Picard map
      $T: U \times \mathcal{R}(S) \to U \cap W^{s,p}$.
      Invariance of $T$ on $U=[\phi_-,\phi_+]_{\tilde{s},p} 
          \cap \overline{B}_M \subset W^{\tilde{s},p}$
      is established using a combination of {\em a priori} order cone bounds
      and norm bounds.
      (\S\ref{sec:Hamiltonian})
\item[{\em Step 4:}]{\em Barrier construction}.
      Global weak sub- and super-solutions $\phi_-$ and $\phi_+$
      for the Hamiltonian constraint are explicitly constructed to
      build a nonempty, convex, closed, and bounded
      subset $U=[\phi_-,\phi_+]_{\tilde{s},p} \cap \overline{B}_M \subset W^{\tilde{s},p}$, 
      which is a strictly positive interval.
      These include variations of known barrier constructions which 
      require the near-CMC condition, and also some new barrier
      constructions which are free of the near-CMC condition.
      (\S\ref{sec:barriers})
      {\bf\em Note: This is the only place in the argument where
      near-CMC conditions may potentially arise.}
\item[{\em Step 5:}]{\em Application of fixed-point theorem}.
      The global barriers and continuity properties
      are used together with the abstract topological fixed-point
      result (Theorem~\ref{T:FIXPT2}) to 
      establish existence
      of solutions $\phi \in U \cap W^{s,p}$ and $w \in \biW^{e,q}$
      to the coupled system: $w=S(\phi), \phi=T(\phi,w).$ (\S\ref{sec:proof})
\item[{\em Step 6:}]{\em Bootstrap}.
      The above application of a fixed-point theorem is actually performed for some low regularity spaces,
      i.e., for $s\leqs2$ and $e\leqs2$
,      and a bootstrap argument is then given to extend the results to
      the range of $s$ and $p$ given in the statement of the Theorem.
      (\S\ref{sec:proof})
\end{itemize}

The ordered Banach space $Z$ plays a central role in Theorem~\ref{T:FIXPT2}.
We will use $Z=W^{t,q}$, ~$t \geqs 0$, ~$1 \leqs q \leqs \infty$, 
with order cone defined as in~(\ref{E:wkp-cone}).
Given such an order cone, one can define the closed interval
$$
[\phi_-,\phi_+]_{t,q}
  = \{ \phi \in W^{t,q} : \phi_- \leqs \phi \leqs \phi_+ \} \subset W^{t,q},
$$
which as noted earlier is denoted more simply as
$[\phi_-,\phi_+]_q$ when $t=0$,
and as simply
$[\phi_-,\phi_+]$ when $t=0$, $q=\infty$.
When $t=0$, the $W^{t,q}$ order cone is {\em normal} for 
$1 \leqs q \leqs \infty$, meaning that
closed intervals $[\phi_-,\phi_+]_q \subset L^q = W^{0,q}$
are automatically bounded in the metric given by the norm on $L^q$.

If we consider the interval $U=[\phi_-,\phi_+]_{t,q} \subset W^{t,q} = Z$
defined using this order structure, it will be critically important
to establish that $U$ is 
convex (with respect to the vector space structure of $Z$),
closed (in the topology of $Z$),
and (when possible) bounded (in the metric given by the norm on $Z$).
It will also be important that $U$ be nonempty as a subset of $Z$;
this will involve choosing compatible $\phi_-$ and $\phi_+$.
Regarding convexity, closure, and boundedness, we have the
following lemma.

\begin{lemma}
{\bf (Order cone intervals in $W^{t,q}$)}
\label{L:wsp-interval}
For $t \geqs 0$, $1 \leqs q \leqs \infty$, the set
$$
U = [\phi_-,\phi_+]_{t,q}
  = \{ \phi \in W^{t,q} : \phi_- \leqs \phi \leqs \phi_+ \} \subset W^{t,q}
$$
is convex with respect to the vector space structure of \ $W^{t,q}$
and closed in the topology of \ $W^{t,q}$.
For $t=0$, $1 \leqs q \leqs \infty$, the set $U$ is also bounded with 
respect to the metric space structure of $L^q=W^{0,q}$.
\end{lemma}

\begin{proof}
That $U$ is convex for $t\geqs 0$, $1 \leqs q \leqs \infty$,
follows from the fact that any interval
built using order cones is convex.
That $U$ is closed in the case of $t=0$, $1 \leqs q \leqs \infty$
follows from the fact that norm convergence in 
$L^q$ for $1 \leqs q \leqs \infty$ implies pointwise subsequential 
convergence almost everywhere (see Theorem~3.12 in~\cite{Rudi87}).
That $U$ is bounded when $t=0$, $1 \leqs q \leqs \infty$ follows 
from the fact that the order cone $L^q_+$ is 
normal (see Appendix~\ref{sec:OBS}).

What remains is to show that
$U$ is closed in the case of $t>0$, $1 \leqs q \leqs \infty$.
The argument is as follows.
Let $\{ u_k \}_{k=1}^{\infty}$
be a Cauchy sequence in $U \subset W^{t,q} \subset L^q$,
with $t>0$, $1 \leqs q \leqs \infty$.
From completeness of $W^{t,q}$
there exists $\lim_{k\rightarrow \infty} u_k = u \in W^{t,q}$.
From the continuous embedding $W^{t,q} \hookrightarrow L^q$ for $t>0$,
we have that
$$
\| u_k - u_l \|_q \leqs C \| u_k - u_l \|_{t,q}
$$
so that $u_k$ is also Cauchy in $L^q$.
Moreover, the continuous embedding also implies that
$u$ is also the limit of $u_k$ as a sequence in $L^q$.
Since $[\phi_-,\phi_+]_{0,q}$ is closed in $L^q$, 
we have $u \in [\phi_-,\phi_+]_{0,q}$,
and so $u \in U = [\phi_-,\phi_+]_{t,q} 
             = [\phi_-,\phi_+]_{0,q} \cap W^{t,q}$.
\end{proof}

\Remark
We indicate now how the far-CMC result outlined in~\cite{mHgNgT07a}
can be recovered using Theorem~\ref{T:FIXPT1} above.
The framework is constructed by taking
$X=W^{2,p}$, $Y=W^{2,p}$, and $Z=L^{\infty}$, with $p>3$,
giving the compact embedding $W^{2,p} \hookrightarrow L^{\infty}$.
The coefficients are assumed to satisfy
$\tau \in W^{1,p}$ and $\sigma^2, j^a, \rho \in L^{p}$ as well as the 
assumptions for the construction of a near-CMC-free global 
super-solution (presented in~\cite{mHgNgT07a} as Theorem~1, analogous
to Lemma~\ref{L:global-super} in this paper), and for the construction
of a near-CMC-free global sub-solution 
(presented in~\cite{mHgNgT07a} as Theorem~2, analogous 
to Lemma~\ref{L:global-sub} in this paper).
One then takes $U=[\phi_-,\phi_+] \subset Z=L^{\infty}$, where
the compatible $0 < \phi_- \leqs \phi_+$ are these near-CMC-free barriers.
Since $Z=L^{\infty}$ is an ordered Banach space with normal
order cone, we have (by Lemma~\ref{L:wsp-interval} in this paper) that $U$ is 
non-empty, convex, closed and bounded as a subset of $Z$.
The invariance of the Picard mapping on the interval $[\phi_-,\phi_+]$ is proven using a monotone shift (cf. Lemma \ref{l:shift1} in this paper) and a barrier argument (cf. Lemma \ref{l:shiftsubsup} in this paper).
The main result in~\cite{mHgNgT07a} (stated in~\cite{mHgNgT07a} as Theorem~4),
now follows from Theorem~\ref{T:FIXPT1} in this paper
(stated in~\cite{mHgNgT07a} as Lemma~1).

\section{Weak solution results for the individual constraints}
   \label{sec:individual}

\subsection{The momentum constraint and the solution map $S$}
   \label{sec:momentum}
In this section we fix a particular scalar function $\phi\in
W^{{s},p}$ with ${s}p>3$, and consider separately the momentum constraint
equation~(\ref{WF-LYm2}) to be solved for the vector valued function
$\biw$. The result is a linear elliptic system of
equations for this variable $\biw=\biw_{\phi}$.
For convenience, we reformulate the problem here in a self-contained manner.
Note that the problem \eqref{MC-LYm1} below is identical to \eqref{WF-LYm2}
provided the functions $\bib_\tau$ and $\bib_j$ are defined accordingly.
Our goal is not only to develop some existence results for the momentum
constraint, but also to derive the estimates for the momentum constraint
solution map $S$ that we will need later in our analysis of the coupled
system.
We note that a complete weak solution theory for the momentum constraint
on compact manifolds with boundary, using both variational methods and
Riesz-Schauder Theory, is developed in~\cite{mHgNgT08b}.

Let $(\cM, h)$ be a 3-dimensional Riemannian manifold, where $\cM$
is a smooth, compact manifold without boundary, and with $p\in(1,\infty)$ and $s\in(\frac3p,\infty)$, $h\in W^{s,p}$ is a
positive definite metric. With 
$$\textstyle
q\in(1,\infty),
\qquad\textrm{and}\qquad
e\in(2-s,s]\cap(-s+\frac3p-1+\frac3q,s-\frac3p+\frac3q], 
$$
introduce the bounded linear operator
\[
A_{\tiIL} : \biW^{e,q}\to \biW^{e-2,q},
\]
as the unique extension of the operator $\IL$ in~\eqref{CF-def-IL},
cf. Lemmata \ref{l:bdd-operator} and \ref{l:ell-est-loc} 
in Appendix \ref{sec:killing}.
Fix the source terms $\bib_{\tau}, \bib_{j} \in \biW^{e-2,q}$.
Fix a function $\phi\in W^{{s},p}$, and define
\begin{equation}
\bif_{\phi} \in \biW^{s-2,q},\qquad
\label{MC-x0}
\bif_{\phi} :=  \bib_{\tau} \phi^6 + \bib_{j}.
\end{equation}
We used the subscript $\phi$ in $\bif_{\phi}$ to emphasize that
$\phi$ is not a variable (but the ``source'') of the problem.
Note that the above conditions on $q$ and $e$ are sufficient for
the pointwise multiplication by an element of $W^{{s},p}$ to be a bounded map in $\biW^{e-2,q}$, cf. Corollary \ref{c:alg}(a) in Appendix \ref{sec:Sobolev}.

The momentum constraint equation is the
following: find an element $\biw\in\biW^{e,q}$ solution of
\begin{equation}
\label{MC-LYm1}
A_{\tiIL}\biw + \bif_{\phi} = 0.
\end{equation}

We sketch here a proof of existence of weak solutions of the momentum constraint equation \eqref{MC-LYm1}.

\begin{theorem}{\bf (Momentum constraint)}
\label{T:w-MC}
Let $e$ and $q$ be as above.
Then there exists a solution $\tbw\in\tbW^{e,q}$
to the momentum constraint equation \eqref{MC-LYm1} if and only if
$\tbf_\phi(\tbv)=0$ for all
$\tbv\in \tbW^{2-e,q'}$ satisfying $A_{\tiIL}^{*}\tbv=0$.
The solution is unique if and only if
the kernel of $A_{\tiIL}^{*}$ is trivial.
Moreover, if a solution exists at all in $\tbW^{e,q}$, for any given closed linear space $K\subseteq\tbW^{e,q}$ such that $\tbW^{e,q}=\mathrm{ker}\,A_{\tiIL}\oplus K$,
there is a unique solution satisfying $\tbw\in K$, and for this solution, we have
\begin{equation}
\label{w-MC-est}
\|\tbw\|_{e,q}
\leqs
C\, \|\tbf_{\phi}\|_{e-2,q},
\end{equation}
with some constant $C>0$ not depending on $\tbw$.
\end{theorem}

\begin{proof}
By Lemma \ref{l:ell-semi-fred} in Appendix \ref{sec:killing}, the operator $A_{\tiIL}$ is semi-Fredholm,
and moreover since $A_{\tiIL}$ is formally self-adjoint, it is Fredholm.
The formal self-adjointness also implies that when the metric is smooth, index of $A_{\tiIL}$ is zero independent of $e$ and $q$.
Now we can approximate the metric $h$ by smooth metrics so that $A_{\tiIL}$ is sufficiently close to a Fredholm operator with index zero. 
Since the set of Fredholm operators with constant index is open, we conclude that the index of $A_{\tiIL}$ is zero, and the theorem follows.
\end{proof}

In the later sections we need to bound the coefficient $a_{\biw}$ in the Hamiltonian constraint equation, which can be obtained by using the following observation.

\begin{corollary}
\label{T:MC-E-aw}
Let $p\in(1,\infty)$ and $s\in(1+\frac3p,\infty)$.
In addition, let $q\in(3,\infty)$ and $e\in(1,s]\cap(1+\frac3q,s-\frac3p+\frac3q]\cap(1,2]$,
and with $z=\frac{3q}{3+(2-e)q}$, let $\tbb_{\tau}\in\tbL^{z}$.
Assume that the momentum constraint equation has a solution $\tbw\in\tbW^{e,q}$.
Then, we have
\begin{equation}
\label{w-MC-awest}
\|\cL\tbw\|_{\infty}
\leqs
C\, \|\phi\|_{\infty}^6\|\tbb_{\tau}\|_{z}+C\,\|\tbb_{j}\|_{e-2,q},
\end{equation}
with $C>0$ not depending on $\tbw$.
Moreover, if the solution is unique, the norm $\|\tbw\|_{e,q}$ can be bounded by the same expression.
\end{corollary}

\begin{proof}
Since the kernel of $A_{\tiIL}$ is finite dimensional, we can write $\biW^{e,q}=\mathrm{ker}\,A_{\tiIL}\oplus K$
with a closed linear space $K\subseteq\biW^{e,q}$.
We have the splitting $\biw=\biw_0+\biw_1$ such that
$\biw_0\in\mathrm{ker}\,A_{\tiIL}=\mathrm{ker}\,\cL$ and
$\biw_1\in K$, implying that
$$
\|\cL\biw\|_{\infty}=\|\cL\biw_1\|_{\infty}\leqs c\,\|\biw_1\|_{1,\infty}\leqs c'\,\|\biw_1\|_{e,q},
$$
the latter inequality by $\biW^{e,q}\hookrightarrow\biW^{1,\infty}$.
We note that demanding $\biW^{e,q}\hookrightarrow\biW^{1,\infty}$ gives us the lower bound $e>1+\frac3q$,
and this in turn implies $s>1+\frac3p$ if the range of $e$ is to be nonempty.
To complete the proof, we note that $\biw_1$ is also a solution of the momentum constraint, and 
taking into account $\biL^z\hookrightarrow\biW^{e-2,q}$, we apply Theorem \ref{T:w-MC} to bound the norm $\|\biw_1\|_{e,q}$.
Note that the latter embedding requires $e\leqs2$, and combining this with $e>1+\frac3q$, we need $q>3$.
\end{proof}

We now establish some properties of the momentum constraint solution
map $S$ that we will need later for our analysis of the coupled system.
Suppose that the conditions for Theorem~\ref{T:w-MC} hold, so that
the momentum constraint is uniquely solvable.
Then for any fixed $\phi_+\in W^{s,p}$ with $\phi_{+}>0$, there exists a
mapping
\begin{equation}
S:[0,\phi_{+}]\cap W^{s,p} \to\biW^{e,q}
\end{equation}
that sends the source $\phi$ to the corresponding solution $\biw$
of the momentum constraint equation.
Since the momentum constraint is linear, it follows easily that $S$
is Lipschitz continuous as stated in the following lemma.

\begin{lemma}
{\bf (Properties of the map $S$)}
\label{T:MC-E-Lip1}
In addition to the conditions imposed in the beginning of this section, let $s\geqs1$.
Let $e\in[1,3]$ and $\frac1q\in(\frac{e-1}2\delta,1-\frac{3-e}2\delta)$, where $\delta=\max\{0,\frac1p-\frac{s-1}3\}$.
Assume that the momentum constraint \eqref{MC-LYm1} is uniquely
solvable in $\tbW^{e,q}$.
With some $\phi_+\in W^{s,p}$ satisfying $\phi_{+}>0$, let $\tbw_1$ and $\tbw_2$ be the solutions
to the momentum constraint with the source
functions $\phi_1$ and $\phi_2$ from the set $[0,\phi_+]\cap W^{s,p}$, respectively.
Then,
\begin{equation}
\|\tbw_1-\tbw_2\|_{e,q}
\leqs
C\, \|\phi_+\|_{\infty}^5
\|\tbb_{\tau}\|_{e-2,q}\,
\|\phi_1-\phi_2\|_{s,p}.
\end{equation}
\end{lemma}

\begin{proof}
The functions $\phi_1$ and $\phi_2$ pointwise satisfy the following inequalities
\begin{equation}\label{HC-phin}
\begin{split}
\phi_2^n-\phi_1^n &=
\Bigl(\sum_{j=0}^{n-1} \phi_2^j \phi_1^{n-1-j} \Bigr)
(\phi_2- \phi_1)
\leqs n\, (\phi_{+})^{n-1} \, |\phi_2- \phi_1|,\\
-\bigl[ \phi_2^{-n} -\phi_1^{-n} \bigr] &=
\frac{\phi_2^n-\phi_1^n}{(\phi_2\phi_1)^n}
\leqs n\, \frac{(\phi_{+})^{n-1}}{(\phi_{-})^{2n}} \, |\phi_2- \phi_1|,
\end{split}
\end{equation}
for any integer $n>0$.
Since the equation \eqref{MC-LYm1} is linear, applying Theorem \ref{T:w-MC} with the right hand side $\bif:=\bif_{\phi_1}-\bif_{\phi_2}$, and by using Lemma \ref{l:nem} in Appendix, we obtain
\begin{align*}
\|\biw_1-\biw_2\|_{e,q}
&\leqs
\|\bib_{\tau}\|_{e-2,q}\,
\|\phi_1^6-\phi_2^6\|_{s,p}
\leqs
6\|\phi_+\|_\infty^5
\|\bib_{\tau}\|_{e-2,q}\,
\|\phi_1-\phi_2\|_{s,p}.
\end{align*}
\end{proof}

\subsection{The Hamiltonian constraint and the Picard map $T$}
   \label{sec:Hamiltonian}
In this section we fix a particular function $a_{\biw}$ in an
appropriate space and we then separately look for weak solutions of the
Hamiltonian constraint equation \eqref{WF-LYs2}.
For convenience, we reformulate the problem here in a self-contained manner.
Note that the problem \eqref{HC-LYs1} below is identical to \eqref{WF-LYs2},
provided the functionals $a_\tau$ and $a_\rho$ are defined accordingly.
Our goal here is primarily to establish some properties and derive some
estimates for a Hamiltonian constraint fixed-point map $T$ that we will 
need later in our analysis of the coupled system, and also
for the analysis of the Hamiltonian constraint alone in the CMC setting.
We remark that a complete weak solution theory for the Hamiltonian constraint
on compact manifolds with boundary, using both variational methods and
fixed-point arguments based on monotone increasing maps, combined
with sub- and super-solutions, is developed in~\cite{mHgNgT08b}.

Let $(\cM, h)$ be a 3-dimensional Riemannian
manifold, where $\cM$ is a smooth, compact manifold without
boundary, and with $p\in(1,\infty)$ and $s\in(\frac3p,\infty)\cap[1,\infty)$, $h\in W^{s,p}$ is a positive definite metric. 
Introduce the operator
\[
A_{\tiL} : W^{s,p}\to W^{s-2,p},
\]
as the unique extension of the Laplace-Beltrami operator $L=-\Delta$, cf. Lemma \ref{l:bdd-operator} in Appendix \ref{sec:killing}.
Fix the source functions
$$\textstyle
a_\tau, a_\rho, a_{\biw} \in W^{s-2,p}_{+}, 
\qquad\textrm{and}\quad
a_{\tiR}=\frac18R \in W^{s-2,p},
$$
where $R$ is the scalar curvature of the metric $h$.
(By Corollary \ref{c:alg}(b) in Appendix \ref{sec:Sobolev},
we know $h_{ab} \in W^{s,p}$ implies $R \in W^{s-2,p}$.)
Given any two functions $\phi_{-},\phi_{+}
\in W^{s,p}$ with $0<\phi_{-}\leqs\phi_{+}$, introduce the nonlinear operator
\begin{gather}
f_{\biw}: [\phi_{-},\phi_{+}]_{s,p} \to W^{s-2,p},\qquad
\label{HC-x0}
f_{\biw}(\phi) =
  a_{\tau} \phi^5 + a_{\tiR} \phi
- a_{\rho} \phi^{-3} - a_{\biw} \phi^{-7},
\end{gather}
where the pointwise multiplication by an element of $W^{{s},p}$ defines a bounded linear map in
$W^{s-2,p}$ since $s-2\geqs-s$ and $2(s-\frac3p)>0>2-3$, cf. Corollary \ref{c:alg}(a) in Appendix \ref{sec:Sobolev}.
In case the coupled system is under consideration, the dependence of $f_{\biw}$ on $\biw$ is hidden in the fact that the coefficient $a_{\biw}$ depends on $\biw$, cf. \eqref{CF-def-coeff2}.
For generality, in the following we will view that the operator $f_{\biw}$ depends on $a_{\biw}$.

We now formulate the Hamiltonian constraint equation as follows: find an element $\phi\in[\phi_{-},\phi_{+}]_{s,p}$ solution of
\begin{equation}
\label{HC-LYs1}
A_{\tiL}\phi + f_{\biw}(\phi) = 0.
\end{equation}
To establish existence results for weak solutions to the 
Hamiltonian constraint equation using fixed-point arguments, we will
rely on the existence of generalized (weak) sub- and super-solutions 
(sometimes called barriers) which will be derived later 
in~\S\ref{sec:barriers}.
Let us recall the definition of sub- and super-solutions in the following, 
in a slightly generalized form that will be necessary in our study of the 
coupled system.

A function $\phi_{-}\in (0,\infty)\cap W^{s,p}$ is called a {\bf sub-solution} of \eqref{WF-LYs2} iff the function
$\phi_{-}$ satisfies the inequality
\begin{equation}
\label{WF-Sb} 
A_{\tiL}\phi_{-} + f_{\biw}(\phi_{-})\leqs0,
\end{equation}
for some $a_{\biw}\in W^{s-2,p}$.
A function $\phi_{+} \in (0,\infty)\cap W^{s,p}$ is called a {\bf super-solution} of \eqref{WF-LYs2} iff the function $\phi_{+}$ satisfies the inequality
\begin{equation}
\label{WF-Sp}
A_{\tiL}\phi_{+} + f_{\biw}(\phi_{+})\geqs0,
\end{equation}
for some $a_{\biw}\in W^{s-2,p}$.
We say a pair of sub- and super-solutions is {\em compatible} if they satisfy
\begin{equation}
   \label{eqn:compat}
0 < \phi_- \leqs \phi_+ < \infty,
\end{equation}
so that the interval $[\phi_-,\phi_+] \cap W^{s,p}$ 
is both nonempty and bounded. 

We now turn to the construction of the fixed-point mapping 
$T : U \times \mathcal{R}(S) \to X$ for the Hamiltonian constraint 
and its properties.
There are a number of possibilities for defining $T$; the requirements
are 
(1) that every fixed-point of $T$ must be a solution to the Hamiltonian
constraint;
(2) $T$ must be a continuous map from its domain to its range;
and
(3) $T$ must be invariant on a non-empty, convex, closed, bounded
    subset $U$ of an ordered Banach space $Z$, with $X \hookrightarrow Z$
    compact.
It will be sufficient 
to define $T$ using a variation of the Picard iteration as follows.
Due to the presence of the non-trivial kernel of the operator $A_L$, which
is a consequence of working with a closed manifold, we must introduce
a shift into the Hamiltonian constraint equation in order to construct
$T$ with the required properties.

\begin{lemma}
{\bf (Properties of the map $T$)}
\label{l:shift}
In the above described setting, assume that $p\in(\frac32,\infty)$ and $s\in(\frac3p,\infty)\cap[1,3]$.
With $a_0\in W^{s-2,p}_{+}$ satisfying $a_0\neq0$, and $\psi\in W^{s,p}_{+}$, let $a_s=a_0+a_{\tbw}\psi\in W^{s-2,p}$.
Fix the functions $\phi_{-},\phi_{+}\in W^{s,p}$ such that $0<\phi_{-}\leqs\phi_{+}$, and define the shifted operators
\begin{align}
\label{HC-def-As}
A_{\tiL}^s & :W^{s,p} \to W^{s-2,p},&
A^s_{\tiL}\phi
&:= A_{\tiL}\phi +  a_s\phi,\\
\label{HC-def-fs}
f_{\tbw}^s & :[\phi_{-},\phi_{+}]_{s,p} \to W^{s-2,p},&
f^s_{\tbw}(\phi)
&:= f_{\tbw}(\phi) - a_s\phi.
\end{align}
Let, for $\phi\in[\phi_{-},\phi_{+}]_{s,p}$ and $a_{\tbw}\in W^{s-2,p}$,
\begin{equation}
\label{T:HC-E-def-Tw}
T^s(\phi,a_{\tbw}) := -(A_{\tiL}^s)^{-1} f_{\tbw}^s(\phi).
\end{equation}
Then, the map $T^s : [\phi_{-},\phi_{+}]_{s,p}\times W^{s-2,p}\to W^{s,p}$ is continuous in both arguments.
Moreover, there exist $\tilde{s}\in(\frac3p,s)$ and a constant $C$ such that
\begin{equation}\label{e:T-cpt}
\|T(\phi,a_{\tbw})\|_{s,p}\leqs C \left(1+\|a_{\tbw}\|_{s-2,p}\right)\|\phi\|_{\tilde{s},p},
\end{equation}
for all $\phi\in[\phi_{-},\phi_{+}]_{s,p}$ and $a_{\tbw}\in W^{s-2,p}$.
\end{lemma}

\begin{proof}
In this proof, we denote by $C$ a generic constant that may have different values at its different occurrences.
By applying Lemma \ref{l:nem} from Appendix, for any $\tilde{s}\in(\frac3p,s]$, $s-2\in[-1,1]$ and $\frac1p\in(\frac{s-1}2\delta,1-\frac{3-s}2\delta)$ with $\delta=\frac1p-\frac{\tilde{s}-1}3$, we have
\begin{align*}
\|f_{\biw}^s(\phi)\|_{s-2,p}
&\leqs  C \left(\|a_{\tau}\|_{s-2,p}\,\|\phi_{+}^4\|_{\infty}
        + \|a_{\rho}\|_{s-2,p} \, \|\phi_{-}^{-4}\|_{\infty}\right.
\\
&\quad \left. + \|a_{\biw}\|_{s-2,p} \, (\|\phi_{-}^{-8}\|_{\infty}+\|\psi\|_{\tilde{s},p})
  + \|a_{\tiR}+a_0\|_{s-2,p}\right)\|\phi\|_{\tilde{s},p}.
\end{align*}
Let us verify if $\frac1p$ is indeed in the prescribed range.
First, we have $\delta=\frac13+\frac1p-\frac{\tilde{s}}3<\frac13$ since $\frac{\tilde{s}}3-\frac1p>0$, and taking into account $s\geqs1$,
we infer $1-\frac{3-s}2\delta\geqs1-\frac{3-1}2\frac13=\frac23$.
This shows $\frac1p<1-\frac{3-s}2\delta$ for $p>\frac32$, which is not sharp, but will be sufficient for our analysis.
For the other bound, we need $\frac1p<\frac{s-1}2\delta=\frac{s-1}{2p}-\frac{(s-1)(\tilde{s}-1)}{6}$,
or in other words, $\frac{(s-1)(\tilde{s}-1)}{6}>\frac{s-3}{2p}$.
Since $s\in[1,3]$, it is possible to choose $\tilde{s}\in(\frac3p,s]$ satisfying this inequality.

To finalize the proof of \eqref{e:T-cpt}, we note that by Lemma \ref{l:lapinv} in Appendix \ref{sec:maxprinciple}, the operator $A_{\tiL}^s$ is invertible, 
since the function $a_s$ is positive, and that by Corollary \ref{C:ell-est} also in that appendix, the inverse $(A_{\tiL}^s)^{-1}:W^{s-2,p}\to W^{s,p}$ is bounded.

The continuity of the mapping $f_{\biw}^s:[\phi_{-},\phi_{+}]_{s,p}\to W^{s-2,p}$ for any $a_{\biw}\in W^{s-2,p}$ is obtained similarly, and the continuity of
$a_{\biw}\mapsto f_{\biw}(\phi)$ for fixed $\phi\in[\phi_{-},\phi_{+}]_{s,p}$ is obvious.
Being the composition of continuous maps, $(\phi,a_{\biw})\mapsto T_{\biw}^s(\phi)$ is also continuous.
\end{proof}


The following lemma shows that by choosing the shift sufficiently large, 
we can make the map $T^s$ monotone increasing.
This result is important for ensuring that the
Picard map $T$ for the Hamiltonian constraint is invariant
on the interval $[\phi_-,\phi_+]$ defined by sub- and super-solutions.
There is an obstruction that the scalar curvature should be continuous,
which will be handled in general case by conformally transforming the metric to a metric with continuous scalar curvature and using the conformal covariance 
of the Hamiltonian constraint, cf. Section \ref{sec:proof1}.

\begin{lemma}
{\bf (Monotone increasing property of $T$)}
\label{l:shift1}
In addition to the conditions of Lemma \ref{l:shift},
let $a_{\tiR}$ be continuous and define the shift function $a_s$ by
\begin{equation*}
\label{T:HC-E-def-alpha}
a_s = \max\{1, a_{\tiR} \}
+ 3\,\frac{\phi_{+}^2}{\phi_{-}^6}\,a_{\rho}
+ 5 \,\phi_{+}^4 a_{\tau} 
+ 7\,\frac{\phi_{+}^6}{\phi_{-}^{14}}\,a_{\tbw}.
\end{equation*}
Then, for any fixed $a_{\tbw}\in W^{s-2,p}$, the map $\phi\mapsto T^s(\phi,a_{\tbw}) : [\phi_{-},\phi_{+}]_{s,p}\to W^{s,p}$ is monotone increasing.
\end{lemma}

\begin{proof}
The shifted operator $A_{\tiL}^s$ satisfies the maximum principle,
hence the inverse $(A_{\tiL}^s)^{-1}:W^{s-2,p} \to W^{s,p}$ is monotone increasing.

Now we will show that the operator $f_{\biw}^s$ is monotone
decreasing. Given any two functions $\phi_2$, $\phi_1\in
[\phi_{-},\phi_{+}]_{s,p}$ with $\phi_2\geqs\phi_1$, we have
\begin{multline*}
f_{\biw}^s(\phi_2) - f_{\biw}^s(\phi_1)
= f_{\biw}(\phi_2) -f_{\biw}(\phi_1)
- a_s [\phi_2-\phi_1]\\
= a_{\tau} 
\bigl[ \phi_2^5 - \phi_1^5 \bigr]
+ a_{\tiR} [\phi_2 - \phi_1]
- a_s [\phi_2-\phi_1]
 - a_{\rho} \bigl[ \phi_2^{-3} - \phi_1^{-3}\bigr]
- a_{\biw} \bigl[ \phi_2^{-7} - \phi_1^{-7}\bigr].
\end{multline*}
The inequalities \eqref{HC-phin}, the condition $0<\phi_1\leqs\phi_2$, and the choice of $a_s$ imply
\[
f_{\biw}^s(\phi_2) - f_{\biw}^s(\phi_1) \leqs 0,
\]
which establishes that $f_{\biw}^s$ is monotone decreasing.

Both the operator $(A_{\tiL}^s)^{-1}$ and the
map $-f_{\biw}^s$ are monotone increasing, therefore the operator
$T^s(\cdot,a_{\biw})$ is also monotone increasing.
\end{proof}

\begin{lemma}
{\bf (Barriers for $T$ and the Hamiltonian constraint)}
\label{l:shiftsubsup}
Let the conditions of Lemma \ref{l:shift1} hold, with $\phi_{-}$ and $\phi_{+}$ sub- and super-solutions of the Hamiltonian constraint equation \eqref{HC-LYs1}, respectively.
Then, we have
$T^s(\phi_{+},a_{\tbw})\leqs\phi_{+}$ and $T^s(\phi_{-},a_{\tbw})\geqs\phi_{-}$.
\end{lemma}

\begin{proof}
We have
\begin{equation*}
\phi_{+}-T^s(\phi_{+},a_{\biw})
=(A_{\tiL}^s)^{-1}\bigl[A_{\tiL}^s\phi_{+}+f_{\biw}^s(\phi_{+})\bigr],
\end{equation*}
which is nonnegative since $\phi_{+}$ is a super-solution and $(A_{\tiL}^s)^{-1}$ is linear and monotone increasing.
The proof of the other inequality is completely analogous.
\end{proof}

Since we are no longer using normal order cones, 
our non-empty, convex, closed interval $[\phi_-,\phi_+]_{s,p}$ 
is not necessarily bounded as a subset of $W^{s,p}$.
Therefore, we also need {\em a priori} bounds in the norm on $W^{s,p}$
to ensure the Picard iterates stay inside the intersection of the interval 
with the closed ball $\overline{B}_M$ in $W^{s,p}$ of radius $M$, centered at the origin.
We first establish a lemma to this effect that will be useful for 
both the non-CMC and CMC cases.

\begin{lemma}
{\bf (Invariance of $T$ on the ball $\overline{B}_M$)}
\label{T:HC-ball-gen}
Let the conditions of Lemma \ref{l:shift} hold, and let $a_{\tbw}\in W^{s-2,p}$.
Then, for any $\tilde{s}\in(\frac3p,s]$ and for some $t\in(\frac3p,\tilde{s})$ there exists a closed ball $\overline{B}_M\subset W^{\tilde{s},p}$ 
of radius $M=\cO\left([1+\|a_{\tbw}\|_{s-2,p}]^{\tilde{s}/(\tilde{s}-t)}\right)$,
such that
\begin{equation*}
\phi \in [\phi_-,\phi_+]_{\tilde{s},p}\cap \overline{B}_M
\quad\Rightarrow\quad
T^s(\phi,a_{\tbw})\in\overline{B}_M.
\end{equation*}
\end{lemma}

\begin{proof}
From Lemma \ref{l:shift}, there exist $t\in(\frac3p,\tilde{s})$ and $K>0$ such that
$$
\|T^s(\phi,a_{\biw})\|_{\tilde{s},p}\leqs K(1+\|a_{\biw}\|_{s-2,p})\|\phi\|_{t,p}, \qquad
\forall\phi \in [\phi_-,\phi_+]_{\tilde{s},p}.
$$
For any $\varepsilon>0$, the norm $\|\phi\|_{t,p}$ can be bounded by the interpolation estimate
$$
\|\phi\|_{t,p}\leqs\varepsilon\|\phi\|_{\tilde{s},p}+C\varepsilon^{-t/(\tilde{s}-t)}\|\phi\|_p,
$$
where $C$ is a constant independent of $\varepsilon$.
Since $\phi$ is bounded from above by $\phi_+$, $\|\phi\|_p$ is bounded uniformly,
and now demanding that $\phi\in\overline{B}_M$, we get
\begin{equation}\label{e:Tbound}
\|T^s(\phi,a_{\biw})\|_{\tilde{s},p}\leqs K[1+\|a_{\biw}\|_{s-2,p}]\left(M\varepsilon+C\varepsilon^{-t/(\tilde{s}-t)}\right),
\end{equation}
with possibly different constant $C$.
Choosing $\varepsilon$ such that $2\varepsilon K[1+\|a_{\biw}\|_{s-2,p}]=1$ and setting $M=2KC[1+\|a_{\biw}\|_{s-2,p}]\varepsilon^{-t/(\tilde{s}-t)}$, we can ensure that
the right hand side of \eqref{e:Tbound} is bounded by $M$.
\end{proof}

\section{Barriers for the Hamiltonian constraint}
   \label{sec:barriers}

The results developed in~\S\ref{sec:Hamiltonian} for a particular
fixed-point map $T$ for analyzing the Hamiltonian constraint equation 
and the coupled system rely on the existence
of generalized (weak) sub- and super-solutions, or barriers.
There, the Hamiltonian constraint was studied in isolation from the 
momentum constraint, and these generalized barriers only needed to
satisfy the conditions given at the beginning of~\S\ref{sec:Hamiltonian}
for a given fixed function $\biw$ appearing as a source term in the
nonlinearity of the Hamiltonian constraint.
Therefore, these types of barriers are sometimes referred to as 
{\em local barriers}, in that the coupling to the momentum constraint
is ignored. In order to establish existence results for the
coupled system in the non-CMC case, it will be critical that the
sub- and super-solutions satisfy one additional property that now
reflects the coupling, giving rise to the term {\em global barriers}. 
It will be useful now to define this global property precisely.

\begin{definition}
  \label{D:barriers}
A sub-solution $\phi_{-}$ is called {\bf global} iff it is a
sub-solution of \eqref{WF-LYs2} for all vector fields $\tbw_\phi$
solution of \eqref{WF-LYm2} with source function
$\phi\in[\phi_{-},\infty)\cap W^{s,p}$. A super-solution
$\phi_{+}$ is called {\bf global} iff it is a super-solution of
\eqref{WF-LYs2} for all vector fields $\tbw_\phi$ solution of
\eqref{WF-LYm2} with source function $\phi\in (0,\phi_{+}]\cap W^{s,p}$. A pair $\phi_-\leqs\phi_+$ of sub- and super-solutions is
called an {\bf admissible pair} if $\phi_-$ and $\phi_+$ are sub-
and super-solutions of \eqref{WF-LYs2} for all vector fields
$\tbw_\phi$ of \eqref{WF-LYm2} with source function $\phi\in
[\phi_{-},\phi_{+}]\cap W^{s,p}$.
\end{definition}

It is obvious that if $\phi_-$ and $\phi_+$ are respectively global
sub- and super-solutions, then the pair $\phi_-,\phi_+$ is
admissible in the sense above, provided they satisfy the compatibility
condition~\eqref{eqn:compat}.

Below we give a number of (local and global) sub- and super-solution 
constructions for closed manifolds; analogous constructions for compact 
manifolds with boundary are given in~\cite{mHgNgT08b}.
These constructions are based on generalizing known constant 
sub- and super-solution constructions given
previously in the literature for closed manifolds.
On one hand, the generalized global sub-solution constructions appearing
here and in~\cite{mHgNgT08b} do not require the near-CMC condition, 
inheriting this property from the known sub-solutions from literature on 
which they are based.
However, on the other hand, all previously known global super-solutions for 
the Hamiltonian constraint equation have required the near-CMC condition.

Here and in~\cite{mHgNgT07a,mHgNgT08b}, one of our primary interests is in 
developing existence results for weak (and strong) non-CMC solutions to the 
coupled system which are
free of the near-CMC assumption. This assumption had appeared in two
distinct places in all prior literature on this
problem~\cite{jIvM96,pAaCjI07}; the first assumption appears in the
construction of a fixed-point argument based on strict
$k$-contractions, and the second assumption appears in the
construction of global super-solutions. 
Here and in~\cite{mHgNgT07a,mHgNgT08b}, an alternative fixed-point framework 
based on compactness arguments rather than $k$-contractions is used to remove 
the first of these near-CMC assumptions. In this section, we give some new
constructions of global super-solutions that are free of the
near-CMC assumption, along with some compatible sub-solutions. These
sub- and super-solution constructions are needed (without their
global property) for the existence result for the Hamiltonian
constraint (Theorem~\ref{T:main3}), and they are also needed (now
with their global property) for the general fixed-point result for
the coupled system (Theorem~\ref{T:FIXPT2}), leading to our two
main non-CMC results (Theorems~\ref{T:main1} and Theorem~\ref{T:main2}).
The super-solutions in
Lemmata \ref{L:HC-Sp}(b) and \ref{L:global-super} appear to be the
first such near-CMC-free constructions, and provide the second key
piece of the puzzle we need in order to establish non-CMC results
through Theorem~\ref{T:FIXPT2} without the near-CMC condition.

Throughout this section, we will assume that the background metric $h$ belongs to $W^{s,p}$ with $p\in(1,\infty)$ and $s\in(\frac3p,\infty)\cap(1,2]$.
Recall that $r=\frac{3p}{3+(2-s)p}$, so that the
continuous embedding $L^{r}\hookrightarrow W^{s-2,p}$ holds.
Given a symmetric two-index tensor $\sigma\in L^{2r}$ and a vector
field $\biw\in\biW^{1,2r}$, introduce the functions
$a_{\sigma}=\frac18\sigma^2\in L^r$ and $a_{\cL\biw}=\frac18(\cL\biw)^2\in L^{r}$.
Note that under these conditions $a_{\biw}$ belongs to $L^r\hookrightarrow W^{s-2,2}$, and that if $a_{\sigma},a_{\cL\biw}\in L^\infty$ we have the pointwise estimate
\begin{equation*}
a_{\biw}^{\tiwedge}\leqs 2a_{\sigma}^{\tiwedge}+2a_{\cL\biw}^{\tiwedge}.
\end{equation*}
Here and in what follows, given any scalar function $u\in L^{\infty}$,
we use the notation
\[
u^{\tiwedge}:= \mbox{ess~sup}\, u,\qquad
u^{\tivee}:= \mbox{ess~inf}\, u.
\]
In some places we will assume that when the vector field $\biw\in\biW^{1,2r}$ is given by the solution of the momentum constraint equation \eqref{WF-LYm2} (or \eqref{MC-LYm1}) with the source term $\phi\in W^{s,p}$,
\begin{equation}
\label{CS-aLw-bound}
a_{\cL\biw}^{\tiwedge} \leqs \ttk(\phi):=
\ttk_1 \, \|\phi\|_{\infty}^{12} + \ttk_2,
\end{equation}
with some positive constants $\ttk_1$ and $\ttk_2$. 
We can verify this assumption e.g. when the conditions of
Corollary \ref{T:MC-E-aw} are satisfied, since from Corollary \ref{T:MC-E-aw} we would get
\[
a_{\cL\biw}^{\tiwedge} =
\|\cL\biw\|_{\infty}^2 \leqs
C^2\left(\|\phi\|_{\infty}^6\|\bib_{\tau}\|_{z}+
\|\bib_{j}\|_{e-2,q}\right)^2,
\]
giving the bound \eqref{CS-aLw-bound} with the constants
\begin{align}\label{e:k1k2}
\ttk_1 = 2 C^2 \|\bib_{\tau}\|_{z}^2,
\qquad\textrm{and}\qquad
\ttk_2 = 2C^2\|\bib_{j}\|_{e-2,q}^2.
\end{align}

\subsection{Constant barriers}
   \label{sec:constant}
Now we will present some global sub- and super-solutions for the
Hamiltonian constraint equation \eqref{WF-LYs2} which are constant functions.
The proofs essentially follow the arguments in~\cite{mHgNgT08b} for
the case of compact manifolds with boundary.

\begin{lemma}{\bf (Global super-solution)}
\label{L:HC-Sp}
Let $(\cM,h)$ be a 3-dimensional, smooth, closed
Riemannian manifold with metric $h \in W^{s,p}$.
Assume that the estimate \eqref{CS-aLw-bound} holds for the 
solution of the momentum constraint equation, 
and assume that $a_\rho,a_\sigma\in L^\infty$ and that $a_{\tiR}$ 
is uniformly bounded from below.
With the parameter $\varepsilon>0$ to be chosen later,
define the rational polynomial
$$
q_{\varepsilon}(\chi)=
(a_{\tau}^{\tivee}-\ttK_{1\varepsilon}) \, \chi^5
+ a_{\tiR}^{\tivee} \, \chi
- a_{\rho}^{\tiwedge}\,\chi^{-3}
- \ttK_{2\varepsilon} \chi^{-7},
$$
where $\ttK_{1\varepsilon}:=(1+\frac{1}\varepsilon)\ttk_1$ and $\ttK_{2\varepsilon}:=(1+\varepsilon)a_{\sigma}^{\tiwedge}+(1+\frac{1}{\varepsilon})\ttk_2$.
We distinguish the following two cases:

(a) In case $\ttk_1<a_{\tau}^{\tivee}$, choose $\displaystyle\varepsilon>\frac{\ttk_1}{a_{\tau}^{\tivee}-\ttk_1}$.
If $q_{\varepsilon}$ has a root, let $\phi_+=\phi_1(a_{\tau}^{\tivee}-\ttK_{1\varepsilon},a_{\tiR}^{\tivee},a_{\rho}^{\tiwedge},\ttK_{2\varepsilon})$ be the largest positive root of $q$, and if $q$ has no positive roots, let $\phi_+=1$.
Now, the constant $\phi_+$ is a global super-solution of the Hamiltonian constraint equation (\ref{WF-LYs2}).

(b) In case $\ttk_1\geqs a_{\tau}^{\tivee}$, choose $\varepsilon>0$.
In addition, assume that $a_{\tiR}^{\tivee}>0$ and that both
$a_{\rho}^{\tiwedge}$ and $\ttK_{2\varepsilon}$ are sufficiently
small, such that $q$ has two positive roots. Then, the largest root
$\phi_+=\phi_2(a_{\tau}^{\tivee}-\ttK_{1\varepsilon},a_{\tiR}^{\tivee},a_{\rho}^{\tiwedge},\ttK_{2\varepsilon})$
of $q$ is a super-solution of the Hamiltonian constraint equation
(\ref{WF-LYs2}).
\end{lemma}

\begin{proof}
We look for a super-solution among the constant functions.  Let $\chi$ be any positive constant.
Then we have
\begin{align*}
 f(\chi,\biw)
&= a_{\tau} \chi^{5}
+ a_{\tiR} \chi
- a_{\rho}\chi^{-3}
- a_{\biw} \chi^{-7}
\geqs a_{\tau}^{\tivee} \chi^5
+ a_{\tiR}^{\tivee} \chi
- a_{\rho}^{\tiwedge} \chi^{-3}
- a_{\biw}^{\tiwedge} \chi^{-7}.
\end{align*}
Given any $\varepsilon>0$, the
inequality $2|\sigma_{ab} (\cL\biw)^{ab}|\leqs\varepsilon \sigma^2 +
\frac1\varepsilon(\cL\biw)^2$ implies that
\[\textstyle
8a_{\biw} = \sigma^2+(\cL \biw)^2 + 2 \sigma_{ab}(\cL\biw)^{ab}
\leqs (1+\varepsilon)\,\sigma^2
+ ( 1+\frac{1}{\varepsilon} ) \,(\cL \biw)^2,
\]
hence, taking into account \eqref{CS-aLw-bound}, for any
$\biw\in\biW^{1,2r}$ that is a solution of the momentum constraint
equation \eqref{WF-LYm2} with any source term $\phi\in(0,\chi]$, the
constant $a_{\biw}^{\tiwedge}$ must fulfill the inequality
\begin{equation}
\label{CS-aw-bound}\textstyle
a_{\biw}^{\tiwedge}
\leqs (1+\varepsilon)a_{\sigma}^{\tiwedge}+(1+\frac1\varepsilon)a_{\cL\biw}^{\tiwedge}
\leqs \ttK_{1\varepsilon}\|\phi\|_{\infty}^{12}+\ttK_{2\varepsilon}.
\end{equation}
Thus, for any constant $\chi>0$ and all $\phi\in(0,\chi]$, it holds that
\begin{equation*}
\begin{split}
f(\chi,\biw_{\phi})
&\geqs
a_{\tau}^{\tivee} \chi^5
+ a_{\tiR}^{\tivee} \chi
- a_{\rho}^{\tiwedge} \chi^{-3}
- \left( \ttK_{1\varepsilon} \, \|\phi\|_{\infty}^{12} + \ttK_{2\varepsilon}\right)
 \chi^{-7}\\
&\geqs
B_{\varepsilon} \chi^5
+ a_{\tiR}^{\tivee} \chi
- a_{\rho}^{\tiwedge} \chi^{-3}
- \ttK_{2\varepsilon} \chi^{-7},
\end{split}
\end{equation*}
where $B_{\varepsilon}:=a_{\tau}^{\tivee}-\ttK_{1\varepsilon}$.
Introduce the rational polynomial on $\chi$ given by
\begin{equation}
\label{HC-def-q}
q_{\varepsilon}(\chi) :=
B_{\varepsilon} \chi^5
+ a_{\tiR}^{\tivee} \chi
- a_{\rho}^{\tiwedge} \chi^{-3}
- \ttK_{2\varepsilon} \chi^{-7}.
\end{equation}
We calculate the first and second derivative of $q_{\varepsilon}$ as
\begin{equation}
\begin{split}
q_{\varepsilon}'(\chi)
&=
5B_{\varepsilon} \chi^4
+ a_{\tiR}^{\tivee}
+3 a_{\rho}^{\tiwedge} \chi^{-4}
+7 \ttK_{2\varepsilon} \chi^{-8},\\
q_{\varepsilon}''(\chi)
&=
20B_{\varepsilon} \chi^3
-12 a_{\rho}^{\tiwedge} \chi^{-5}
-56 \ttK_{2\varepsilon} \chi^{-9}.
\end{split}
\end{equation}

Consider the case (a). In this case, because of the choice
$\varepsilon>\frac{\ttk_1}{a_{\tau}^{\tivee}-\ttk_1}$,
we have $B_{\varepsilon}>0$, and so $q_{\varepsilon}(\chi)>0$ for
sufficiently large $\chi$, and $q_{\varepsilon}$ is increasing. The
function $q_{\varepsilon}$ has no positive root only if
$a_{\rho}^{\tiwedge}=\ttK_{2\varepsilon}=0$. So if $q_{\varepsilon}$
has no positive root, $q_{\varepsilon}(\chi)\geqs0$ for all
$\chi\geqs0$. If $q_{\varepsilon}$ has at least one positive root,
denoting by $\phi_1$ the largest positive root, $q(\chi)\geqs0$ for
all $\chi\geqs\phi_1$. Recalling now that any constant $\chi$
satisfies $A_{\tiL}\chi=0$, we conclude that
\[
A_{\tiL}\chi + f(\chi,\biw_\phi) \geqs 0 \qquad
\qquad \forall \,\chi\geqs\phi_1,
\,\forall \,\phi\in(0,\chi],
\]
implying that $\phi_{+}$ is a global super-solution of the
Hamiltonian constraint (\ref{WF-LYs1}).

For the case (b), since $B_{\varepsilon}<0$ and
$a_{\rho}^{\tiwedge}$ and $\ttK_{2\varepsilon}$ are nonnegative, the
first derivative $q_{\varepsilon}'(\chi)$ is strictly decreasing for
$\chi>0$, and since $q_{\varepsilon}'(\phi)>0$ for sufficiently
small $\chi>0$ and $q_{\varepsilon}'(\chi)<0$ for sufficiently large
$\chi>0$, the derivative $q_{\varepsilon}'$ has a unique positive
root, at which the polynomial $q_{\varepsilon}$ attains its maximum
over $(0,\infty)$. This maximum is positive if both
$a_{\rho}^{\tiwedge}$ and $\ttK_{2\varepsilon}$ are sufficiently
small, and hence the polynomial $q_{\varepsilon}$ has two positive
roots $\phi_1\leqs\phi_2$. Similarly to the above we conclude that
\[
A_{\tiL}\chi + f(\chi,\biw_\phi) \geqs 0
\qquad \forall \,\chi\in[\phi_1,\phi_2], \,\forall \,\phi\in(0,\chi],
\]
implying that $\phi_{+}$ is a global super-solution of
the Hamiltonian constraint (\ref{WF-LYs1}).
\end{proof}

Case (a) of the above lemma has the condition $\ttk_1<a_{\tau}^{\tivee}$, 
which is the near-CMC condition.
This condition seems to be present in all non-CMC results to date.
The above condition also requires 
that the extrinsic mean curvature $\tau$ is nowhere zero.
Noting that there are solutions even for $\tau\equiv0$ in some 
cases (cf.~\cite{jI95}), the condition $\inf\tau>0$ appears 
as a rather strong restriction.
We see that case (b) of the above lemma removes this restriction, 
in exchange for the smallness conditions on $\rho$, $j$, and $\sigma$.
We also need the scalar curvature to be strictly positive, 
which condition is relaxed in the next subsection to allow 
any metric in the positive Yamabe class.

In the following lemma, we list some constant sub-solutions.
They impose considerable restrictions on the allowable data,
which is the main reason to consider non-constant sub-solutions 
in the next subsection.

\begin{lemma}{\bf(Global sub-solution)}
\label{L:HC-GSb}
Let $(\cM,h)$ be a 3-dimensional, smooth, closed
Riemannian manifold with metric $h \in W^{s,p}$.
Assume that $a_\tau\in L^\infty$ and that $a_{\tiR}$ is uniformly bounded from above.
We distinguish the following three cases.

(a) If $a_{\tiR}^{\tiwedge}<0$, then the unique positive root of the polynomial
$$
q(\chi)=
a_{\tau}^{\tiwedge} \,\chi^4
+ a_{\tiR}^{\tiwedge},
$$
is a global sub-solution of \eqref{WF-LYs2}.

(b) If $a_{\rho}^{\tivee}>0$, then the unique positive root of the polynomial
$$
q_\rho(\chi)=
a_{\tau}^{\tiwedge} \,\chi^8
+ \max\{1,a_{\tiR}^{\tiwedge}\}\, \chi^4
- a_{\rho}^{\tivee},
$$
is a global sub-solution of \eqref{WF-LYs2}.

(c) Let $\phi_{+}>0$ be a global super-solution of the Hamiltonian
constraint. Let $a_{\sigma}^{\tivee}>\ttk(\phi_{+})$, where $\ttk$
is as in \eqref{CS-aLw-bound}. Then, with some
$\varepsilon\in(\ttk(\phi_{+})/a_{\sigma}^{\tivee},1)$, the unique
positive root $\phi_+$ of the polynomial
$$
q_\sigma(\chi)=
a_{\tau}^{\tiwedge} \,\chi^{12}
+ \max\{1,a_{\tiR}^{\tiwedge}\}\, \chi^8
- \ttK_{\varepsilon},
$$
where $\ttK_{\varepsilon}:=(1-\varepsilon)a_{\sigma}^{\tivee}-\bigl(\frac1\varepsilon-1\bigr)\ttk(\phi_{+})$,
is a global sub-solution of \eqref{WF-LYs2}.
\end{lemma}

\begin{proof}
For the proof of (a,b), see e.g.~\cite{mHgNgT08b}. We give a
proof of (c) here.

Let $\chi>0$ be any constant function, and let $\biw\in\biW^{1,2r}$.
Then we have
\begin{equation}
\label{prf-sigma-1}
\begin{split}
f(\chi,\biw)
&= a_{\tau} \chi^{5}
+ a_{\tiR} \chi
- a_{\rho} \chi^{-3}
- a_{\biw} \chi^{-7}
\leqs
a_{\tau}^{\tiwedge} \chi^5
+ a_{\tiR}^{\tiwedge} \chi
- a_{\biw}^{\tivee} \chi^{-7}\\
&\leqs
a_{\tau}^{\tiwedge} \chi^5
+ C \chi
- a_{\biw}^{\tivee} \chi^{-7},
\end{split}
\end{equation}
where we have used that $a_{\rho}$ is nonnegative, and
introduced the constant $C=\max\{1,a_{\tiR}^{\tiwedge}\}$. Given any
$\varepsilon>0$, the inequality $2|\sigma_{ab}
(\cL\biw)^{ab}|\leqs\varepsilon \sigma^2 + \frac1\varepsilon(\cL\biw)^2$
implies that
\[\textstyle
8a_{\biw} = \sigma^2+(\cL\biw)^2 + 2 \sigma_{ab}(\cL\biw)^{ab}
\geqs (1-\varepsilon)\,\sigma^2
- ( \frac{1}{\varepsilon}-1 ) \,(\cL\biw)^2,
\]
hence, taking into account \eqref{CS-aLw-bound},
for any $\biw\in\biW^{1,2r}$ that is a solution of the momentum constraint equation \eqref{WF-LYm2} with any source term $\phi\in(0,\phi_{+}]$,
the constant $a_{\biw}^{\tivee}$ must fulfill the inequality
\begin{equation*}\textstyle
a_{\biw}^{\tivee}
\geqs (1-\varepsilon)a_{\sigma}^{\tivee}-(\frac1\varepsilon-1)a_{\cL\biw}^{\tiwedge}
\geqs (1-\varepsilon)a_{\sigma}^{\tivee}-(\frac1\varepsilon-1)\ttk(\phi_{+})=:\ttK_{\varepsilon}.
\end{equation*}
We use the above estimate in \eqref{prf-sigma-1} to get, for any $\biw\in\biW^{1,2r}$ that
is a solution of the momentum constraint equation \eqref{WF-LYm2}
with any source term $\phi\in(0,\phi_{+}]$
\begin{equation*}
\begin{split}
f(\chi,\biw)
&\leqs
a_{\tau}^{\tiwedge} \chi^5
+ C \chi
-\ttK_{\varepsilon} \chi^{-7}.
\end{split}
\end{equation*}
Because of the choice
$\ttk(\phi_{+})/a_{\sigma}^{\tivee}<\varepsilon<1$, we have
$\ttK_{\varepsilon}>0$. So with the unique positive root $\chi_{*}$
of
\[
q_{\sigma}(\chi) :=
a_{\tau}^{\tiwedge} \,\chi^5
+ C\, \chi
- \ttK_{\varepsilon}\, \chi^{-7},
\]
we have $q_{\sigma}(\chi)\leqs0$ for any constant
$\chi\in(0,\chi_{*}]$, establishing the proof.
\end{proof}

\subsection{Non-constant barriers}
   \label{sec:nonconstant}

All global super-solutions found to date appear to require the near-CMC
condition; Lemma \ref{L:HC-Sp}(b) avoids the near-CMC condition,
but it requires the scalar curvature to be strictly positive.
The following lemma extends this result to arbitrary metrics in the positive 
Yamabe class $\cY^+(\cM)$.

\begin{lemma}
{\bf (Global super-solution $h\in\cY^{+}$)}
\label{L:global-super}
Let $(\cM,h)$ be a 3-dimensional, smooth, closed Riemannian
manifold with metric $h \in W^{s,p}$ in $\cY^+(\cM)$.
Assume there exist continuous positive functions $u,\Lambda \in W^{s,p}$
that together satisfy:
\begin{equation}
  \label{eqn:yamabe}
-\Delta u+\textstyle\frac18Ru = \Lambda>0,
\quad u > 0.
\end{equation}
Let $0<\ttk_3:=u^{\tiwedge}/u^{\tivee} < \infty$, which is a trivially satisfied Harnack-type inequality.
Assume that the estimate \eqref{CS-aLw-bound} 
is satisfied for the solution of the momentum constraint equation
for two positive constants $\ttk_1$ and $\ttk_2$,
and assume that $a_\rho,a_\sigma\in L^\infty$.
If the constants 
$a_{\rho}^{\tiwedge}$, $a_{\sigma}^{\tiwedge}$, and $\ttk_2$ 
are sufficiently small, then
\begin{equation}
  \label{eqn:alpha}
\phi_+ = \beta u, \quad \beta = \left[
\frac{\Lambda^{\tivee}}
     {2\ttk_1\ttk_3^{12} (u^{\tiwedge})^5} \right]^{1/4} > 0,
\end{equation}
is a positive global super-solution to the Hamiltonian constraint equation.
\end{lemma}

\begin{proof}
Taking $\phi= \beta u$ with a constant $\beta > 0$
in~\eqref{eqn:yamabe}, gives
\begin{equation}\textstyle
-\Delta\phi+a_{\tiR}\phi=\beta(-\Delta u + \frac18R u) = \beta \Lambda.
\end{equation}
Then for any $\varphi \in C^{\infty}_+$, by using \eqref{CS-aw-bound}
with $\ttK_{1}:=2\ttk_1$ and $\ttK_{2}:=2a_{\sigma}^{\tiwedge}+2\ttk_2$,
we infer
\begin{align*}
\langle A_{\tiL}\phi + f(\phi,\biw), \varphi \rangle
&= \langle \nabla \phi, \nabla \varphi \rangle
    + \langle a_{\tiR}\phi + a_{\tau} \phi^{5} 
        - a_{\rho}\phi^{-3} - a_{\biw} \phi^{-7}, \varphi \rangle \\
&\geqs
    \langle \beta \Lambda
         + a_{\tau}^{\tivee} \phi^5
         - [\ttK_1 (\phi^{\tiwedge})^{12} + \ttK_2] \phi^{-7}
         - a_{\rho}^{\tiwedge} \phi^{-3} , \varphi \rangle \\
&\geqs
    \langle \beta \Lambda
         + [a_{\tau}^{\tivee} - \ttK_1 \ttk_3^{12}] \phi^5
         - \ttK_2 \phi^{-7}
         - a_{\rho}^{\tiwedge} \phi^{-3} , \varphi \rangle \\
&\geqs
    \langle \beta 
        G(\beta,\ttK_2,a_{\rho}), \varphi \rangle
\end{align*}
where
\begin{equation*}
G(\beta,\ttK_2,a_{\rho}) := 
\Lambda^{\tivee}
- \ttK_1 \ttk_3^{12} \beta^4 (u^{\tiwedge})^5 -
\ttK_2 \beta^{-8} (u^{\tiwedge})^{-7} - a_{\rho}^{\tiwedge}
\beta^{-4} (u^{\tiwedge})^{-3},
\end{equation*}
and where we have used the fact that 
$\phi^{\tiwedge}/\phi^{\tivee} = u^{\tiwedge}/u^{\tivee} = \ttk_3$.
Therefore, to ensure $\phi$ is a super-solution we must
now pick arguments ensuring
$G(\beta,\ttK_2,a_{\rho}) \geqs 0$.
We first pick $\beta$ as in~\eqref{eqn:alpha} giving
\begin{equation*}\textstyle
\frac{1}{2}\Lambda^{\tivee} 
   = \Lambda^{\tivee}
    - \ttK_1 \ttk_3^{12} (u^{\tiwedge})^5 \beta^4 > 0.
\end{equation*}
For this fixed $\beta$, we then pick $\ttK_2$ and $a_{\rho}^{\tivee}$,
each sufficiently small, so that
\begin{equation*}\textstyle
\frac{1}{2}\Lambda^{\tivee} - \ttK_2 \beta^{-8} (u^{\tiwedge})^{-7} -
a_{\rho}^{\tiwedge} \beta^{-4} (u^{\tiwedge})^{-3} \geqs 0.
\end{equation*}
The result then follows.
\end{proof}

\Remark
We now make some remarks about the existence of a pair of positive functions
$(u,\Lambda)$ which satisfy the hypotheses of Lemma~\ref{L:global-super}.
Let the background metric $h_{ab} \in W^{s,p}$ be in the positive Yamabe class.
Then in Theorem \ref{t:subcrit} in Appendix \ref{sec:yamabe},
for the sub-critical range $1 \leqs q < 5$
we establish the existence of a positive $u \in W^{s,p}$ and a constant $\mu_q > 0$ satisfying 
\begin{equation*}\textstyle
-8\Delta u+Ru=\mu_qu^q.
\end{equation*}
So the pair $(u,\frac18\mu_qu^q)$ readily satisfies \eqref{eqn:yamabe}.
In a sense the simplest construction of the near-CMC-free global super-solution
in Lemma~\ref{L:global-super} arises by taking $q=1$;
one is then simply using the first eigenfunction of the
conformal Laplacian to build the global super-solution.

Alternatively, one
can consider a solution to the Yamabe problem
\begin{equation*}
-8\Delta u+Ru= u^5,\quad u>0,
\end{equation*}
which exists for sufficiently smooth metrics in the positive Yamabe class, cf. \cite{jLtP87}.
This approach is taken for simplicity in~\cite{mHgNgT07a}.

In any case, note that the function $u>0$ that satisfies~\eqref{eqn:yamabe} 
is the conformal factor which transforms $h_{ab}$ into a metric with
scalar curvature $R_u=8\Lambda u^{-5}> 0$.

We remark that without the near-CMC condition, the only potentially strictly 
positive term appearing in the nonlinearity of the Hamiltonian constraint 
is the term involving the scalar curvature $R$.
Therefore, global super-solution
constructions based on the approach in Lemma~\ref{L:global-super}
are restricted to data in $\cY^{+}(\cM)$.
We extend this observation in the next lemma, which essentially says that in a nonpositive Yamabe class, there is no way to build a positive global super-solution without the near-CMC condition as long as we use a global estimate of type \eqref{CS-aLw-bound}.

\begin{lemma}
{\bf (Near-CMC condition and $a_{\biw}$ bounds)}
\label{L:global-super-limit}
Let $(\cM,h)$ be a 3-dimensional, smooth, closed Riemannian
manifold with metric $h \in W^{s,p}$ in a nonpositive Yamabe class, 
and let $a_{\tau}$ be continuous.
Let $\phi_{+}\in W^{s,p}$ with $\phi_{+}>0$ be a global super-solution to the Hamiltonian constraint
equation.
We assume that any vector field $\tbw\in\tbW^{1,2r}$ that is a solution of the momentum constraint equation with a source $\phi\leqs\phi_{+}$ satisfies the estimate
\begin{equation}
   \label{awest2}
a_{\tbw}\leqs\ttK_1 \|\phi_{+}\|_{\infty}^{12} + \ttK_2,
\end{equation}
with some positive constants $\ttK_1$ and $\ttK_2$.
Moreover, we assume that this estimate is sharp in the sense that for any $x\in\cM$ there exist an open neighborhood $U\ni x$ and a vector field $\tbw\in\tbW^{1,2r}$ a solution of the momentum constraint equation with a source $\phi\leqs\phi_{+}$, such that
\begin{equation}\label{awest-sharp}
a_{\tbw}=\ttK_1 \|\phi_{+}\|_{\infty}^{12} + \ttK_2
\qquad\textrm{in }U.
\end{equation}
Then, we have $\ttK_1\leqs\sup_{\cM} a_{\tau}$.
\end{lemma}

\begin{proof}
Since the metric is in a nonpositive Yamabe class, there exists $\tilde\varphi\in W^{2-s,p'}_{+}$ such that
$\langle\nabla\phi_{+},\nabla\tilde\varphi\rangle + \langle a_{\tiR}\phi_{+},\tilde\varphi\rangle\leqs0$.
The collection of all neighborhoods in \eqref{awest-sharp} will form an open cover of $\cM$, and let $\{U_i\}$ be one of its finite subcovers.
Let $\{\mu_i\}$ be a partition of unity subordinate to $\{U_i\}$.
Then, by writing $\tilde\varphi=\sum_i\mu_i\tilde\varphi$, we can expand the expression $\langle\nabla\phi_{+},\nabla\tilde\varphi\rangle + \langle a_{\tiR}\phi_{+},\tilde\varphi\rangle$ into 
a finite sum, which has at least one non-positive term.
Without loss of generality, let us assume $\langle\nabla\phi_{+},\nabla\varphi\rangle + \langle a_{\tiR}\phi_{+},\varphi\rangle\leqs0$ with $\varphi=\mu_i\tilde\varphi$.
With $\biw\in\biW^{1,2r}$ being a vector field that satisfies \eqref{awest-sharp} with respect to $U:=U_i$, we have
\begin{align*}
0&\leqs
	\langle\nabla\phi_{+},\nabla\varphi\rangle + \langle a_{\tiR}\phi_{+}
         + a_{\tau} \phi_{+}^5
         - a_{\biw} \phi_{+}^{-7}
         - a_{\rho} \phi_{+}^{-3},\varphi\rangle\\
&\leqs
         \langle a_{\tau} \phi_{+}^5
         - a_{\biw} \phi_{+}^{-7}
         - a_{\rho} \phi_{+}^{-3},\varphi\rangle\\
&=
         \langle a_{\tau} \phi_{+}^5
         - [\ttK_1 (\phi_{+}^{\tiwedge})^{12} + \ttK_2] \phi_{+}^{-7}
         - a_{\rho} \phi_{+}^{-3},\varphi\rangle\\
&\leqs
         ([ a_{\tau} - \ttK_1 (\phi_{+}^{\tiwedge}/\phi_{+})^{12} ] \phi_{+}^{5},\varphi).
\end{align*}
Using partitions of unity we can make the support of $\varphi$ arbitrarily small, from which we conclude that $a_{\tau} \geqs \ttK_1 (\phi_{+}^{\tiwedge}/\phi_{+})^{12}\geqs\ttK_1$ at some $x\in\cM$.
\end{proof}

All of the subsequent barrier constructions below are more or less known.
A number of the more technically sophisticated construction techniques 
we employ below were pioneered by Maxwell in~\cite{dM05}.
For completeness, we first construct local super-solutions and
then global super-solutions for the near-CMC case.

\begin{lemma}{\bf (Local super-solution)}
\label{L:local-super-near}
Let $(\cM,h)$ be a 3-dimensional, smooth, closed Riemannian manifold with metric $h \in W^{s,p}$.
Let $a_{\tau},a_{\rho},a_{\tbw}\in W^{s-2,p}_{+}$, and let one of the following conditions hold:
\begin{itemize}
\item[(a)] The metric $h$ is in a non-negative Yamabe class, $a_{\tau}\neq0$, and $a_{\rho}+a_{\tbw}\neq0$.
\item[(b)] The metric $h$ is in the positive Yamabe class, and $a_{\rho}+a_{\tbw}\neq0$.
\item[(c)] The metric $h$ is conformally equivalent to a metric with scalar curvature $-a_{\tau}\neq0$, thus in particular the metric is in the negative Yamabe class.
\end{itemize}
Then, there is a positive (local) super-solution $\phi_{+}\in W^{s,p}$ of the Hamiltonian constraint equation (\ref{WF-LYs2}).
\end{lemma}

\begin{proof}
First we prove (a) and (b).
Let  $u\in W^{s,p}$ be a (weak) solution to
\begin{equation*}\textstyle
-\Delta u+\frac18Ru = \lambda u,
\quad u > 0,
\end{equation*}
with a constant $\lambda\geqs0$, which exists by Theorem \ref{t:subcrit} in Appendix \ref{sec:yamabe},
and let $v\in W^{s,p}$ be the solution to
\begin{equation}\label{e:rhowcurv}
\langle u^2\nabla v, \nabla\varphi\rangle+\langle\lambda u^2v+a_{\tau}v,\varphi\rangle=\langle a_{\rho}+a_{\biw},\varphi\rangle,
\qquad\forall\varphi\in C^{\infty}.
\end{equation}
Since $a_{\tau},a_{\rho},a_{\biw}\in W^{s-2,p}_{+}$ with $sp>3$, we have $v\in W^{s,p}\hookrightarrow L^\infty$, and since $\lambda u^2+a_{\tau}\neq0$ and $a_{\rho}+a_{\biw}\neq0$, Lemma \ref{l:max-princ} (maximum principle) in Appendix \ref{sec:maxprinciple} implies that $v>0$.
Let us define $\phi= \beta uv\in W^{s,p}$ for a constant $\beta > 0$. Then for any $\varphi\in C^{\infty}_{+}$ we have
\begin{multline*}
\langle A_{\tiL}\phi +  f(\phi,\biw),u\varphi\rangle
=
\langle \nabla\phi,\nabla(u\varphi)\rangle + \langle a_{\tau}\phi^5 + a_{\tiR}\phi - a_{\rho}\phi^{-3}-a_{\biw}\phi^{-7},u\varphi\rangle\\
=
\beta\langle u^2\nabla v,\nabla\varphi\rangle+\langle\beta\lambda u^2v+a_{\tau}u\phi^5 - a_{\rho}u\phi^{-3}-a_{\biw}u\phi^{-7},\varphi\rangle\\
=
\langle a_{\tau}[\beta^5u^6v^5-\beta v],\varphi\rangle
+\langle a_{\rho}[\beta-\beta^{-3}u^{-2}v^{-3}],\varphi\rangle
+\langle a_{\biw}[\beta-\beta^{-7}u^{-6}v^{-7}],\varphi\rangle,
\end{multline*}
where the second line is obtained by 
\begin{equation}\label{e:auv}
\begin{split}
\langle A_{\tiL}\phi + a_{\tiR}\phi,u\varphi\rangle
&=\textstyle
\beta\langle\nabla(uv),\nabla(u\varphi)\rangle+\frac\beta8\langle Ruv,u\varphi\rangle\\
&=\textstyle
\beta\langle\nabla u,\nabla(uv\varphi)\rangle+\frac\beta8\langle Ru,uv\varphi\rangle+\beta\langle u\nabla v,u\nabla\varphi\rangle\\
&=
\beta\langle \lambda u,uv\varphi\rangle + \beta\langle u^2\nabla v,\nabla\varphi\rangle,
\end{split}
\end{equation}
and the third line is from \eqref{e:rhowcurv}.
Now, choosing $\beta>0$ sufficiently large, so that $\beta^4u^6v^5-v\geqs0$, $1-\beta^{-4}u^{-2}v^{-3}\geqs0$ and $1-\beta^{-8}u^{-6}v^{-7}\geqs0$, we ensure that $\phi$ is a super-solution.

Now, let us consider (c).
Let $u>0$ be the conformal factor which transforms $h$ into a metric with scalar curvature $\lambda=-8a_{\tau}$,
i.e., let $u\in W^{s,p}$ be a weak solution to
\begin{equation*}\textstyle
-\Delta u+\frac18Ru + a_{\tau}u^5 = 0,
\quad u > 0.
\end{equation*}
If $a_{\rho}=a_{\biw}=0$, the Hamiltonian constraint equation reduces to the above equation and we can take $u$ as a super-solution (it is even a solution).
So we can assume in the following that $a_{\rho}+a_{\biw}\neq0$.
Let $v\in W^{s,p}$ be the solution to
\begin{equation*}
\langle u^2\nabla v, \nabla\varphi\rangle+\langle a_{\tau}v,\varphi\rangle=\langle a_{\rho}+a_{\biw},\varphi\rangle,
\qquad\forall\varphi\in C^{\infty}.
\end{equation*}
Defining $\phi= \beta uv\in W^{s,p}$ for a constant $\beta > 0$, the rest of the proof proceeds superficially in the same way as the above case.
\end{proof}

\begin{lemma}{\bf (Near-CMC global super-solution)}
\label{L:global-super-near} 
Let $(\cM,h)$ be a 3-dimensional, smooth, closed Riemannian manifold with metric $h \in W^{s,p}$.
Let $a_{\tau},a_{\rho}\in W^{s-2,p}_{+}$ and $a_{\sigma}\in L^\infty_{+}$, and let one of the following conditions hold:
\begin{itemize}
\item[(a)] 
The metric $h$ is in a non-negative Yamabe class, $a_{\tau}\neq0$, and $a_{\rho}+a_{\sigma}\neq0$.
Let  $u\in W^{s,p}$ and $v\in W^{s,p}$ be the solutions to
\begin{equation}\label{e:rhosigmacurva}
\begin{split}\textstyle
-\Delta u+\frac18Ru &= \lambda u,\\
-\nabla(u^2\nabla v)+(\lambda u^2+a_{\tau})v &= a_{\rho}+a_{\sigma}.
\end{split}
\end{equation}
with a constant $\lambda\geqs0$.
\item[(b)] 
The metric $h$ is conformally equivalent to a metric with scalar curvature $-a_{\tau}\neq0$, thus in particular the metric is in the negative Yamabe class.
Let $u\in W^{s,p}$ and $v\in W^{s,p}$ be the solutions to
\begin{equation}\label{e:rhosigmacurvb}
\begin{split}\textstyle
-\Delta u+\frac18Ru + a_{\tau}u^5 &= 0,\\
-\nabla (u^2\nabla v) +  a_{\tau}v &=  a_{\rho}+a_{\sigma}.
\end{split}
\end{equation}
\end{itemize}
Assume that the estimate \eqref{CS-aLw-bound} holds for the momentum constraint equation,
and let $\ttk_1<a_{\tau}^{\tivee}(\frac{\min uv}{\max uv})^{12}$.
Then, for any sufficiently large constant $\beta>0$, $\phi_{+}=\beta uv$ is a global super-solution of the Hamiltonian constraint equation (\ref{WF-LYs2}).
\end{lemma}

\begin{proof}
We give a proof of (a). The proof of (b) is similar.
Proceeding as in the proof of the preceding lemma, for any $\varphi\in C^{\infty}_{+}$ we have
\begin{align*}
\langle A_{\tiL}\phi &+  f(\phi,\biw),u\varphi\rangle
=
\langle \nabla\phi,\nabla(u\varphi)\rangle + \langle a_{\tau}\phi^5 + a_{\tiR}\phi - a_{\rho}\phi^{-3}-a_{\biw}\phi^{-7},u\varphi\rangle\\
&=
\beta\langle u^2\nabla v,\nabla\varphi\rangle+\langle\beta\lambda u^2v+a_{\tau}u\phi^5 - a_{\rho}u\phi^{-3}-a_{\biw}u\phi^{-7},\varphi\rangle\\
&\geqs
\beta\langle u^2\nabla v,\nabla\varphi\rangle+\langle\beta\lambda u^2v+a_{\tau}u\phi^5 - a_{\rho}u\phi^{-3}-2[a_{\sigma}+a_{\cL\biw}]u\phi^{-7},\varphi\rangle\\
&=
\langle a_{\rho}[\beta-\beta^{-3}u^{-2}v^{-3}],\varphi\rangle
+\langle a_{\sigma}[\beta-2\beta^{-7}u^{-6}v^{-7}],\varphi\rangle\\
&\quad
+ \langle a_{\tau}[\beta^5u^6v^5-\beta v] - 2a_{\cL\biw}u\phi^{-7},\varphi\rangle.
\end{align*}
Then, choosing $\beta$ sufficiently large, and by using \eqref{CS-aLw-bound}, with $\theta=uv$ we infer
\begin{align*}
A_{\tiL}\phi+f(\phi,\biw)
&\geqs
    [a_{\tau}^{\tivee} (\theta^{\tivee})^5 - 2\ttk_1 (\theta^{\tiwedge})^{12}(\theta^{\tivee})^{-7}]\beta^5
    - p(\beta),
\end{align*}
where $p(\beta)=a_{\tau}(v^{\tiwedge}/u^{\tivee})\beta + 2\ttk_2(\theta^{\tivee})^{-7} \beta^{-7}$.
Now, if we have $\ttk_1<\frac12a_{\tau}^{\tivee}(\theta^{\tivee}/\theta^{\tiwedge})^{12}$,
then choosing $\beta$ large enough, we ensure that $\phi$ is a super-solution.
If we proceeded as in the proof of Lemma \ref{L:HC-Sp}, we could remove the factor $\frac12$ from the condition $\ttk_1<\frac12a_{\tau}^{\tivee}(\theta^{\tivee}/\theta^{\tiwedge})^{12}$; however, we omit it for clarity.
\end{proof}

We now also give some examples of non-constant global
sub-solutions $\phi_-$ which are compatible with $\phi_+$ above
in the sense that $0 < \phi_- \leqs \phi_+$.
Such a pair of compatible sub- and super-solutions are
needed to establish existence of solutions to the 
individual Hamiltonian constraint (Theorem~\ref{T:main3}), 
and are also needed again to establish existence of solutions to the 
coupled system (Theorems~\ref{T:main1} and~\ref{T:main2}).

\begin{lemma}
{\bf (Global sub-solution $h\not\in\cY^{-},\,\rho\not\equiv0$)}
\label{L:global-sub}
Let $(\cM,h)$ be a 3-dimensional, smooth, closed
Riemannian manifold with metric $h \in W^{s,p}$ 
in a non-negative Yamabe class.
Let $a_{\rho},a_{\tau}\in W^{s-2,p}_{+}\setminus\{0\}$.
Then, there exists a positive scalar $\phi_{-}\in W^{s,p}$ 
such that for any constant $\beta\in(0,1]$,
$\beta\phi_{-}$ is a global sub-solution of the
Hamiltonian constraint equation.
\end{lemma}

\begin{proof}
Let  $u\in W^{s,p}$ be a (weak) solution to
\begin{equation*}\textstyle
-\Delta u+\frac18Ru = \lambda u,
\quad u > 0,
\end{equation*}
with a constant $\lambda\geqs0$, which exists by Theorem \ref{t:subcrit} in Appendix \ref{sec:yamabe},
and let $v\in W^{s,p}$ be the solution to
\begin{equation}\label{e:rhocurv}
\langle u^2\nabla v, \nabla\varphi\rangle+\langle\lambda u^2v+a_{\tau}v,\varphi\rangle=\langle a_{\rho},\varphi\rangle,
\qquad\forall\varphi\in C^{\infty}.
\end{equation}
Since $a_{\rho},a_{\tau}\in W^{s-2,p}_{+}$ with $sp>3$, we have $v\in W^{s,p}\hookrightarrow L^\infty$, and Lemma \ref{l:max-princ} (maximum principle) in Appendix \ref{sec:maxprinciple} implies that $v>0$.
Let us define $\phi= \beta uv\in W^{s,p}$ for a constant $\beta > 0$. Then for any $\varphi\in C^{\infty}_{+}$ we have
\begin{align*}
\langle A_{\tiL}\phi +  f(\phi,\biw),u\varphi\rangle
&\leqs
\langle A_{\tiL}\phi,u\varphi\rangle + \langle a_{\tau}\phi^5 + a_{\tiR}\phi - a_{\rho}\phi^{-3},u\varphi\rangle\\
&=
\beta\langle u^2\nabla v,\nabla\varphi\rangle+\langle\beta\lambda u^2v+a_{\tau}u^6(\beta v)^5 - a_{\rho}u^{-2}(\beta v)^{-3},\varphi\rangle\\
&=
\beta\langle a_{\rho}[1-u^{-2}v^{-3}\beta^{-4}],\varphi\rangle
+\beta\langle a_{\tau}[u^6v^5\beta^4-1],\varphi\rangle,
\end{align*}
where the second line is obtained by \eqref{e:auv}, and the third line is from \eqref{e:rhocurv}.
Now, choosing $\beta>0$ sufficiently small, so that $1-u^{-2}v^{-3}\beta^{-4}\leqs0$ and $(\beta v)^4-1\leqs0$, we ensure that $\phi$ is a sub-solution.
\end{proof}

The following lemma extends Lemma \ref{L:HC-GSb}(a) 
to all reasonable metrics in the negative Yamabe class.

\begin{lemma}
{\bf (Global sub-solution $h\in\cY^{-}$)}
\label{L:global-sub-Y-}
Let $(\cM,h)$ be a 3-dimensional, smooth, closed
Riemannian manifold with metric $h \in W^{s,p}$ in $\cY^{-}(\cM)$.
In addition, let $a_{\tau}\in W^{s-2,p}$, and let the metric $h$ be conformally equivalent to a metric with scalar curvature $(-a_{\tau})$.
Then, there exists a positive scalar function $\phi_{-}\in W^{s,p}$ such that for any $\beta\in(0,1]$,
$\beta\phi_{-}$ is a global sub-solution of the Hamiltonian constraint equation.
\end{lemma}

\begin{proof}
Let $u>0$ be the conformal factor which transforms $h$ into a metric with scalar curvature $\lambda=-8a_{\tau}$,
i.e., let $u\in W^{s-2,p}$ be a weak solution to
\begin{equation*}\textstyle
-\Delta u+\frac18Ru + a_{\tau}u^5 = 0,
\quad u > 0.
\end{equation*}
Taking $\phi= \beta u$ with a constant $\beta > 0$, we have
\begin{align*}
A_{\tiL}\phi+f(\phi,\biw)
&\leqs
A_{\tiL}\phi + a_{\tau}\phi^5
+ a_{\tiR}\phi
=\textstyle
-\beta\Delta u + a_{\tau}(\beta u)^5
+ \frac\beta8Ru\\
&=\beta a_{\tau}u^5(\beta^4-1).
\end{align*}
By choosing $\beta\in(0,1]$, we get the sub-solution.
\end{proof}

The following lemma shows that the additional condition on the metric 
appearing in Lemma~\ref{L:global-sub-Y-} is indeed not restrictive.
It is worth noting that this next result can be viewed as an apparently 
new non-existence result in the context of the non-CMC constraints, 
which is interesting in its own right.
This result was first proved in~\cite{dM05} for the case of $p=2$;
we just need to reinterpret it here in our setting.
It states that for there to be a (CMC or non-CMC) solution to the 
Hamiltonian constraint, the background metric $h_{ab}$ must be conformally
equivalent to a metric with scalar curvature equal to $(-a_{\tau})$.

\begin{lemma}
{\bf (Non-existence $h\in\cY^{-}$)}
   \label{L:nec}
Let $(\cM,h)$ be a 3-dimensional, smooth, closed
Riemannian manifold with metric $h \in W^{s,p}$ in $\cY^{-}(\cM)$.
Let $a_{\tau}\in W^{s-2,p}$,
and let there exist a solution to the Hamiltonian constraint equation.
Then, the metric $h$ is conformally equivalent to a metric with scalar curvature $(-a_{\tau})$.
\end{lemma}

\begin{proof}
It suffices to show that the equation
\begin{equation}\label{e:taucurv}\textstyle
-\Delta\psi+\frac18R\psi+a_{\tau}\psi^5=0,
\end{equation}
has a solution $\psi>0$.
Since the above equation is just a Hamiltonian constraint equation with $a_{\rho}=a_{\biw}=0$,
Theorem \ref{T:main3} establishes the proof upon 
constructing sub- and super-solutions to \eqref{e:taucurv}.

Let $\phi>0$ be a solution to the (general) Hamiltonian constraint equation.
Then, since both $a_{\rho}$ and $a_{\biw}$ are non-negative, we have
\begin{align*}\textstyle
-\Delta\phi+\frac18R\phi+a_{\tau}\phi^5\geqs0,
\end{align*}
which means that $\phi$ is a super-solution to \eqref{e:taucurv}.

Let  $u\in W^{s,p}$ be a (weak) solution to
\begin{equation*}\textstyle
-\Delta u+\frac18Ru = -\lambda u,
\quad u > 0,
\end{equation*}
with a constant $\lambda>0$, which exists by Theorem \ref{t:subcrit} in Appendix \ref{sec:yamabe},
and with a real parameter $\varepsilon$, let $v_{\varepsilon}\in W^{s,p}$ be the solution to
\begin{equation*}
\langle u^2\nabla v_{\varepsilon}, \nabla\varphi\rangle+\langle\lambda u^2v_{\varepsilon},\varphi\rangle=\langle\lambda u^2-a_{\tau}\varepsilon,\varphi\rangle,
\qquad\forall\varphi\in C^{\infty}.
\end{equation*}
We have $v_{\varepsilon}\equiv1$ for $\varepsilon=0$, and we have $v_{\varepsilon}\in W^{s,p}\hookrightarrow L^{\infty}$,
so as ${\varepsilon}$ goes to $0$, $v_{\varepsilon}$ tends to $1$ uniformly.
Let us fix ${\varepsilon}>0$ such that $v_{\varepsilon}\geqs\frac12$.
By taking $\psi= \beta uv_{\varepsilon}$ with a constant $\beta > 0$, and using \eqref{e:auv}, it holds for any $\varphi\in C^{\infty}_{+}$ that
\begin{align*}
\langle\nabla\psi,\nabla(u\varphi)\rangle+&\langle\textstyle\frac18R\psi+a_{\tau}\psi^5,u\varphi\rangle
=
\beta\langle u^2\nabla v_{\varepsilon},\nabla\varphi\rangle+\langle a_{\tau}u^6(\beta v_{\varepsilon})^5-\beta\lambda u^2v_{\varepsilon},\varphi\rangle\\
&=
\beta\langle a_{\tau}(u^6v_{\varepsilon}^5\beta^4-\varepsilon),\varphi\rangle
+\beta\lambda\langle u^6(1-2v_{\varepsilon}),\varphi\rangle.
\end{align*}
Now, by choosing $\beta>0$ small enough, we can ensure that $\psi$ is a sub-solution of \eqref{e:taucurv}.
\end{proof}

\subsection{A priori $L^{\infty}$ bounds on $W^{1,2}$ solutions}
   \label{sec:apriori}
We now establish some related {\em a~priori} $L^{\infty}$-bounds on
any $W^{1,2}$-solution to the Hamiltonian constraint equation.
Although such results are standard for semi-linear scalar problems
with monotone nonlinearities (for example, see~\cite{jJ85}), the
nonlinearity appearing in the Hamiltonian constraint becomes
non-monotone when $R$ becomes negative.
Nonetheless, we are able to obtain {\em a~priori} $L^{\infty}$-bounds
on solutions to the Hamiltonian constraint in all cases including
the non-monotone case.
See~\cite{mHgNgT08b} for an analogue of this result in the case of
compact manifolds with boundary; 
in that case a more general result is possible.

\begin{lemma}{\bf (Pointwise {\em a priori} bounds)}
\label{L:HC-ape}
Let $\phi\in W^{1,2}$ be any non-constant positive solution of the
Hamiltonian constraint equation \eqref{WF-LYs2}.
\begin{itemize}
\item[(a)] Let $a_{\tau\tiR}^{\tivee}:=\mathrm{ess~inf}\,(a_{\tau}+a_{\tiR})>0$,
and let $a_{\rho}^{\tiwedge}$ and $a_{\tbw}^{\tiwedge}$ be finite.
Then, $\phi$ satisfies the {\em a priori} bound
$$
\phi^4\leqs\max\left\{1,\frac{a_{\rho}^{\tiwedge}+a_{\tbw}^{\tiwedge}}{a_{\tau\tiR}^{\tivee}}\right\}.
$$

\item[(b)] Let $a_{\tau}^{\tivee}>0$ and let $a_{\rho}^{\tiwedge}$ and $a_{\tbw}^{\tiwedge}$ be finite.
Then, $\phi$ satisfies the {\em a priori} bound
$$
\phi^4\leqs\max\left\{1,
\frac{\sqrt{(a_{\tiR}^{\tivee})^2+a_{\tau}^{\tivee}(a_{\rho}^{\tiwedge}+a_{\tbw}^{\tiwedge})}-a_{\tiR}^{\tivee}}{a_{\tau}^{\tivee}}\right\}.
$$

\item[(c)] Let $a_{\rho\tbw}^{\tivee}:=\mathrm{ess~inf}\,(a_{\rho}+a_{\tbw})>0$,
and let $a_{\tau}^{\tiwedge}$ be finite.
Then, $\phi$ satisfies the {\em a priori} bound
$$
\phi^4\geqs\frac{a_{\rho\tbw}^{\tivee}}{\max\{a_{\rho\tbw}^{\tivee},a_{\tau}^{\tiwedge}+a_{\tiR}^{\tiwedge}\}}.
$$
\end{itemize}
\end{lemma}

\begin{proof}
We will only prove (a) since the other cases can be proven similarly.

Let $\chi\in W^{1,2}$ be any function with $\chi\geqs1$.
Then for $\varphi\in C^{\infty}_{+}$ we have
\begin{align*}
\langle f_{\biw}(\chi),\varphi\rangle
&\geqs (\chi^{\tivee})^{5} \langle a_{\tau},\varphi\rangle
+ \chi^{\tivee} \langle a_{\tiR},\varphi\rangle
- (\chi^{\tivee})^{-3} ( a_{\rho},\varphi)
- (\chi^{\tivee})^{-7} ( a_{\biw},\varphi)\\
&\geqs \bigl( a_{\tau\tiR}^{\tivee} \, \chi^{\tivee}
- (\chi^{\tivee})^{-3} [a_{\rho}^{\tiwedge}+a_{\biw}^{\tiwedge}]
\bigr)
\,\|\varphi\|_{1}.
\end{align*}
So we conclude that
\[
\langle f_{\biw}(\chi),\varphi\rangle \geqs 0
\qquad \forall \chi\geqs\phi^{\wedge},
\, \chi\in W^{1,2},
\qquad  \forall \varphi\in C^{\infty}_{+},
\]
where $(\phi^{\tiwedge})^{4}=\max\{1,\frac{a_{\rho}^{\tiwedge}+a_{\biw}^{\tiwedge}}{a_{\tau\tiR}^{\tivee}}\}$.

Now, suppose that $\phi\in W^{1,2}$ is a solution of the Hamiltonian constraint equation,
such that $\phi\not\leqs\phi^{\tiwedge}$.
Denoting by $(\phi-\phi^{\tiwedge})^{+}$ the positive part of $\phi-\phi^{\tiwedge}$ (cf. Appendix \ref{sec:maxprinciple}),
then we have
\begin{align*}
0 &\geqs -\langle f_{\biw}(\phi),(\phi-\phi^{\tiwedge})^{+}\rangle
=(\nabla\phi,\nabla(\phi-\phi^{\tiwedge})^{+})
=(\nabla(\phi-\phi^{\tiwedge})^{+},\nabla(\phi-\phi^{\tiwedge})^{+})\\
&\geqs c\|(\phi-\phi^{\tiwedge})^{+}-\overline{(\phi-\phi^{\tiwedge})^{+}}\|_2^2,
\end{align*}
where $c>0$,
and $\overline{(\phi-\phi^{\tiwedge})^{+}}$ is the integral average of $(\phi-\phi^{\tiwedge})^{+}$.
This implies that $\phi$ is constant, leading to a contradiction.
\end{proof}

\section{Proof of the main results}
   \label{sec:proof}

It is convenient to prove Theorem~\ref{T:main2} first, which is
the most general of the three; the proofs of Theorem~\ref{T:main1}
and Theorem~\ref{T:main3} involve minor modifications of the
proof of Theorem~\ref{T:main2}.

\subsection{Proof of Theorem~\ref{T:main2}}
\label{sec:proof1}
Our strategy will be to prove the theorem first for the case $s\leqs2$, and then to bootstrap to include the higher regularity cases.

{\em Step 1: The choice of function spaces.}
We have the (reflexive) Banach spaces
$X=W^{s,p}$
and
$Y=\biW^{e,q}$,
where $p,q \in (3,\infty)$, $s=s(p) \in (1+\frac{3}{p},2]$,
and $e=e(p,s,q)\in(1,s]\cap(1+\frac3q,s-\frac3p+\frac3q]$.
We have the ordered Banach space
$Z=W^{\tilde{s},p}$
with the compact embedding $X=W^{s,p}\hookrightarrow W^{\tilde{s},p}=Z$,
for $\tilde{s} \in (\frac{3}{p}, s)$.
The interval $[\phi_{-},\phi_{+}]_{\tilde{s},p}$ is
nonempty (by compatibility of the barriers we will choose below),
and by Lemma~\ref{L:wsp-interval} on page \pageref{L:wsp-interval}
it is also convex with respect to the vector space structure of
$W^{\tilde{s},p}$ and closed with respect to the norm topology
of $W^{\tilde{s},p}$.
We then take $U=[\phi_-,\phi_+]_{\tilde{s},p} \cap \overline{B}_M$ 
for sufficiently
large $M$ (to be determined below), where
$\overline{B}_M$ is the closed ball in $Z=W^{\tilde{s},p}$
of radius $M$ about the origin, ensuring that $U$ is
non-empty, convex, closed, and bounded as a subset of $Z=W^{\tilde{s},p}$.

{\em Step 2: Construction of the mapping $S$.}
We have $\bib_j\in\biW^{e-2,q}$, and $\bib_\tau\in\biL^{z}$ with $z=\frac{3q}{3+(2-e)q}$ so that $\biL^{z}\hookrightarrow\biW^{e-2,q}$.
Moreover, since the metric admits no conformal Killing field, by Lemma~\ref{T:w-MC} the momentum constraint equation is uniquely solvable 
for any ``source'' $\phi\in [\phi_-,\phi_+]_{\tilde{s},p}$.
The ranges for the exponents ensure that Lemma~\ref{T:MC-E-Lip1} holds, so that the momentum
constraint solution map 
$$S : [\phi_{-},\phi_{+}]_{\tilde{s},p} \to\biW^{e,q}=Y,$$
is continuous.

{\em Step 3: Construction of the mapping $T$.}
Define $r=\frac{3p}{3+(2-s)p}$, so that the
continuous embedding $L^{r}\hookrightarrow W^{s-2,p}$ holds.
Since the pointwise multiplication is bounded on $L^{2r}\otimes L^{2r}\to L^{r}$, and $\biw\in\biW^{e,q}\hookrightarrow\biW^{1,2r}$, we have $a_{\biw}\in W^{s-2,p}$ by $\sigma\in L^{2r}$.
The embeddings $W^{1,z}\hookrightarrow W^{e-1,q}\hookrightarrow L^{2r}$ also guarantee that $a_{\tau}=\frac1{12}\tau^2\in W^{s-2,p}$.
We have the scalar curvature $R\in W^{s-2,p}$, and these considerations show that the Hamiltonian constraint equation is well defined
with $[\phi_-,\phi_+]_{s,p}$ as the space of solutions.

Suppose for the moment that the scalar curvature $R$ of the background metric $h$ is continuous, 
and by using the map $T^s$ introduced in Lemma \ref{l:shift}, define the map $T$ by $T(\phi,\biw)=T^s(\phi,a_{\biw})$,
where $a_{\biw}$ is now considered as an expression depending on $\biw$.
Then Lemma~\ref{l:shift} implies that the map $T:[\phi_{-},\phi_{+}]_{\tilde{s},p}\times\biW^{e,q}\to W^{s,p}$ is continuous for any reasonable shift $a_s$,
which, by Lemma~\ref{l:shift1}, can be chosen so that $T$ is monotone in the first variable.
Combining the monotonicity with Lemma \ref{l:shiftsubsup}, we infer that the interval $[\phi_{-},\phi_{+}]_{\tilde{s},p}$
is invariant under $T(\cdot,a_{\biw})$ if $\biw\in S([\phi_{-},\phi_{+}]_{\tilde{s},p})$.
Since $\biL^z\hookrightarrow\biW^{e-2,q}$, from Theorem \ref{T:w-MC} we have 
$$
\|\biw\|_{e,q}\leqs C\, \|\bib_{\tau}\phi^6+\bib_{j}\|_{e-2,q}\leqs C\, \|\phi_{+}\|_{\infty}^6\|\bib_{\tau}\|_{z}+C\,\|\bib_{j}\|_{e-2,q}
$$ for any $\biw\in S([\phi_-,\phi_+]_{\tilde{s},p})$.
In view of Lemma~\ref{T:HC-ball-gen}, this shows that
there exists a closed ball $\overline{B}_M\subset W^{\tilde{s},p}$ such that
\begin{equation*}
\phi \in [\phi_-,\phi_+]_{\tilde{s},p}\cap \overline{B}_M,
\quad \biw \in S([\phi_-,\phi_+]_{\tilde{s},p}\cap \overline{B}_M)
\quad\Rightarrow\quad
T(\phi,\biw)\in\overline{B}_M.
\end{equation*}
Under the conditions in the above displayed formula, from the invariance of the interval $[\phi_{-},\phi_{+}]_{\tilde{s},p}$, we indeed have $T(\phi,\biw)\in U=[\phi_-,\phi_+]_{\tilde{s},p}\cap \overline{B}_M$.

However, the scalar curvature of $h$ may be not continuous, and in general it is not clear how to introduce a shift so that the resulting operator is monotone.
Nevertheless, we can conformally transform the metric into a metric with continuous scalar curvature, cf. Theorem \ref{t:yclass}, and by using the conformal 
covariance of the 
Hamiltonian constraint, we will be able to construct an appropriate mapping $T$.
Let $\tilde{h}=\theta^4h$ be a metric with continuous scalar curvature, where $\theta\in W^{s,p}$ is the (positive) conformal factor of the scaling.
Let $\tilde{T}^s$ be the mapping introduced in Lemma \ref{l:shift}, 
corresponding to the Hamiltonian constraint equation with the background metric $\tilde{h}$, 
and the coefficients $\tilde{a}_{\tau}=a_{\tau}$, and $\tilde{a}_{\rho}=\theta^{-8}a_{\rho}$.
With $\tilde{a}_{\biw}=\theta^{-12}a_{\biw}$, this {\em scaled} Hamiltonian constraint equation has sub- and super-solutions $\theta^{-1}\phi_{-}$ and $\theta^{-1}\phi_{+}$,
respectively, as long as $\phi_{-}$ and $\phi_{+}$ are sub- and super-solutions respectively of the original Hamiltonian constraint equation, cf. Appendix \ref{sec:conf-inv}.
We choose the shift in $\tilde{T}^s$ so that it is monotone in $[\theta^{-1}\phi_{-},\theta^{-1}\phi_{+}]_{\tilde{s},p}$. 
Then by the monotonicity and the above mentioned sub- and super-solution property under conformal scaling, for $\biw\in S([\phi_{-},\phi_{+}]_{\tilde{s},p})$, $\tilde{T}^s(\cdot,\theta^{-12}a_{\biw})$ is invariant on $[\theta^{-1}\phi_{-},\theta^{-1}\phi_{+}]_{\tilde{s},p}$.
Finally, we define
$$
T(\phi,\biw)=\theta\tilde{T}^s(\theta^{-1}\phi,\theta^{-12}a_{\biw}),
$$
where, as before, $a_{\biw}$ is considered as an expression depending on $\biw$.
From the pointwise multiplication properties of $\theta$ and $\theta^{-1}$,
the map $T:[\phi_{-},\phi_{+}]_{\tilde{s},p}\times\biW^{e,q}\to W^{s,p}$ is continuous,
and from the monotonicity and Lemma \ref{T:HC-ball-gen} , $T(\cdot,\biw)$ is invariant on $U=[\phi_-,\phi_+]_{\tilde{s},p}\cap \overline{B}_M$ for $\biw\in S(U)$,
where $M$ is taken to be sufficiently large.
Moreover, if the fixed point equation
$$
\phi=\theta\tilde{T}^s(\theta^{-1}\phi,\theta^{-12}a_{\biw}),
$$
is satisfied, then $\theta^{-1}\phi$ is a solution to the scaled Hamiltonian constraint equation with $\tilde{a}_{\biw}=\theta^{-12}a_{\biw}$,
and so by conformal covariance,
$\phi$ is a solution to the original Hamiltonian constraint equation, cf. Appendix \ref{sec:conf-inv}.

{\em Step 4: Barrier choices and application of the fixed point theorem.}
At this point, Theorem \ref{T:FIXPT2} implies the Main Theorem \ref{T:main2},
provided that we have an admissible pair of barriers for the Hamiltonian constraint. 
The ranges for the exponents ensure through
Corollary~\ref{T:MC-E-aw} that we can use the
estimate \eqref{CS-aLw-bound};
see the discussion following the estimate on page \pageref{CS-aLw-bound}.
We will separate into the two cases in the theorem,
depending on which Yamabe class we are in:
\begin{itemize}
\item[(a)]
  $h_{ab}$ is in $\cY^{-}(\cM)$:
  We use the global constant super-solution from Lemma \ref{L:HC-Sp}(a)
  or the non-constant super-solution from Lemma \ref{L:global-super-near}
  depending on whether $\rho$ and $\sigma$ are both in $L^{\infty}$,
  and the global sub-solution from Lemma \ref{L:global-sub-Y-}.
\item[(b)]
  $h_{ab}$ is in $\cY^{0}(\cM)$ or in $\cY^{+}$:
  We use the global constant super-solution from Lemma \ref{L:HC-Sp}(a)
  or the non-constant super-solution from Lemma \ref{L:global-super-near}
  depending on whether $\rho$ and $\sigma$ are both in $L^{\infty}$,
  and the global sub-solution from Lemma \ref{L:global-sub} 
  or Lemma \ref{L:HC-GSb}(c).
\end{itemize}
This concludes the proof for the case $s\leqs2$.

{\em Step 5: Bootstrap.}
Now suppose that $s>2$.
First of all we need to show that the equations are well defined in the sense that the involved operators are bounded in appropriate spaces.
All other conditions being obviously satisfied, we will show that $a_{\tau}\in W^{s-2,p}$, and $a_{\biw}\in W^{s-2,p}$ for any $\biw\in\biW^{e,q}$.
Since $\tau$, $\sigma$ and $\cL\biw$ belong to $W^{e-1,q}$, it suffices to show that the pointwise multiplication is bounded on $W^{e-1,q}\otimes W^{e-1,q}\to W^{s-2,p}$,
and by employing Corollary \ref{c:alg}(b) in Appendix, we are done as long as $s-2\leqs e-1\geqs0$, $s-2-\frac3p<2(e-1-\frac3q)$, and $s-2-\frac3p\leqs e-1-\frac3q$.
After a rearrangement these conditions read as $e\geqs1$, $e\geqs s-1$, $e>\frac3q+\frac{d}2$, and $e\geqs\frac3q+d-1$, with the shorthand $d=s-\frac3p>1$, the latter inequality by the hypothesis of the theorem.
We have $d-1>\frac{d}2$ for $d>2$, and $1\geqs\frac{d}2$ for $d\leqs2$,
meaning that the condition $e>\frac3q+\frac{d}2$ is implied by the hypotheses $e\geqs\frac3q+d-1$ and $e>1+\frac3q$.
So we conclude that the constraint equations are well defined.

Next, we will treat the equations as equations defined with $s=e=2$ and with $p$ and $q$ appropriately chosen.
This is possible, since if the quadruple $(p,s,q,e)$ satisfies the hypotheses of the theorem,
then $(\tilde p,\tilde s=2,\tilde q,\tilde e=2)$ satisfies the hypotheses too, provided that 
$2-\frac3{\tilde p}\leqs s-\frac3p$, and 
$1<2-\frac3{\tilde q}\leqs e-\frac3q$.
Since the latter conditions reflect the Sobolev embeddings $W^{s,p}\hookrightarrow W^{2,\tilde p}$ and $W^{e,q}\hookrightarrow W^{2,\tilde q}\hookrightarrow W^{1,\infty}$,
the coefficients of the equations can also be shown to satisfy sufficient conditions for posing the problem for $(\tilde p,2,\tilde q,2)$.
Finally, we have $\tau\in W^{e-1,q}\hookrightarrow W^{1,\tilde{q}}=W^{1,z}$ since $z=\tilde{q}$ by $\tilde{e}=2$ for this new formulation.
Now, by the special case $s\leqs2$ of this theorem that is proven in the above steps, 
under the remaining hypotheses including the conditions on the metric and the near-CMC condition,
we have $\phi\in W^{2,\tilde{p}}$ with $\phi>0$ and $\biw\in\biW^{2,\tilde{q}}$ solution to the coupled system.

To complete the proof we only need to show that these solutions indeed satisfy $\phi\in W^{s,p}$ and $\biw\in\biW^{e,q}$.
Suppose that $\phi\in W^{s_1,p_1}$ and $\biw\in\biW^{e_1,q_1}$, with 
$1<s_1-\frac3{p_1}\leqs s-\frac3p$,
$1<e_1-\frac3{q_1}\leqs e-\frac3q$,
$\max\{2,s-2\}\leqs s_1\leqs s$, and
$\max\{2,e-2\}\leqs e_1\leqs\min\{e,s\}$.
Then we have $\bib_\tau\phi^6+\bib_j\in\biW^{e-2,q}$, 
and so Corollary \ref{C:ell-est} from Appendix \ref{sec:killing} implies that $\biw\in\biW^{e,q}$.
This implies that $a_{\biw}\in W^{s-2,p}$, and by employing Corollary \ref{C:ell-est} once again, we get $\phi\in W^{s,p}$.
The proof is completed by induction.
\qed

\subsection{Proof of Theorem~\ref{T:main1}}
The proof is identical to the proof of Theorem~\ref{T:main2},
except for the particular barriers used.
In the proof of Theorem~\ref{T:main2}, the near-CMC condition
is used to construct global barriers satisfying
$$
0 < \phi_{-} \leqs \phi_{+} < \infty,
$$
for all three Yamabe classes, and then the supporting
results for the operators $S$ and $T$ established 
in~\S\ref{sec:momentum} and~\S\ref{sec:Hamiltonian}
are used to reduce the proof to invoking Theorem~\ref{T:FIXPT2}.
The construction of $\phi_{+}$ is in fact the only place in the
proof of Theorem~\ref{T:main2} that requires the near-CMC condition.
Here, the proof is identical, except that the additional conditions 
made on the background metric $h_{ab}$ (that it be in $\cY^{+}(\cM)$),
and on the data (the smallness conditions on $\sigma$, $\rho$, and $j$)
allow us to make use of the alternative construction 
of a global super-solution given 
in Lemma~\ref{L:global-super}, together with compatible global 
sub-solution given in Lemma~\ref{L:global-sub}, properly scaled
for compatibility with the super-solution.
Theorem \ref{T:main1} now follows from Theorem \ref{T:FIXPT2},
without the use of near-CMC conditions.
\qed

\subsection{Proof of Theorem~\ref{T:main3}}
The CMC result in this theorem can be proved using
the same analysis framework used for the proofs of the two non-CMC
results in Theorem~\ref{T:main1} and Theorem~\ref{T:main2} above.
Therefore, the proof follows the same general outline of the 
proof of Theorem~\ref{T:main2}, with slightly different
spaces and supporting results.
The main difference is that we can avoid having to construct 
``global'' barriers and getting uniform bounds on the solution to the 
momentum constraint, since it is solved only once {\em a priori}
and then is input as data into the nonlinearity of the Hamiltonian constraint.

The case (d) follows from the Yamabe classification, cf. Appendix \ref{sec:yamabe}.

Since otherwise we can use the conformal covariance 
of the Hamiltonian constraint as in Section \ref{sec:proof1},
for simplicity, assume that the scalar curvature of the background metric is continuous.
Also assume that $s\leqs2$, and let us look at the hypotheses of Theorem~\ref{T:FIXPT2}.
We have the (reflexive) Banach spaces
$X=W^{s,p}$
and
$Y=\biW^{1,2r}$,
where $p \in (\frac32,\infty)$, $s=s(p) \in (\frac{3}{p},\infty) \cap [1,2]$,
and $r=r(s,p) = \frac{3p}{3 + (2-s)p}$.
On the diagram in Figure \ref{f:cmc},  for $s\leqs2$
the space $\biW^{1,2r}$ corresponds to the lower right corner of the shaded parallelogram,
and so $\biW^{1,2r}$ contains all the spaces $\biW^{e,q}$ which are represented by the points in the shaded parallelogram.
In fact, $\biW^{1,2r}$ is outside of this parallelogram, because of the strict inequality relating $e$ and $q$
in order to have the boundedness of the pointwise multiplication on $W^{e-1,q}\otimes W^{e-1,q}\to W^{s-2,p}$
by using Corollary \ref{c:alg}(b).
However, the conditions of Corollary \ref{c:alg}(b) are not necessary conditions when some of the smoothness indices are integers,
for example, in our case the pointwise multiplication is bounded on $L^{2r}\otimes L^{2r}\to L^{r}$, even though these spaces do not satisfy the conditions of the corollary.
As a consequence, as we have seen e.g. in Section \ref{sub:weak}, the constraint equations are well defined for these spaces.

We have the ordered Banach space
$Z=W^{\tilde{s},p}$
with the compact embedding $X=W^{s,p}\hookrightarrow W^{\tilde{s},p}=Z$,
for $\tilde{s} \in (\frac{3}{p}, s)$.
The interval $[\phi_{-},\phi_{+}]_{\tilde{s},p}$ is
nonempty (by compatibility of the barriers we will choose below),
and by Lemma~\ref{L:wsp-interval} on page \pageref{L:wsp-interval}
it is also convex with respect to the vector space structure of
$W^{\tilde{s},p}$ and closed with respect to the norm topology
of $W^{\tilde{s},p}$.
We then take $U=[\phi_-,\phi_+]_{\tilde{s},p} \cap \overline{B}_M$ 
for sufficiently large $M$ (to be determined below), where
$\overline{B}_M$ is the closed ball in $Z=W^{\tilde{s},p}$
of radius $M$ about the origin, ensuring that $U$ is
non-empty, convex, closed, and bounded as a subset of $Z=W^{\tilde{s},p}$.

We take as $T$ the shifted Picard mapping $T^s$
having as its fixed-point a solution to the Hamiltonian constraint,
and we take $S(\phi)=\biw=-A_{\IL}^{-1}\bib_j \in \biW^{1,2r}$
which is independent of $\phi$, since the momentum equation
decouples from the Hamiltonian constraint in this case.
The map $S$,
which is constant as a function of $\phi$ due to the CMC de-coupling,
is trivially continuous as a 
map $S : U \to \biW^{1,2r}=Y$.
We now consider properties we have for $T$.
By Lemma~\ref{l:shift}, $T:U \times \mathcal{R}(S) \to W^{s,p} = X$ 
is a continuous map.
By Lemma~\ref{l:shift1}, $T$ is invariant on the closed interval 
$[\phi_-,\phi_+]_{\tilde{s},p}$, and by Lemma~\ref{T:HC-ball-gen},
$T$ is invariant on $U=[\phi_-,\phi_+]_{\tilde{s},p} \cap \overline{B}_M$.
To summarize, $T$ is invariant on the non-empty, closed, convex, bounded set 
$U$.

Finally, Theorem \ref{T:FIXPT2} implies the Main Theorem \ref{T:main3},
as long as we have an admissible pair of barriers for the 
Hamiltonian constraint. 
That is when we need to separate into the three remaining
cases in the theorem, depending on which Yamabe class we are in:
\begin{itemize}
\item[(a)]
 $h_{ab}$ is in $\cY^{-}(\cM)$; $\tau\neq0$:
  We take the super-solution from Lemma~\ref{L:local-super-near}(c),
  and we take the sub-solution from Lemma~\ref{L:global-sub-Y-}.
  These lemmata require that the metric $h_{ab}$ is conformally equivalent 
  to a metric with scalar curvature $(-a_{\tau})$, and we shall verify this condition.
  By conformal invariance, it suffices to verify the condition for metrics with
  continuous and negative scalar curvature,
  meaning that we have to solve the equation \eqref{e:taucurv} with $R<0$ continuous and $a_{\tau}>0$ constant.
  Indeed, this equation has a positive solution $\psi\in W^{s,p}$
  as the constants $\psi_{-}=(\frac{\min |R|}{8a_{\tau}})^{1/4}$
  and $\psi_{+}=(\frac{\max |R|}{8a_{\tau}})^{1/4}$
  are respectively sub- and super-solutions of \eqref{e:taucurv}.
\item[(b)]
 $h_{ab}$ is in $\cY^{+}(\cM)$; $\rho\neq0$ or $\sigma\neq0$:
  We take the super-solution from Lemma~\ref{L:local-super-near}(b),
  and we take the sub-solution from Lemma~\ref{L:global-sub}.
  For the case $\rho=0$ and $\sigma\neq0$, a local sub-solution can easily be constructed 
  following the approach in the proof of Lemma~\ref{L:global-sub}.
\item[(c)]
 $h_{ab}$ is in $\cY^{0}(\cM)$; $\tau\neq0$; $\rho\neq0$ or $\sigma\neq0$:
  We take the super-solution from Lemma~\ref{L:local-super-near}(a),
  and we take the sub-solution from Lemma~\ref{L:global-sub}.
  The case $\rho=0$ and $\sigma\neq0$ is treated as above.
\end{itemize}
To complete the proof one can bootstrap as in Section \ref{sec:proof1}.
\qed

\section{Summary}
   \label{sec:summary}

We began in~\S\ref{sec:constraints} by
summarizing the conformal decomposition of Einstein's constraint
equations introduced by Lichnerowicz and York, on a closed manifold.
After this setting up of the notation,
we gave an overview of our main results in~\S\ref{sec:main},
represented by three new weak solution existence results for the
Einstein constraint equations in the far-from-CMC, near-CMC, and CMC cases.
In~\S\ref{sec:individual} we then developed the necessary results
we need for the individual constraint equations in order to
analyze the coupled system.
In particular, in~\S\ref{sec:momentum},
we first developed some basic technical results for the momentum constraint
operator under weak assumptions on the problem data.
We also established the properties we
need for the momentum constraint solution mapping $S$ appearing 
in the analysis of the coupled system.
In~\S\ref{sec:Hamiltonian}, we assumed the existence of barriers
$\phi_-$ and $\phi_+$
(weak sub- and super-solutions) to the Hamiltonian constraint equation,
forming a nonempty positive bounded interval, and then established the 
properties we need for the Hamiltonian constraint Picard mapping $T$ 
appearing in the analysis of the coupled system.
We then derived
several weak global sub- and super-solutions in~\S\ref{sec:barriers},
based both on constants and on more complex non-constant constructions.
While the sub-solutions are similar to those found previously in
the literature, some of the super-solutions were new.
In particular, we gave two super-solution constructions that do
not require the near-CMC condition.
The first was constant, and requires that the scalar curvature be
strictly globally positive.
The second was based on a scaled solution to a Yamabe-type problem,
and is valid for any background metric in the positive Yamabe class.

In~\S\ref{sec:proof}, we proved the main results.
In particular, using topological fixed-point arguments and global barrier
constructions, we combined the results for the individual constraints and
the global barriers to establish existence of coupled non-CMC
weak solutions with (positive) conformal factor
$\phi \in W^{s,p}$
where $p\in(1,\infty)$
and $s(p) \in (1+\frac{3}{p},\infty)$.
In the CMC case, the regularity can be reduced to
$p \in (1,\infty)$ and $s(p) \in (\frac{3}{p}, \infty) \cap [1,\infty)$.
In the case of $s=2$, we reproduce the CMC existence results of 
Choquet-Bruhat~\cite{yCB04}, and
in the case $p=2$, we reproduce the CMC existence results of
Maxwell~\cite{dM05}, but with a different proof;
our CMC proof goes through the same analysis framework that we use to 
obtain the non-CMC results (Theorems~\ref{T:FIXPT1} and~\ref{T:FIXPT2}).
We also assembled a number of new supporting technical results in the
body of the paper and in several appendices, including: 
topological fixed-point arguments designed for the Einstein
constraints;
construction and properties of general Sobolev classes $W^{s,p}$
and elliptic operators on closed manifolds with weak metrics;
the development of a very weak solution theory for the momentum 
constraint; {\em a priori} $L^{\infty}$-estimates for 
weak $W^{1,2}$-solutions to the Hamiltonian constraint; 
Yamabe classification of non-smooth metrics in general 
Sobolev classes $W^{s,p}$; and a discussion and analysis of conformal 
covariance 
and the connection between conformal rescaling and the near-CMC condition.

An important feature of the results we presented here is the absence of
the near-CMC assumption in the case of the rescaled background metric
in the positive Yamabe class, as long as
the freely specifiable part of the data given by the matter fields
(if present) and the traceless-transverse part of the rescaled extrinsic
curvature are taken to be sufficiently small.
In this case, the mean extrinsic curvature can be taken to be an
arbitrary smooth function without restrictions on the size of its
spatial derivatives, so that it can be arbitrarily far from constant.
Under these conditions, we have the first existence result for
non-CMC solutions without the near-CMC condition.
The two advances in the analysis of the Einstein constraint equations
make these results possible were: A topological fixed-point theorem
based on compactness arguments that is free of the near-CMC condition
(Theorems~\ref{T:FIXPT1} and~\ref{T:FIXPT2} and in~\cite{mHgNgT08b}),
and a new construction of global super-solutions for the
Hamiltonian constraint that is similarly free of the near-CMC
condition (Lemma~\ref{L:HC-Sp} and Lemma~\ref{L:global-super}).
We note that the near-CMC-free constructions based on
scaled solutions to a Yamabe-like problem also work for compact manifolds 
with boundary and other cases; see e.g.~\cite{mHgNgT08b}.

Finally, we point out that
our results here and in~\cite{mHgNgT08b,mHgNgT07a} can be viewed as
reducing the remaining open questions of existence of 
non-CMC (weak and strong) solutions without near-CMC conditions to 
two more basic and clearly stated open problems:
(1) Existence of near-CMC-free global {\em super}-solutions for the
Hamiltonian constraint equation when the background metric 
is in the non-positive Yamabe classes and for large data;
and 
(2) existence of near-CMC-free global {\em sub}-solutions for the
Hamiltonian constraint equation when the background metric 
is in the positive Yamabe class in vacuum (without matter).
However, an important new development, which occurred a few months after 
the first draft of this article was made available, is that Maxwell 
has now shown~\cite{dM08} how a related topological fixed-point argument 
can be constructed so that a global sub-solution is not needed, as long 
as the global super-solution is available; this allows for the extension 
of the far-CMC results in this article to the vacuum case without 
having to solve problem (2).

\section{Acknowledgments}
The authors would like to thank Jim Isenberg, David Maxwell,
and Daniel Pollack for many very insightful comments and 
suggestions about this work.
We would like to thank David Maxwell in particular for his careful reading 
of earlier drafts of this work and for pointing out various errors.
MH was supported in part by NSF Awards~0715146, 0411723, 
and 0511766, and DOE Awards DE-FG02-05ER25707 and DE-FG02-04ER25620.
GN and GT were were supported in part by NSF Awards~0715146 and 0411723.

\appendix
\section{Some key technical tools and some supporting results}
\label{sec:app}

\subsection{Topological fixed-point theorems}
\label{sec:fixpt}

In this appendix, we give a brief review of some standard topological
fixed-point theorems in Banach spaces that provide the framework
for our analysis of the coupled constraint equations.
The analysis framework that was developed earlier in~\cite{jIvM96} for 
analyzing the coupled constraints was based on $k$-contractive mappings, 
and as a result required the near-CMC condition in order to establish
$k$-contractivity.
All subsequent non-CMC results (see e.g.~\cite{pAaCjI07})
are based on the framework from~\cite{jIvM96}, and as a result
remain limited to the near-CMC case.
Our interest here is on more general topological fixed-point arguments
that will allow us to avoid the near-CMC condition.

\medskip
\noindent
{\bf Brouwer, Schauder, and Leray-Schauder Fixed-Point Theorems.}
To establish the main abstract results we will need,
we first give a brief overview of some standard results
on topological fixed-point arguments involving compactness.
\begin{theorem}
{\bf (Brouwer Theorem)}
\label{T:BROUWER}
Let $U \subset \mathbb{R}^n$ be a non-empty, convex, compact subset,
with $n \geqs 1$.
If $T:U\to U$ is a continuous mapping,
then there exists a fixed-point $u \in U$ such that $u=T(u)$.
\end{theorem}
\begin{proof}
See Proposition~2.6 in~\cite{Zeidler-I}; a short proof can be
based on homotopy-invariance of topological degree.
\end{proof}
\begin{theorem}
{\bf (Schauder Theorem)}
\label{T:SCHAUDER}
Let $X$ be a Banach space,
and let $U \subset X$ be a non-empty, convex, compact subset.
If $T:U\to U$ is a continuous operator,
then there exists a fixed-point $u \in U$ such that $u=T(u)$.
\end{theorem}
\begin{proof}
This is a direct extension of the Brouwer Fixed-Point Theorem
from $\mathbb{R}^n$ to $X$; see Corollary 2.13 in~\cite{Zeidler-I}.
The short proof involves a simple finite-dimensional approximation 
algorithm and a limiting argument, extending the Brouwer Fixed-Point
Theorem (itself generally having a more complicated proof) 
from $\mathbb{R}^n$ to $X$.
\end{proof}
\begin{theorem}
{\bf (Schauder Theorem B)}
\label{T:SCHAUDER-B}
Let $X$ be a Banach space,
and let $U \subset X$ be a non-empty, convex, closed, bounded subset.
If $T:U\to U$ is a compact operator,
then there exists a fixed-point $u \in U$ such that $u=T(u)$.
\end{theorem}
\begin{proof}
See Theorem 2.A in~\cite{Zeidler-I};
the proof is a simple consequence of Theorem~\ref{T:SCHAUDER} above.
\end{proof}

\subsection{Ordered Banach spaces}
\label{sec:OBS}

These notes follow the main ideas and definitions given Chapter~7.1,
page~275, in \cite{Zeidler-I}, while some examples were taken
from~\cite{hA76} and~\cite{Du06}. Let $X$ be a Banach space, $\R_{+}$
be the non-negative real numbers. A subset $C\subset X$ is a {\bf
cone} iff given any $x\in C$ and $a\in\R_{+}$ the element $ax\in C$. A
subset $X_{+}\subset X$ is an {\bf order cone} iff the following
properties hold:
\begin{enumerate}[{\it(i)}]
\item
The set $X_{+}$ is non-empty, closed, and $X_{+}\neq \{0\}$;
\item
Given any $a$, $b\in\R_{+}$ and $x$, $\un x\in X_{+}$ then 
$ax +b\un x \in X_{+}$;
\item 
If $x\in X_{+}$ and $-x\in X_{+}$, then $x=0$.
\end{enumerate}
The second property above says that every order cone is in fact a
cone, and that the set $X_{+}$ is convex. The space $X=\R^2$ is a
convenient Banach space to picture non-trivial examples of cones and
order cones, as can be seen in Fig.~\ref{F:cones}. A pair $X$,
$X_{+}$ is called an {\bf ordered Banach space} iff $X$ is a Banach
space and $X_{+}\subset X$ is an order cone. The reason for this name
is that the order cone $X_{+}$ defines several relations on elements
in $X$, called order relations, as follows:
\[
\begin{gathered}
u \geqs v \mbox{~~iff~~} u-v \in X_{+},\\
u \gg v \mbox{~~iff~~} u-v \in \mbox{int}(X_{+}),
\end{gathered}
\qquad
\begin{gathered}
u > v \mbox{~~iff~~} u\geqs v \mbox{~~and~~} u\neq v,\\
u \ngeqs v \mbox{~~iff~~} u\geqs v \mbox{~is false};
\end{gathered}
\]
finally it is also used the notation $u\leqs v$, $u< v$, and $u\ll v$
to mean $v\geqs u$, $v>u$, $v\gg u$, respectively. A simple example of
an ordered Banach space is $\R$ with the usual order. Another example
can be constructed when this order on $\R$ is transported into
$C^0(\cM)$, the set of scalar-valued functions on a set
$\cM\subset\R^n$, with $n\geqs 1$. An order on $C^0(\cM)$
is the following: the functions $u$, $v\in C^0(\cM)$
satisfy $u\geqs v$ iff $u(x) \geqs v(x)$ for all $x\in \cM$. The
following Lemmas summarize the main properties of order relations
in Banach spaces.
\begin{lemma}
\label{L:o-1}
Let $X$, $X_{+}$ be an ordered Banach space. Then, for all elements
$u$, $v$, $w \in X$, hold: (i) $u \geqs u$; (ii) If $u\geqs v$ and $v
\geqs u$, then $u=v$; (iii) If $u \geqs v$ and $v \geqs w$, then
$u\geqs w$.
\end{lemma}

\begin{proof}
The property that $u-u=0\in X_{+}$ implies that $u\geqs u$. If
$u\geqs v$ and $v\geqs u$ then $u-v \in X_{+}$ and $-(u-v)\in X_{+}$,
therefore $u-v =0$. Finally, if $u\geqs v$ and $v\geqs w$, then $u-v
\in X_{+}$ and $v-w\in X_{+}$, which means that $u-w = (u-v) + (v-w)
\in X _{+}$.
\end{proof}

Furthermore, the order relation is compatible with the vector space
structure and with the limits of sequences.
\begin{lemma}
\label{L:o-2}
Let $X$, $X_{+}$ be an ordered Banach space. Then, for all $u$, $\hat
u$, $v$, $\hat v$, $w \in X$, and $a$, $b\in \R$, hold
\begin{enumerate}[(i)]
\item \label{L:o-2i}
If $u\geqs v$ and $a\geqs b \geqs 0$, then $au \geqs bv$;
\item \label{L:o-2ii}
If $u\geqs v$ and $\hat u \geqs \hat v$, then $u+\hat u \geqs v+\hat v$;
\item \label{L:o-2iii}
If $u_n \geqs v_n$ for all $n\in \N$, 
then $\lim_{n\to\infty}u_n \geqs \lim_{n\to\infty}v_n$.
\end{enumerate}
\end{lemma}

\begin{proof}
The first two properties are straightforward to prove, and we do not
do it here. The third property holds because the order cone is a
closed set. Indeed, $u_n \geqs v_n$ means that $u_n-v_n \in X_{+}$ for
all $n\in\N$, and then $\lim_{n\to\infty}(u_n-v_n)\in X_{+}$ because
$X_{+}$ is closed, then Property {\it(\ref{L:o-2iii})} follows.
\end{proof}

The remaining order relations have some other interesting properties.
\begin{lemma}
\label{L:o-3}
Let $X$, $X_{+}$ be an ordered Banach space. Then, for all $u$, $v$,
$w \in X$, and $a\in \R$, hold: (i) If $u\gg v$ and $v\gg w$, then $u
\gg w$; (ii) If $u\gg v$ and $v\geqs w$, then $u \gg w$; (iii) If
$u\geqs v$ and $v\gg w$, then $u \gg w$; (iv) If $u \gg v$ and $a > 0$,
then $au \gg av$.
\end{lemma}
The Proof of Lemma~\ref{L:o-3} is similar to the previous Lemma, and
is not reproduced here. Given an ordered Banach space $X$, $X_{+}$,
and two elements $u\geqs v$, introduce the intervals
\[
[v,u] := \{w\in X : v\leqs w \leqs u \},\qquad
(v,u) := \{w\in X : v \ll w \ll u \}.
\]
Analogously, introduce the intervals $[v,u)$ and $(v,u]$. See
Fig.~\ref{F:cones} for an example in $X=\R^2$.
\begin{figure}[htb]
\begin{center}
\includegraphics[height=3cm,width=4cm]{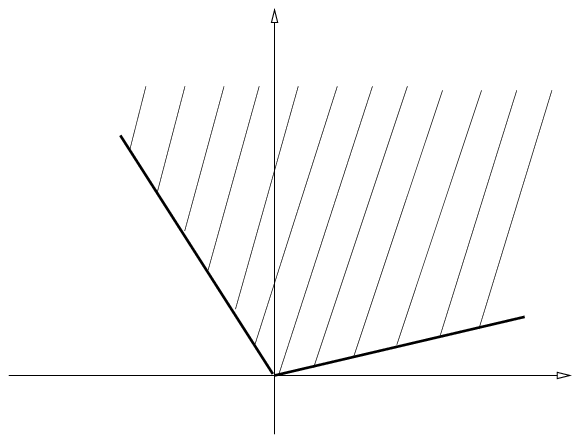}
\includegraphics[height=3cm,width=4cm]{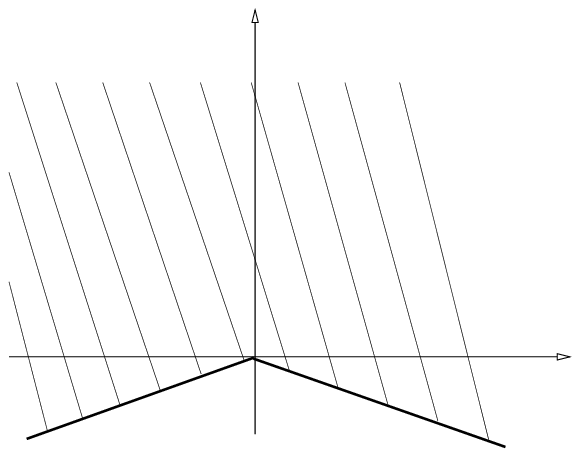}
\includegraphics[height=3cm,width=4cm]{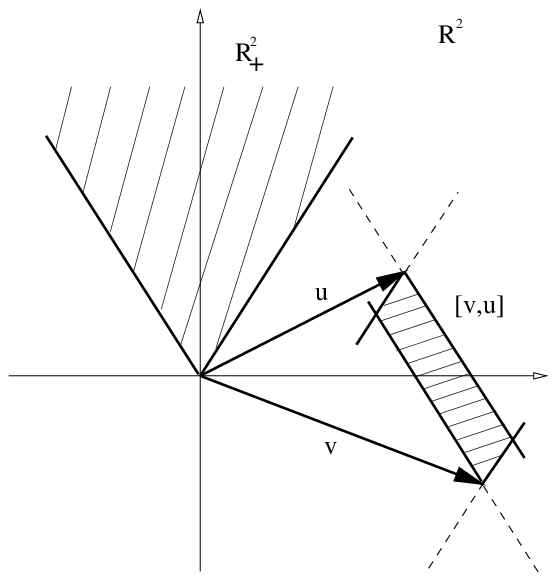}
\end{center}
\caption{The shaded regions in the first picture represents an
order cone, while the second picture represents a cone that is not an
order cone. The shaded region between $u$ and $v$ in the third
picture represents the closed interval $[v,u]$, constructed with the 
order cone $\R^2_{+}$, which is also represented by a shaded region.}
\label{F:cones}
\end{figure}
Useful order cones for solving PDE are those that define an order
structure in the Banach space which is related with the norm and the
notion of boundedness. These type of order cones are called normal. More
precisely, an order cone $X_{+}$ in a Banach space $X$ is called {\bf
normal order cone} iff there exists $0<a\in\R$ such that for all $u$,
$v\in X$ with $0 \leqs v\leqs u$ holds $\|v\| \leqs a\, \|u\|$.

\begin{lemma}
\label{L:normal}
If $X$, $X_{+}$ is an ordered Banach space with normal order cone
$X_{+}$, then every closed interval in $X$ is bounded.
\end{lemma}

\begin{proof}
Let $w\in [v,u]$, then $v\leqs w\leqs u$, and so $0\leqs w-v \leqs
u-v$. Since the cone $X_{+}$ is normal, this implies that there exists
$a>0$ such that $\|w-v\|\leqs a\,\|u-v\|$. Then, the inequalities
$\|w\|\leqs\|w-v\| +\|v\|\leqs a\,\|u-v\| +\|v\|$, which hold for all
$w\in [v,u]$, establish the Lemma.
\end{proof}

Not every order cone is normal. For example, consider the Sobolev
spaces $W^{k,p}$ of scalar-valued functions on an $n$-dimensional,
closed manifold $\cM$ 
(or a compact manifold with Lipschitz continuous boundary), 
where $k$ is a non-negative integer, and $p>1$ is a real number. An order cone
in $W^{k,p}$ is defined translating the order on the real numbers,
almost everywhere in $\cM$, that is,
\[
W^{k,p}_{+} :=\{ u\in W^{k,p} : u \geqs 0 \mbox{ a.e. in } \cM\}.
\]
In the case $k=0$, that is, we have $W^{0,p} = L^p$, the order cone
above is a normal cone \cite{hA76,Zeidler-I}. However, in the case
$k\geqs 1$ the cone above cannot be normal, since on the one hand, the
cone definition involves information only of the values of $u(x)$ and
not of its derivatives; on the other hand, the norm in $W^{k,p}$
contains information of both the values of $u(x)$ and its derivatives.
In the case of a compact manifold with boundary,
since there are no boundary conditions on $\partial\cM$ in the
definition of $W^{k,p}$, there is no way to relate the values of a
function in $\cM$ with the values of its derivatives. (In other words,
there is no Poincar\'e inequality for elements in $W^{k,p}$, with
$k\geqs 1$.)

An order cone $X_{+}\subset X$ is {\bf generating} iff $\Span (X_{+})
= X$. An order cone $X_{+}\subset X$ is called {\bf total} iff
$\Span(X_{+})$ is dense in $X$. Total order cones are important
because the order structure associated with them can be translated
from the space $X$ into its dual space $X^{*}$.
\begin{lemma}
\label{L:dual}
Let $X$, $X_{+}$ be an ordered Banach space. If $X_{+}$ is a total
order cone, then an order cone in $X^{*}$ is given by the set
$X^{*}_{+} \subset X^{*}$ defined as
\[
X^{*}_{+} := \{ u^{*} \in X^{*} : u^{*}(v) \geqs 0 \quad
\forall \, v\in X_{+} \}.
\]
\end{lemma}

\begin{proof}
We check the three properties in the definition of the order cone. The
first property is satisfied because $X_{+}$ is an order cone, so there
exists $v\neq 0$ in $X_{+}$, and then there exists $u^{*}\neq 0$ in
$X^{*}$ such that $u^{*}(v)=1 \geqs 0$, so $X^{*}_{+}$ is
non-empty. Trivially, $0\in X^{*}_{+}$. Finally, $X^{*}_{+}$ is closed
because the order relation $\geqs$ for real numbers is used in its
definition. The second property of an order cone is satisfied, because
given any $u^{*}$, $v^{*}\in X^{*}_{+}$ and any non-negative $a$,
$b\in\R$, then for all $\un u\in X_{+}$ holds
\[
(au^{*}+bv^{*})(\un u) = 
au^{*}(\un u) + bv^{*}(\un u) \geqs 0
\]
since each term is non-negative. This implies that $(au^{*}+bv^{*})\in
X^{*}_{+}$. The third property is satisfied because the order cone
$X_{+}$ is total. Suppose that the element $u^{*}\in X^{*}_{+}$ and
$-u^{*}\in X^{*}_{+}$, then for all $\un u\in X_{+}$ holds that
$u^{*}(\un u) \geqs 0$ and $-u^{*}(\un u) \geqs 0$, which implies that
$u^{*}(\un u) =0$ for all $\un u\in X_{+}$. Therefore, $u^{*}\in
X_{+}^{\circ}\subset X^{*}$, where the super-script $\circ$ in
$X^{\circ}_{+}$ means the Banach annihilator of the set $X_{+}$,
which is a subset of the space $X^{*}$. Therefore, we conclude that
$u^{*}\in\bigl[\Span (X_{+})\bigr]^{\circ}$. Since the order cone is
total, $\overline{\Span (X_{+})} = X$, that implies $\bigl[\Span
(X_{+})\bigr]^{\circ} = \{0\}$, so $u^{*}=0$. This establishes the Lemma.
\end{proof}

An order cone $X_{+}$ in a Banach space $X$ is called a {\bf solid
cone} iff $X_{+}$ has non-empty interior. The following result asserts
that solid order are generating. We remark that the converse is not
true. In the examples below we present function spaces frequently used
in solving PDE with order cones having empty interior which are indeed
generating.
\begin{lemma}
\label{L:int}
Let $X$, $X_{+}$ be an order Banach space. If $X_{+}$ is a solid cone,
then $X_{+}$ is generating.
\end{lemma}

\begin{proof}
The cone $X_{+}$ has a non-empty interior, so there exists
$x_0\in\mbox{int}(X_{+})$ and $x_0\neq 0$. This means that given any
$x\in X$ there exists $0<a\in\R$ small enough such that both $x_{+} :=
x_0 + a x$ and $x_{-} := x_0 - ax$ belong to $\mbox{int}(X_{+})$. But
then, $x = (x_{+} - x_{-})/(2a)$, so $x\in\Span(X_{+})$. This
establishes the Lemma. 
\end{proof}

Here is a list of examples of several order cones used in function
spaces. All these examples use order cones obtained from the usual
order in $\R$. In particular, they refer to scalar-valued functions on
an $n$- dimensional, closed manifold $\cM$ 
(or a compact manifold with Lipschitz boundary).
\begin{itemize}
\item 
Introduce on $C^k$ the cone
$
C^k_{+} := \{u\in C^k : u(x) \geqs 0 \quad \forall x \in \cM\}.
$
This is an order cone for all non-negative integer $k$. The cone is a
normal cone in the particular case $k=0$. The cone is solid for all
$k\geqs 0$, therefore it is a generating cone.

\item 
Introduce on $L^{\infty}$ the cone
$L^{\infty}_{+} := \{u\in L^{\infty} : u \geqs 0 
\quad \mbox{a.e. in } \cM\}$.
This is a normal, order cone. It is a solid cone, therefore is
generating.

\item 
Introduce on $W^{k,\infty}$ the cone
$W^{k,\infty}_{+} := \{u\in W^{k,\infty} : u \geqs 0 
\quad \mbox{a.e. in } \cM\}$.
This is an order cone. It is not normal for $k\geqs 1$. The cone is
solid, therefore it is generating.

\item 
Introduce on $L^p$ the cone
$L^p_{+} := \{u\in L^p : u \geqs 0 \quad \mbox{a.e. in } \cM\}$.
This is a normal, order cone every real numbers $p\geqs 1$. The cone
is not solid, however it is a generating cone.

\item 
Introduce on $W^{k,p}$ the cone
$W^{k,p}_{+} := \{u\in W^{k,p} : u \geqs 0 \quad \mbox{a.e. in } \cM\}$.
This is an order cone every real numbers $p\geqs 1$. The cone is not
normal for $k\geqs 1$. The cone is not solid for $kp \leqs n$, and it
is solid for $kp > n$. In both cases, the cone is generating.
\end{itemize}


A key concept that becomes possible in ordered Banach spaces is that of 
an operator satisfying a maximum principle.
We have not seen in the literature an approach to maximum principles on 
ordered Banach spaces in the generality we now present.
Let $X$, $X_{+}$ and $Y$, $Y_{+}$ be ordered Banach
spaces. An operator $A : D_A \subset X\to Y$ satisfies the {\bf
maximum principle} iff for every $u$, $v\in D_A$ such that $Au-Av\in
Y_{+}$ holds that $u-v\in X_{+}$. In the particular case that the
operator $A$ is linear, then it satisfies the maximum principle iff
for all $u\in X$ such that $Au\in Y_{+}$ holds that $u\in X_{+}$. The
main example is the Laplace operator acting on scalar-valued functions
defined on different domains. It is shown later on in this Appendix
that the inverse of an operator that satisfies the maximum principle
is monotone increasing. The following result gives a simple sufficient
condition for an operator to satisfy the maximum principle. This
result is useful on weak formulations of PDE.
\begin{lemma}
\label{L:OBS-MP-suff}
Let $X$, $X_{+}$ be an ordered Banach space, and $A: X\to X^{*}$ be a
linear and coercive map. Assume that $X_{+}$ is a generating order
cone, and that for all $u\in X$ such that $Au\in X_{+}^{*}$ there
exists a decomposition $u= u^{+} - u^{-}$ with $u^{+}$, $u^{-}\in
X_{+}$ that also satisfies $Au^{+}(u^{-}) =0$. Then, the operator $A$
satisfies the maximum principle.
\end{lemma}

\begin{proof}
Since the order cone $X_{+}$ is generating, the space $X^{*}$ is also
an ordered Banach space. Denote its order cone by $X^{*}_{+}$.  The
assumption that the order cone $X_{+}$ is generating also implies that
for any element $u\in X$ there exists a decomposition $u = u^{+} -
u^{-}$ with $u^{+}$, $u^{-}\in X_{+}$. By hypothesis, there exists at
least one decomposition with the extra property that
$Au^{+}(u^{-})=0$.  Now, by definition of the order in the space
$X^{*}$ we have that
\[
Au\in X_{+}^{*} \LRI Au(\un u)\geqs 0 \quad \forall \, \un u \in X_{+}.
\]
Pick as test function $\un u = u^{-}$. Then,
\[
0 \leqs Au(u^{-}) = A(u^{+} - u^{-})(u^{-}) 
=Au^{+}(u^{-}) - Au^{-}(u^{-}) = -Au^{-}(u^{-}),
\]
where the last equality comes from the condition
$Au^{+}(u^{-})=0$. Therefore, we have
\[
Au^{-}(u^{-}) \leqs 0 \RI u^{-}=0,
\]
because $A$ is coercive. So we showed that $u=u^{+} \in X_{+}$. This
establish the Lemma.\end{proof}

An example is the weak form of the shifted Laplace-Beltrami operator
$\Delta + s$ on scalar functions on a closed manifold $\cM$, where $s>0$.
Consider the case $X = W^{1,2}$, with
$Y = X^{*} = W^{-1,2}$, and $X_{+}=W^{1,2}_{+}$, while $Y_{+} =
W^{-1,2}_{+}$. The Laplace operator in this case is given by $A:X \to
X^{*}$ with action $Au(v):=(\nabla u,\nabla v)$.  It is not difficult
to check that this operator satisfies the hypothesis in
Lemma~\ref{L:OBS-MP-suff}. Therefore, this operator satisfies the
maximum principle, that is, $Au\in W^{-1,2}_{+}$ implies $u\in
W^{1,2}_{+}$, that is, $u\geqs 0$ a.e. in the manifold $\cM$.

\subsection{Monotone increasing maps}
\label{sec:monotone}

Let $X$, $X_{+}$ and $Y$, $Y_{+}$ be two ordered Banach spaces. An
operator $F:X\to Y$ is {\bf monotone increasing} iff for all $x$, $\un
x\in X$ such that $x-\un x \in X_{+}$ holds that $F(x) -F(\un x)\in
Y_{+}$. An operator $F:X\to Y$ is {\bf monotone decreasing} iff for
all $x$, $\un x\in X$ such that $x-\un x \in X_{+}$ holds that
$-\bigl[F(x) -F(\un x)\bigr]\in Y_{+}$. 
The main result for these types of maps is the following;
it can be found as Theorem~7.A in \cite{Zeidler-I}, page
283, and Corollary~7.18 on page 284. 
We reproduce it here for completeness, without the proof.
\begin{theorem}{\bf (Fixed point for increasing operators)}
\label{T:FPI}
Let $X$ be an ordered Banach space, with a normal order cone $X_{+}$.
Let $T:[x_{-}, x_{+}]\subset X\to X$ be a monotone increasing, compact
map. If $-\bigl[x_{-} - T(x_{-})\bigr]\in X_{+}$ and $x_{+}-
T(x_{+})\in X_{+}$, then the iterations
\begin{align*}
x_{n+1} &:= T(x_n),\qquad x_0 = x_{-},\\
\hat x_{n+1} &:= T(\hat x_n),\qquad \hat x_0 = x_{+},
\end{align*}
converge to $x$ and $\hat x\in [x_{-},x_{+}]$, respectively, and the
following estimate holds,
\begin{equation}
\label{eqn:bounds}
x_{-} \leqs x_n \leqs x 
\leqs \hat x \leqs \hat x_n \leqs x_{+},
\qquad \forall n=\N.
\end{equation}
\end{theorem}

We are interested in the following class of nonlinear problems:
Find an element $x\in X$ which solves the equation
\begin{equation}
\label{eqn:affine}
Ax + F(x)=0,
\end{equation}
where the principal part involves an
invertible linear operator $A:X \to Y$ satisfying the maximum principle, 
and the non-principal part involves a nonlinear operator
$F:X\to Y$ which has monotonicity properties.
We now establish some basic results for this class of problems.
The first two results relate linear, invertible operators that satisfy the
maximum principle with monotone increasing (decreasing) operators.
\begin{lemma}
\label{L:Linverse}
Let $X$, $X_{+}$ and $Y$, $Y_{+}$ be two ordered Banach spaces.  Let
$A:X\to Y$ be a linear, invertible operator satisfying the maximum
principle. Then, the inverse operator $A^{-1}:Y\to X$ is monotone
increasing.
\end{lemma}

\begin{proof}
Let $y$, $\un y\in Y$ be such that $y-\un y\in Y_{+}$. Then,
\[
A\bigl(A^{-1}(y-\un y)\bigr)\in Y_{+} \RI
A^{-1}(y-\un y)\in X_{+} \LRI
A^{-1}y -A^{-1}\un y \in X_{+}.
\]
This establishes that the operator $A^{-1}$ is monotone increasing.
\end{proof}

\begin{lemma}
\label{L:T-decreasing}
Let $X$, $X_{+}$ and $Y$, $Y_{+}$ be two ordered Banach spaces. Let
$A:X\to Y$ be a linear, invertible operator satisfying the maximum
principle. Let $F:X\to Y$ be a monotone decreasing (increasing)
operator. Then, the operator $T:X\to X$ given by $T:= -A^{-1}F$ is
monotone increasing (decreasing).
\end{lemma}

\begin{proof}
Assume first that the operator $F$ is monotone decreasing.  So, given
any $x$, $\un x\in X$ such that $x-\un x\in X_{+}$, the following
inequalities hold,
\begin{align*}
x-\un x \in X_{+} &\RI -\bigl[F(x) - F(\un x)\bigr] \in Y_{+},\\
&\LRI  A\bigl(-A^{-1}\bigl[ F(x) - F(\un x) \bigr]\bigr)\in Y_{+},\\
&\RI -A^{-1}\bigl[ F(x) - F(\un x) \bigr]\in X_{+},\\
&\LRI -\bigl[ A^{-1}F(x) -A^{-1}F(\un x) \bigr]\in X_{+},\\
&\LRI T(x) - T(\un x) \in X_{+},
\end{align*}
which establishes that the operator $T$ is monotone increasing.  In
the case that the operator $F$ is monotone increasing, then the first
line in the proof above changes into $x -\un x\in X_{+}$ implies that
$F(x) - F(\un x)\in Y_{+}$, and then all the remaining inequalities in
the proof above are reverted. This establishes the Lemma.
\end{proof}

The next result translates the inequalities that satisfy sub- and
super-so\-lu\-tions to the equation $Ax+F(x)=0$, into inequalities for
the operator $T=-A^{-1}F$.
\begin{lemma}
\label{L:T-sub-super}
Assume the hypothesis in Lemma~\ref{L:T-decreasing}. 

If there exists an element $x_{+}\in X$ such that $Ax_{+}+F(x_{+})\in
Y_{+}$, then this element satisfies that
$x_{+}-T(x_{+})\in X_{+}$.

If there exists an element $x_{-}\in X$ such that $-\big[
Ax_{-}+F(x_{-})\bigr]\in Y_{+}$, then this element satisfies that
$-\bigl[x_{-}-T(x_{-})\bigr]\in X_{+}$.
\end{lemma}

\begin{proof}
The first statement in the Lemma can be shown as follows,
\begin{align*}
Ax_{+}+F(x_{+}) \in Y_{+} &\Leftrightarrow
A\bigl( x_{+} +A^{-1}F(x_{+})\bigr) \in Y_{+}\\
&\Rightarrow  x_{+} +A^{-1}F(x_{+}) \in X_{+},
\end{align*}
which then establishes that $x_{+}-T(x_{+})\in X_{+}$.  In a similar
way, the second statement in the Lemma can be shown as follows,
\begin{align*}
-\bigl[Ax_{-}+F(x_{-})\bigr] \in Y_{+} &\Leftrightarrow
A\bigl(- x_{-} -A^{-1}F(x_{-})\bigr) \in Y_{+}\\
&\Rightarrow  -x_{-} -A^{-1}F(x_{-}) \in X_{+},
\end{align*}
which then establishes that $-\bigl[x_{-}-T(x_{-})\bigr]\in X_{+}$.
This establishes the Lemma.
\end{proof}

For nonlinear problems of the form~(\ref{eqn:affine}), one can use
Theorem~\ref{T:FPI} for monotone nonlinearities to conclude the
following.

\begin{corollary}
{\bf (Semi-linear equations with sub-/super-solutions)}
\label{C:equiv2}
Let $X$, $X_{+}$ and $Y$, $Y_{+}$ be two ordered Banach spaces where
$X_{+}$ is a normal order cone. Let $A: X\to Y$ be a linear,
invertible operator satisfying the maximum principle. Let $x_{+}$,
$x_{-}\in X$ be elements such that $(x_{+}-x_{-})\in X_{+}$, and then
assume that the operator $F:[x_{-}, x_{+}]\subset X\to Y$ is monotone
decreasing and compact. If the elements $x_{-}$ and $x_{+}$ satisfy
the relations
\begin{equation}
\label{ss-sol2}
-\bigl[ Ax_{-} + F(x_{-})\bigr] \in Y_{+},\qquad
Ax_{+} + F(x_{+}) \in Y_{+},
\end{equation}
then there exists a solution $x\in [x_{-},x_{+}]\subset X$ of the
equation $Ax + F(x)=0$.
\end{corollary}

\begin{proof}
The operator $A$ is invertible, then rewrite the equation $Ax+F(x)=0$
as a fixed-point equation,
\begin{equation}
\label{eqn:fixed-point2}
x = -A^{-1}F(x) =: T(x).
\end{equation}
By Lemma \ref{L:T-decreasing}, we know that the map $T :X\to X$ is
monotone increasing. Moreover, this operator $T$ it is compact, since
is the composition of the continuous mapping $-A^{-1}$ and the compact
map $F$. The elements $x_{-}$ and $x_{+}$ satisfy Eq.~(\ref{ss-sol2}),
therefore, by Lemma~\ref{L:T-sub-super}, they are also sub- and
super-solutions for the fixed-point equation involving the map $T$. It
follows from Theorem \ref{T:FPI} that there exists an element $x\in X$
solution to the fixed-point equation~(\ref{eqn:fixed-point2}), and
this solution satisfies the bounds $x_{-}\leqs x\leqs x_{+}$.
\end{proof}

\subsection{Sobolev spaces on closed manifolds}
\label{sec:Sobolev}

In this appendix we will recall some properties of Sobolev spaces of sections 
of vector bundles over closed manifolds.
The following definition makes precise what we mean by fractional order 
Sobolev spaces.
We expect that without much difficulty all the results in this paper can be 
modified to reflect other smoothness classes such as Bessel potential spaces 
or general Besov spaces.

\begin{definition}\label{d:sob}
For $s\geqs0$ and $1\leqs p\leqs\infty$, we denote by $W^{s,p}(\R^n)$ the space of all distributions $u$ defined in $\R^n$, such that
\begin{itemize}
\item[(a)] when $s=m$ is an integer,
\begin{equation*}
\|u\|_{m,p}=\sum_{|\nu|\leqs m}\|\partial^\nu u\|_{p}<\infty,
\end{equation*}
where $\|\cdot\|_{p}$ is the standard $L^p$-norm in $\R^n$;
\item[(b)] and when $s=m+\sigma$ with $m$ (nonnegative) integer and $\sigma\in(0,1)$,
\begin{equation*}
\|u\|_{s,p}=\|u\|_{m,p}+\sum_{|\nu|=m}\|\partial^\nu u\|_{\sigma,p}<\infty;
\end{equation*}
where 
\begin{equation*}
\|u\|_{\sigma,p}=\left(\iint_{\R^n\times\R^n}\frac{|u(x)-u(y)|^p}{|x-y|^{n+\sigma p}}dxdy\right)^{\frac1p},
\qquad\textrm{for }1\leqs p<\infty,
\end{equation*}
and
\begin{equation*}
\|u\|_{\sigma,\infty}=\mathrm{ess~sup}_{x,y\in\R^n}\,\frac{|u(x)-u(y)|}{|x-y|^{\sigma}}.
\end{equation*}
\end{itemize}
For $s<0$ and $1<p<\infty$,  $W^{s,p}(\R^n)$ denotes the topological dual of $W^{-s,p'}(\R^n)$, where $\frac1p+\frac1{p'}=1$.
\end{definition}

These well known spaces are Banach spaces with corresponding norms, and become Hilbert spaces when $p=2$.
We refer to \cite{Gris85,Trie83} and references therein for further properties.

Now we will define analogous spaces on closed manifolds.
Let $\cM$ be an $n$-dimensional smooth closed manifold, and let $\{(U_i,\varphi_i)\}$ be a collection of charts such that $\{U_i\}$ forms a finite cover of $\cM$.
Then for any distribution $u\in C^\infty_0(U_i)^*$, the pull-back $\varphi_{i}^*(u)\in C^\infty_0(\varphi_{i}(U_i))^*$ is defined by $\varphi_{i}^*(u)(v)=u(v\circ\varphi_i)$ for all $v\in C^\infty_0(\varphi_{i}(U_i))$.
Extending $\varphi_{i}^*(u)$ by zero outside $\varphi_{i}(U_i)$, in the following we treat it as a distribution on $\R^n$.
Let $\{\chi_i\}$ be a smooth partition of unity subordinate to $\{U_i\}$.

\begin{definition}\label{d:sobm}
For $s\in\R$ and $p\in(1,\infty)$,
we denote by $W^{s,p}(\cM)$ the space of all distributions $u$ defined in $\cM$, such that
\begin{equation}\label{e:sobnormpu}
\|u\|_{s,p}=\sum_{i}\|\varphi_{i}^*(\chi_iu)\|_{s,p}<\infty,
\end{equation}
where the norm under the sum is the $W^{s,p}(\R^n)$-norm.
In case $s\geqs0$, these Sobolev spaces can also be defined for $p=1$ and $p=\infty$.
\end{definition}

We collect the most basic properties of these spaces in the following lemma.
Recall that a Riemannian metric on $\cM$ induces a volume form on $\cM$, 
so that $L^p$ spaces can be defined on $\cM$
(cf.~\cite{Rosenberg97}).

\begin{lemma}\label{l:basob}
Either let $s\geqs0$ and $p\in[1,\infty]$ or let $s<0$ and $p\in(1,\infty)$. Then the space $W^{s,p}(\cM)$ is a Banach space. It is independent of the choice of the covering charts $\{(U_i,\varphi_i)\}$ and the partition of unity $\{\chi_i\}$. 
In particular, the different norms \eqref{e:sobnormpu} are equivalent.
Moreover, the followings are true when $\cM$ is equipped with a smooth Riemannian metric.

(a) Let $\nabla$ be the Levi-Civita connection associated to the Riemannian metric. Then for any nonnegative integer $m$,
\begin{equation*}
\|u\|_{m,p}'=\sum_{i=0}^m\|\nabla^iu\|_{p},
\end{equation*}
is an equivalent norm on $W^{m,p}(\cM)$.
In particular, we have $W^{0,p}(\cM)=L^p(\cM)$.

(b) Identifying $C^\infty(\cM)$ as a subspace of distributions via the $L^2$-inner product, $C^{\infty}(\cM)$ is densely embedded in $W^{s,p}(\cM)$ for any $s\in\R$ and $p\in(1,\infty)$.

(c) Let $s\in\R$ and $p\in(1,\infty)$.
Then the $L^2$-inner product on $C^\infty(\cM)$ extends uniquely to a continuous bilinear pairing $W^{s,p}(\cM)\otimes W^{-s,p'}(\cM)\to\R$,
where $\frac1p+\frac1{p'}=1$.
Moreover, the pairing induces a topological isomorphism between $W^{-s,p'}(\cM)$ and the topological dual space of $W^{s,p}(\cM)$.
\end{lemma}
\begin{proof}
See for example~\cite{Aubin82,Hebey96,Rosenberg97,Schwarz95}.
\end{proof}

A main goal of this subsection is to extend the previous lemma to the case when the Riemannian metric is not smooth.
The following result will be of importance.

\begin{lemma}\label{l:sob-hol}
Let $s_i\geqs s$ with $s_1+s_2\geqs0$, and $1\leqs p,p_i\leqs\infty$ ($i=1,2$) be real numbers satisfying
\begin{equation*}
s_i-s\geqs n\left(\frac1{p_i}-\frac1p\right),
\qquad
s_1+s_2-s > n\left(\frac1{p_1}+\frac1{p_2}-\frac1p\right),
\end{equation*}
where the strictness of the inequalities can be interchanged if $s\in\N_0$.
In case $\min(s_1,s_2)<0$, in addition let $1<p,p_i<\infty$, and let
\begin{equation*}
s_1+s_2\geqs n\left(\frac1{p_1}+\frac1{p_2}-1\right).
\end{equation*}
Then, the pointwise multiplication of functions extends uniquely to a continuous bilinear map
\begin{equation*}
W^{s_1,p_1}(\cM)\otimes W^{s_2,p_2}(\cM)\rightarrow W^{s,p}(\cM).
\end{equation*}
\end{lemma}

\begin{proof}
A proof is given in \cite{jZ77} for the case $s\geqs0$, and by using a duality argument one can easily extend the proof to negative values of $s$.
\end{proof}

Some important special cases are considered in the following corollary.

\begin{corollary}\label{c:alg}
(a) If $p\in(1,\infty)$ and $s\in(\frac{n}p,\infty)$, then $W^{s,p}$ is a Banach algebra.
Moreover, if in addition $q\in(1,\infty)$ and $\sigma\in[-s,s]$ satisfy $\sigma-\frac{n}q\in[-n-s+\frac{n}p,s-\frac{n}p]$,
then the pointwise multiplication is bounded as a map $W^{s,p}\otimes W^{\sigma,q}\rightarrow W^{\sigma,q}$.

(b) Let $1<p,q<\infty$ and $\sigma\leqs s\geqs0$ satisfy $\sigma-\frac{n}q<2(s-\frac{n}p)$ and $\sigma-\frac{n}q\leqs s-\frac{n}p$.
Then the pointwise multiplication is bounded as a map $W^{s,p}\otimes W^{s,p}\rightarrow W^{\sigma,q}$. 
\end{corollary}

The following lemma is proved in \cite{dM05} for the case $p=q=2$.
With the help of Lemma \ref{l:sob-hol}, the proof can easily be adapted to the following general case.

\begin{lemma}\label{l:nem}
Let $p\in(1,\infty)$ and $s\in(\frac{n}p,\infty)$, and let $u\in W^{s,p}$.
Let $\sigma\in[-1,1]$ and $\frac1q\in(\frac{1+\sigma}2\delta,1-\frac{1-\sigma}2\delta)$, and let  $v\in W^{\sigma,q}$, where $\delta=\frac1p-\frac{s-1}n$.
Moreover, let $f:[\inf u,\sup u]\rightarrow\R$ be a smooth function.
Then, we have
\begin{equation*}
\|v(f\circ u)\|_{\sigma,q}
\leqs C\,\|v\|_{\sigma,q}
\left(\|f\circ u\|_{\infty}+\|f'\circ u\|_{\infty}\|u\|_{s,p}\right),
\end{equation*}
where the constant $C$ does not depend on $u$, $v$ or $f$.
\end{lemma}

\begin{proof}
We consider the case $\sigma=1$ first. 
Choosing a smooth Riemannian metric on $\cM$, we have
\begin{equation*}
\begin{split}
\|v(f\circ u)\|_{1,q}
&\leqs 
C \left(\|v(f\circ u)\|_{q}+\|\nabla[v(f\circ u)]\|_{q}\right)\\
&\leqs 
C \left(\|v(f\circ u)\|_{q}+\|(\nabla v)(f\circ u)\|_{q}+\|v(f'\circ u)\nabla u\|_{q}\right)\\
&\leqs 
C \left(\|v\|_{q}\|f\circ u\|_{\infty}+\|v\|_{1,q}\|f\circ u\|_{\infty}+\|f'\circ u\|_{\infty}\|v\nabla u\|_{q}\right).
\end{split}
\end{equation*}
By Lemma \ref{l:sob-hol}, for $\frac1q\geqs\delta$, the last term can be bounded as
\begin{equation*}
\|v\nabla u\|_{q}
\leqs
C \|v\|_{1,q}\|\nabla u\|_{s-1,p}
\leqs
C \|v\|_{1,q}\|u\|_{s,p},
\end{equation*}
proving the lemma for the case $\sigma=1$.
By using duality one proves the case $\sigma=-1$ and $\frac1q\leqs1-\delta$,
and the lemma follows from interpolation.
\end{proof}

Let $\cM$ be an $n$-dimensional smooth closed manifold, and let $E\to\cM$ be a smooth vector bundle over $\cM$.
Analogously to Definition \ref{d:sobm}, we define the Sobolev space $W^{s,p}(E)$ of sections of $E$ by utilizing a finite trivializing cover of coordinate charts, a partition of unity subordinate to the cover, and the space $[W^{s,p}(\R^n)]^k$ of vector functions, where $k$ is the fiber dimension of $E$.
Then, Lemma \ref{l:basob} holds for these spaces with obvious modifications.
When there is no risk of confusion, we will omit the explicit specification of the vector bundle $E$ from the notation $W^{s,p}(E)$.

In the following lemma we consider nonsmooth Riemannian structures on $E$ and nonsmooth volume forms on $\cM$.

\begin{lemma}\label{l:rough-L2}
Let $\gamma\in(1,\infty)$ and $\alpha\in(\frac{n}\gamma,\infty)$.
Fix on $\cM$ a volume form of class $W^{\alpha,\gamma}$, and on $E$ a Riemannian structure of class $W^{\alpha,\gamma}$.

(a) Let $p\in(1,\infty)$ and  $s\leqs\min\{\alpha,\alpha+n(\frac1p-\frac1\gamma)\}$.
Then identifying the space $C^\infty(E)$ of smooth sections of $E$ as a subspace of distributions via the $L^2$-inner product, $C^{\infty}(E)$ is densely embedded in $W^{s,p}(E)$.

(b) Let $s\in[-\alpha,\alpha]$, $p\in(1,\infty)$, and $s-\frac{n}p\in[-n-\alpha+\frac{n}\gamma,\alpha-\frac{n}\gamma]$.
Then the $L^2$-inner product on $C^\infty(E)$ extends uniquely to a continuous bilinear pairing $W^{s,p}(E)\otimes W^{-s,p'}(E)\to\R$,
where $\frac1p+\frac1{p'}=1$.
Moreover, the pairing induces a topological isomorphism $[W^{s,p}(E)]^*\cong W^{-s,p'}(E)$.
\end{lemma}

\begin{proof}
We will prove the lemma for scalar functions on $\cM$, i.e., for the trivial bundle $E=\cM\times\R$.
The general case is only more technical.

Fixing a smooth volume form on $\cM$ and denoting the associated $L^2$-inner product by $(\cdot,\cdot)_*$, the $L^2$-inner product associated to the nonsmooth volume form (and the nonsmooth metric on $\cM\times\R$) satisfies
\begin{equation*}
(u,v)_{L^2}=(hu,v)_*,
\qquad u,v\in C^\infty(\cM),
\end{equation*}
with some strictly positive function $h\in W^{\alpha,\gamma}$.
From Lemma \ref{l:sob-hol}, we have that multiplication by $h$ is continuous on $W^{s,p}$ for $s\in[-\alpha,\alpha]$, $p\in(1,\infty)$, and $s-\frac{n}p\in[-n-\alpha+\frac{n}\gamma,\alpha-\frac{n}\gamma]$. Since $h>0$ this operation is invertible hence a homeomorphism on $W^{s,p}$. Now by using Lemma \ref{l:basob} we complete the proof.
\end{proof}

\begin{corollary}\label{c:adj}
Let $\gamma\in(1,\infty)$ and $\alpha\in(\frac{n}\gamma,\infty)$.
Fix on $\cM$ a volume form of class $W^{\alpha,\gamma}$, and on $E$ a Riemannian structure of class $W^{\alpha,\gamma}$.
With $s\in[-\alpha,\alpha]$, $p\in(1,\infty)$, and $s-\frac{n}p\in[-n-\alpha+\frac{n}\gamma,\alpha-\frac{n}\gamma]$,
let $A:L^p\to W^{s,p}$ be a bounded linear operator and let $A^*$ be its formal $L^2$-adjoint, i.e., let
$$
(Au,v)_{L^2}=(u,A^*v)_{L^2},\qquad\textrm{for }u,v\in C^{\infty}(E).
$$
Then, $A^*$ extends uniquely to a bounded linear map $A^*:W^{-s,p'}\to L^{p'}$, and we have
\begin{equation*}
\langle Au,v\rangle
=
\langle u,A^*v\rangle
,\qquad\textrm{for }u\in L^p(E),\,v\in W^{-s,p'}(E),
\end{equation*}
where $\langle\cdot,\cdot\rangle$ denotes the extension of the $L^2$-inner product.
\end{corollary}

\begin{proof}
This is an application of Lemma \ref{l:rough-L2}.
\end{proof}

\subsection{Elliptic operators on closed manifolds}
\label{sec:killing}

In this appendix we will state {\em a priori} estimates for  general elliptic operators in some Sobolev spaces.
Let $\cM$ be an $n$-dimensional smooth closed manifold, and let $E\to\cM$ be a smooth vector bundle over $\cM$. 

Let $C^{-\infty}(E)$ be the topological dual of the space $C^{\infty}(E)$ of smooth sections of $E$.
Then for $m\in\N$, $\alpha\in\R$, and $\gamma\in[1,\infty]$,
we define $\cD_m^{\alpha,\gamma}(E)$ to be the space of differential operators $A:C^{\infty}(E)\to C^{-\infty}(E)$
that can be written in local coordinates (trivializing $E$) as
\begin{equation*}
A=\sum_{|\nu|\leqs m} a^\nu\partial_\nu
\qquad\textrm{with }
a^\nu\in W^{\alpha-m+|\nu|,\gamma}(\R^n,\R^{k\times k}),\quad|\nu|\leqs m,
\end{equation*}
where $k$ is the fiber dimension of $E$.

One can easily verify that if the metric of a Riemannian manifold is in $W^{\alpha,\gamma}$ with $\alpha\gamma>n$, 
then both the Laplace-Beltrami operator and vector Laplacian defined in \eqref{CF-def-IL} are in the classes $\cD^{\alpha,\gamma}_2(\cM\times\R)$ and $\cD^{\alpha,\gamma}_2(T\cM)$, respectively.

\begin{lemma}\label{l:bdd-operator}
Let $A$ be a differential operator of class $\cD_m^{\alpha,\gamma}(E)$.
Then, $A$ can be extended to a bounded linear map
\begin{equation*}
A:W^{s,q}(E)\to W^{\sigma,q}(E),
\end{equation*}
for $q\in(1,\infty)$, $s\geqs m-\alpha$, and $\sigma$ satisfying
\begin{equation*}
\begin{split}
\sigma\leqs\min\{s,\alpha\}-m,&
\qquad
\sigma< s-m+\alpha-\frac{n}\gamma,\\
\sigma-\frac{n}q\leqs \alpha-\frac{n}\gamma-m,&
\quad\textrm{and}\quad
s-\frac{n}q\geqs m-n-\alpha+\frac{n}\gamma.
\end{split}
\end{equation*}
\end{lemma}

\begin{proof}
This is a straightforward application of Lemma \ref{l:sob-hol}.
\end{proof}


The Laplace-Beltrami operator and vector Laplacian are elliptic operators.
We now consider local {\em a priori} estimates for general elliptic operators.
For any subset $U\subset\cM$, the $W^{s,p}(U)$-norm is denoted by $\|\cdot\|_{s,p,U}$. 

\begin{lemma}\label{l:ell-est-loc}
Let $A\in\cD^{\alpha,\gamma}_{m}(E)$ be an elliptic operator with $\alpha-\frac{n}\gamma>\max\{0,\frac{m-n}2\}$.
Let $q\in(1,\infty)$, $s\in(m-\alpha,\alpha]$, and $s-\frac{n}q\in(m-n-\alpha+\frac{n}\gamma,\alpha-\frac{n}\gamma]$.
Then for any  $y\in\cM$, there exists a constant $c>0$ and open neighborhoods $K\subset U\subset\cM$ of $y$ such that
\begin{equation}\label{e:ell-est-loc}
c\|\chi u\|_{s,q} \leqs 
\|Au\|_{s-m,q}+\|u\|_{s-1,q,U},
\end{equation}
for any $u\in W^{s,q}(E)$ and $\chi\in C^{\infty}_0(K)$ with $\chi\geqs0$.
\end{lemma}

\begin{proof}
We work in a local chart containing $y$, which trivializes $E$.
Let $K$ be the open ball of radius $r$ centered at $y$ contained in the domain of the chart and extend the coefficients of $A$ outside $K$ so that the resulting operator is still in $\cD^{\nu,\gamma}_{m}$, with appropriate vector fields over $\R^n$.
We make the decomposition $A=L+R+B$, where $L$ is the highest order terms of $A$ with coefficients frozen at $y$, and $R$ is what remains in the highest order terms, i.e.,
\begin{equation*}
L=\sum_{|\nu|=m}a^\nu(y)\partial_\nu,
\qquad
R=\sum_{|\nu|=m}[a^\nu-a^\nu(y)]\partial_\nu.
\end{equation*}
Obviously $B=A-L-R$ is the lower order terms.
Let $u\in W^{s,q}$ with $\mathrm{supp}\,u\subset K$.
From the theory of constant coefficient elliptic operators, we infer the existence of a constant $c>0$ such that for any $u\in W^{s,q}(E)$ with $\mathrm{supp}\,u\subset K$,
\begin{equation*}
\begin{split}
c\|u\|_{s,q} 
&\leqs
\|Lu\|_{s-m,q}+\|u\|_{s-m,q}\\
&\leqs
\|Au\|_{s-m,q}+\|Ru\|_{s-m,q}+\|Bu\|_{s-m,q}+\|u\|_{s-m,q}.
\end{split}
\end{equation*}
Since $\alpha>\frac{n}\gamma$, without loss of generality we can assume for $|\nu|=m$ that $a^{\nu}\in C^{0,h}$ for some $h>0$, so
\begin{equation*}
\|Ru\|_{s-m,q}\leqs Cr^h\|u\|_{s,q},
\end{equation*}
where $C$ is a constant depending only on $A$.
By choosing $r$ so small that $Cr^h\leqs \frac{c}2$, we have
\begin{equation*}
\begin{split}\textstyle
\frac{c}2\|u\|_{s,q} 
&\leqs
\|Au\|_{s-m,q}+\|Bu\|_{s-m,q}+\|u\|_{s-m,q}.
\end{split}
\end{equation*}

Now we will work with the lower order term.
Choose $\delta\in(0,\alpha-\frac{n}\gamma)$ such that $\delta\leqs\min\{1,s+\alpha-m,s-\frac{n}q+\alpha-\frac{n}\gamma+n-m\}$.
We have $B\in\cD^{\alpha-1,\gamma}_{m-1}$, so by Lemma \ref{l:bdd-operator}, $B:W^{s-\delta,\gamma}\to W^{s-m,\gamma}$ is bounded.
Then using a well known interpolation inequality, we get
\begin{equation*}
\|Bu\|_{s-m,q}
\leqs 
C\|u\|_{s-\delta,q}
\leqs 
C\varepsilon\|u\|_{s,q}
+
C'\varepsilon^{-(m-\delta)/\delta}\|u\|_{s-m,q},
\end{equation*}
for any $\varepsilon>0$.
Choosing $\varepsilon>0$ sufficiently small, we conclude that
\begin{equation*}
c\|u\|_{s,q} \leqs 
\|Au\|_{s-m,q}+\|u\|_{s-m,q},
\qquad
\forall u\in W^{s,q}(E),\quad\mathrm{supp}\,u\subset K.
\end{equation*}
We apply the this inequality to $\chi u$, and then observing that $[A,\chi]$ is in $\cD^{\alpha,\gamma}_{m-1}(\cM)$, we obtain \eqref{e:ell-est-loc}.
\end{proof}


We can easily globalize the above result as follows.

\begin{corollary}\label{C:ell-est}
Let the conditions of Lemma \ref{l:ell-est-loc} hold.
Then there exists a constant $c>0$ such that
\begin{equation}\label{e:ell-est}
c \|u\|_{s,q} \leqs
\|Au\|_{s-m,q}+\|u\|_{s-m,q},\qquad
\forall u\in W^{s,q}(E).
\end{equation}
\end{corollary}

\begin{proof}
We first cover $\cM$ by open neighborhoods $K$ by applying Lemma \ref{l:ell-est-loc} to every point $y\in\cM$,
and then choose a finite subcover of the resulting cover.
Then a partition of unity argument gives \eqref{e:ell-est} with the term $\|u\|_{s-m,q}$ replaced by $\|u\|_{s-1,q}$,
and finally one can use an interpolation inequality to get the conclusion.
\end{proof}

Let us recall the following well known results from functional analysis.

\begin{lemma}
\label{l:semi-fred}
Let $X$ and $Y$ be Banach spaces with continuous embedding $X\hookrightarrow Y$.
Let $A:X\rightarrow Y$ be a continuous linear map.
Then
\begin{itemize}
\item[(a)] 
A necessary and sufficient condition that the graph of $A$ be closed in $X\times Y$ is that there exists a constant $c>0$ such that
$c\|u\|_{\tiX}\leqs\|Au\|_{\tiY}+\|u\|_{\tiY}$ for all $u\in X$.
\item[(b)]
If in addition the embedding $X\hookrightarrow Y$ is compact then the range of $A$ is closed and the kernel of $A$ is finite-dimensional.
\end{itemize}
\end{lemma}

As an immediate consequence, we obtain the following result.

\begin{lemma}
\label{l:ell-semi-fred}
Let $A\in\cD^{\alpha,\gamma}_{m}(E)$ be an elliptic operator with $\alpha-\frac{n}\gamma>\max\{0,\frac{m-n}2\}$.
Let $q\in(1,\infty)$, $s\in(m-\alpha,\alpha]$, and $s-\frac{n}q\in(m-n-\alpha+\frac{n}\gamma,\alpha-\frac{n}\gamma]$.
Then, the operator $A:W^{s,q}(E)\to W^{s-m,q}(E)$ is semi-Fredholm, i.e., its range is closed and the kernel is finite-dimensional.
\end{lemma}

\subsection{Maximum principles on closed manifolds}
\label{sec:maxprinciple}

In this appendix, we present maximum principles for the operators of the 
form $-\nabla \cdot (u\nabla)$ with positive function $u$, followed by a simple application.
These types of results are well known, but nevertheless we state them
here for completeness.

It is convenient at times when working with barriers and maximum
principle arguments to split real valued functions into
positive and negative parts; we will use the following notation
for these concepts:
\[
\phi^{+}:= \mbox{max}\{\phi,0\},\qquad
\phi^{-}:= -\mbox{min}\{\phi,0\},
\]
whenever they make sense.
In the proof of the following lemma we will use the fact that for $\phi\in W^{1,p}$ it holds $\phi^{+}\in W^{1,p}$ and so $\phi^{-}\in W^{1,p}$, cf. \cite{dMdZ98}.

\begin{lemma}\label{l:max-princ}
Let $p\in(1,\infty)$ and $s\in(\frac{n}p,\infty)\cap[1,\infty)$, and
let $(\cM,h_{ab})$ be an $n$-dimensional, smooth, closed manifold with a Riemannian metric $h_{ab} \in W^{s,p}$.
Moreover, let $u\in W^{s,p}$ be a function with $u>0$ and let $f\in W^{s-2,p}$.
Let $\phi\in W^{s,p}$ be such that
\begin{equation}\label{e:max1}
\langle u\nabla\phi,\nabla\varphi\rangle+\langle f,\phi\varphi\rangle\geqs0,
\qquad\textrm{for all }
\varphi\in C^{\infty}_{+}.
\end{equation}
\begin{itemize}
\item[(a)]
If $f\neq0$ and $\langle f,\varphi\rangle\geqs0$ for all $\varphi\in C^{\infty}_{+}$,
then $\phi\geqs0$.
\item[(b)]
If $\cM$ is connected and $\phi\geqs0$,
then either $\phi\equiv0$ or $\phi>0$ everywhere.
\end{itemize}
\end{lemma}

\begin{proof}
For (a), we will follow the proof of \cite[Lemma 2.9]{dM05}.
Since $\phi\in W^{1,n}$, we have $\phi^{-}\in W^{1,n}_{+}$ and $-\phi\phi^{-}\in W^{1,n}_{+}$.
Note that $W^{1,n}\hookrightarrow (W^{s-2,p})^*$ by $n\geqs2$.
Now, using the positivity of $f$ and the property \eqref{e:max1}, by density we get
\begin{equation*}
0
\geqs\langle f,\phi\phi^{-}\rangle
\geqs -\langle u\nabla\phi,\nabla\phi^{-}\rangle
=\langle u\nabla\phi^{-},\nabla\phi^{-}\rangle,
\end{equation*}
implying that $\phi^{-}=\mathrm{const}$.
So if $\phi<0$, it would have to be a negative constant.
But the property \eqref{e:max1} gives $\langle f,\varphi\rangle\leqs0$ for all $\varphi\in C^{\infty}_{+}$, which, in combination with the positivity, implies $f=0$. 
This contradicts with the hypothesis $f\neq0$ and proves (a).

Now we will prove (b).
Since $\phi$ is continuous, the level set $\phi^{-1}(0)\subset\cM$ is closed.
Following the proof of \cite[Lemma 5.3]{dM06}, we apply the weak Harnack inequality \cite[Theorem 5.2]{nT73} to show that $\phi^{-1}(0)$ is also open.
Then by connectedness of $\cM$ we will have the proof.

The weak Harnack inequality \cite[Theorem 5.2]{nT73} can be applied to second order elliptic operators of the form
\begin{equation*}
L\phi=\partial_i(a^{ij}\partial_j\phi+a^i\phi)+b^j\partial_j\phi+a\phi,
\end{equation*}
where $a^{ij}$ are continuous, and $a^i,b^j\in L^{2t}$, and $a\in L^{t}$ for some $t>\frac{n}2$.
The first term in \eqref{e:max1} satisfies these conditions, and the second term can be cast into a form satisfying the conditions (details can be found in the proof of \cite[Lemma 5.3]{dM06}).
Now suppose that $\phi(x)=0$ for some $x\in\cM$,
and let us work in local coordinates around $x$.
Then the weak Harnack inequality says that for sufficiently small $R>0$, and for some $p>t'$,
\begin{equation*}
\|\phi\|_{L^p(B(x,2R))}\leqs
C R^{\frac{n}p}\inf_{B(x,R)}\phi,
\end{equation*}
where $B(x,R)$ denotes the open ball of radius $R$ (in the background flat metric) centered at $x$,
and $C$ is a constant that depends only on $t$, $p$, and the differential operator.
Since $\phi(x)=0$ and $\phi$ is nonnegative, the infimum is zero and the inequality implies that $\phi\equiv0$ in a neighborhood of $x$.
Hence the set $\phi^{-1}(0)$ is open.
\end{proof}

\begin{lemma}\label{l:lapinv}
Let the hypotheses of Lemma \ref{l:max-princ}(b) hold,
and define the operator $L:W^{s,p}\to W^{s-2,p}$ by 
$$
\langle L\phi,\varphi\rangle=\langle u\nabla \phi,\nabla\varphi\rangle+\langle f,\phi\varphi\rangle,
\qquad\phi\in W^{s,p},\,\varphi\in C^{\infty}.
$$
Then, $L$ is bounded and invertible.
\end{lemma}

\begin{proof}
By Lemma \ref{l:ell-semi-fred}, the operator $L$ is semi-Fredholm,
and moreover since $L$ is formally self-adjoint, it is Fredholm.
It is well known that when the metric is smooth, index of $L$ is zero independent of $s$ and $p$.
We can approximate the metric $h$ by smooth metrics so that $L$ is arbitrarily close to a Fredholm operator with index zero.
Since the level sets of index as a function on Fredholm operators are open, we conclude that the index of $L$ is zero.
The injectivity of $L$ follows from Lemma \ref{l:max-princ}(a),
for if $\phi_1$ and $\phi_2$ are two solutions of $L\phi=g$, then the above lemma implies that $\phi_1-\phi_2\geqs0$ and $\phi_2-\phi_1\geqs0$.
\end{proof}

\subsection{The Yamabe classification of nonsmooth metrics}
\label{sec:yamabe}

Let $\cM$ be a smooth, closed, connected $n$-dimensional Riemannian manifold with a smooth metric $h$, where we assume throughout this section that $n\geqs3$.
With a positive scalar $\varphi$, let $\tilde{h}$ be related to $h$ by the conformal transformation
$\tilde{h}=\varphi^{2^{\star}-2}h$, where $2^{\star}=\frac{2n}{n-2}$.
We say that $\tilde{h}$ and $h$ are conformally equivalent, and this defines an equivalence relation on the space of metrics.
The equivalence class containing $h$ will be denoted by $[h]$; e.g., $\tilde{h}\in[h]$.
It is well known that any smooth Riemannian metric $h$ on a given closed connected manifold $\cM$ satisfies one and only one of the following three conditions:
\begin{itemize}
\item[$\cY^{+}$:] There is a metric in $[h]$ with strictly positive scalar curvature;
\item[$\cY^{0}$:] There is a metric in $[h]$ with vanishing scalar curvature;
\item[$\cY^{-}$:] There is a metric in $[h]$ with strictly negative scalar curvature.
\end{itemize}
These conditions define three disjoint classes in the space of metrics: they are referred to as the Yamabe classes.

We will extend the above classification to metrics in the Sobolev spaces $W^{s,p}$ under rather mild conditions on $s$ and $p$.
Since the case $p=2$ is treated in \cite{dM05} and the argument there easily extends to our slightly general setting,
we shall only sketch the proof here.
Given a Riemannian metric $h\in W^{s,p}$, let us consider the functional $E:W^{1,2}\to\R$ defined by
\begin{equation*}
E(\varphi)=(a\nabla\varphi,\nabla\varphi)+\langle R,\varphi^2\rangle,
\end{equation*}
where $a=4\frac{n-1}{n-2}$.
By Corollary \ref{c:alg}, the pointwise multiplication is bounded on $W^{1,2}\otimes W^{1,2}\to W^{\sigma,q}$
for $\sigma\leqs1$ and $\sigma-\frac{n}q<2-n$.
Putting $\sigma=2-s$ and $q=p'$, these conditions read as $2-s-\frac{n}{p'}=2-n-s+\frac{n}p<2-n$ or $s-\frac{n}p>0$, and $s\geqs1$.
So if  $sp>n$ and $s\geqs1$, $\varphi^2\in W^{2-s,p'}$ for $\varphi\in W^{1,2}$, meaning that the second term is bounded in $W^{1,2}$.

By using the functional $E$, we define the quantity
\begin{equation*}
\mu_q=\mu_q(h)=\inf_{\varphi\in B_{q}}E(\varphi),
\qquad\textrm{where }
B_q=\{\varphi\in W^{1,2}:\|\varphi\|_{q}=1\}.
\end{equation*}
Under the conditions $sp>n$ and $s\geqs1$, one can show that $\mu_q$ is finite for $q\geqs2$,
and moreover that $\mu_{2^{\star}}$ is a conformal invariant, i.e., $\mu_{2^{\star}}(h)=\mu_{2^{\star}}(\tilde{h})$ for any two metrics $\tilde{h}\in[h]$, now allowing $W^{s,p}$ functions for the conformal factor.
We refer to $\mu_{2^\star}(h)$ as the Yamabe invariant of the metric $h$, and we will see that the Yamabe classes correspond to the signs of the Yamabe invariant.

\begin{theorem}\label{t:subcrit}
Let $(\cM,h)$ be a smooth, closed, connected Riemannian manifold with dimension $n\geqs3$ and with a metric $h\in W^{s,p}$, 
where we assume $sp>n$ and $s\geqs1$.
Let $q\in[2,2^{\star})$.
Then, there exists $\phi\in W^{s,p}$, $\phi>0$ in $\cM$, such that
\begin{equation}\label{e:yamabe}
-a\Delta\phi+R\phi=\mu_q\phi^{q-1},
\qquad\textrm{and}\qquad
\|\phi\|_q=1,
\end{equation}
where $\mu_q=\mu_q(h)$ is as defined above.
\end{theorem}

\begin{proof}
The above equation is the Euler-Lagrange equation for the functional $E$, so it suffices to show that $E$ attains its infimum $\mu_q$ over $B_q$ at a positive function $\phi\in W^{s,p}$.
Let $\{\phi_i\}\subset B_q$ be a sequence satisfying $E(\phi_i)\to\mu_q$.
From the continuity of the embedding $L^q\hookrightarrow L^2$, we have $\{\phi_i\}$ is bounded in $L^2$.
It is the content of \cite[Lemma 3.1]{dM05} that
$$
E(\varphi)\geqs C_1\|\varphi\|_{1,2}^2-C_2\|\varphi\|_{2}^2,
\qquad \varphi\in W^{1,2},
$$
for metrics in $W^{s,2}$ with $s>\frac{n}2$.
The proof works verbatim for our case, and since $\mu_q$ is finite, from this we conclude that $\{\phi_i\}$ is bounded in $W^{1,2}$.
By the reflexivity of $W^{1,2}$ and the compactness of $W^{1,2}\hookrightarrow L^{q}$,
there exist an element $\phi\in W^{1,2}$ and a subsequence $\{\phi'_i\}\subset\{\phi_i\}$ such that
$\phi'_i\rightharpoonup\phi$ in $W^{1,2}$ and $\phi'_i\rightarrow\phi$ in $L^{q}$.
The latter implies $\phi\in B_q$.
It is not difficult to show that $E$ is weakly lower semi-continuous, and it follows that $E(\phi)=\mu_q$, so $\phi$ satisfies \eqref{e:yamabe}.
Bootstrapping with Corollary \ref{C:ell-est} implies that $\phi\in W^{s,p}\hookrightarrow W^{1,n}$, so that $|\phi|\in W^{1,n}$.
Since $E(|\phi|)=E(\phi)$, after replacing $\phi$ by $|\phi|$, 
we can assume that $\phi\geqs0$.
Finally, bootstrapping again gives $\phi\in W^{s,p}$,
and since $\phi\neq0$ as $\phi\in B_q$, by Lemma \ref{l:max-princ} we have $\phi>0$.
\end{proof}

Under the conformal scaling $\tilde{h}=\varphi^{2^\star-2}h$, the scalar curvature transforms as
\begin{equation*}
\tilde{R}=\varphi^{1-2^{\star}}(-a\Delta\varphi+R\varphi),
\end{equation*}
so assuming the conditions of the above theorem we infer that any given metric $h\in W^{s,p}$ 
can be transformed to the metric $\tilde{h}=\phi^{2^\star-2}h$ with the continuous scalar curvature $\tilde{R}=\mu_q\phi^{q-2^\star}$,
where the conformal factor $\phi$ is as in the theorem.
In other words, given any metric $h_{ab}\in W^{s,p}$,
there exist continuous functions $\phi \in W^{s,p}$ with $\phi>0$ and $\tilde{R}\in W^{s,p}$ having {\em constant sign}, such that
\begin{equation}\label{e:conformal-transformation}
-a\Delta\phi+R\phi = \tilde{R}\phi^{2^\star-1}.
\end{equation}
We will prove below that the conformal class of the metric $h$ completely determines the sign of $\tilde{R}$,
giving rise to the Yamabe classification of metrics in $W^{s,p}$.

In the class of smooth metrics there is a stronger result known as the Yamabe theorem: each conformal class of smooth metrics contains a metric with constant scalar curvature.
The Yamabe theorem is a non-trivial extension of the above 
theorem to the critical case $q=2^\star$,
and we see that for smooth metrics the sign of the Yamabe invariant determines 
which Yamabe class the metric is in.
A proof of the Yamabe theorem requires more delicate techniques since we lose the compactness of the embedding $W^{1,2}\hookrightarrow L^q$,
see e.g. \cite{jLtP87} for a treatment of smooth metrics.
As far as we know there has not appeared in the literature an explicit proof of the Yamabe theorem for nonsmooth metrics 
such as the ones considered in this paper, although it is generally expected to be true.
We will not pursue this issue here; however, the following simpler result justifies the Yamabe classification of nonsmooth metrics.

\begin{theorem}\label{t:yclass}
Let $(\cM,h)$ be a smooth, closed, connected Riemannian manifold with dimension $n\geqs3$ and with a metric $h\in W^{s,p}$, 
where we assume $sp>n$ and $s\geqs1$.
Then, the followings hold:
\begin{itemize}
\item[$\bullet$] $\mu_{2^{\star}}>0$ iff there is a metric in $[h]$ with continuous positive scalar curvature.
\item[$\bullet$] $\mu_{2^{\star}}=0$ iff there is a metric in $[h]$ with vanishing scalar curvature.
\item[$\bullet$] $\mu_{2^{\star}}<0$ iff there is a metric in $[h]$ with continuous negative scalar curvature.
\end{itemize}
In particular, two conformally equivalent metrics cannot have scalar curvatures with distinct signs.
\end{theorem}

\begin{proof}
We begin by proving that if there is a metric in $[h]$ with continuous scalar curvature of constant sign, then $\mu_{2^{\star}}$ has the corresponding sign.
Since $\mu_{2^\star}$ is a conformal invariant, we can assume that the scalar curvature $R$ of $h$ is continuous and has constant sign.
If $R<0$, then $E(\varphi)<0$ for constant test functions $\varphi=\mathrm{const}$ and there is a constant function in $B_{2^{\star}}$, so we have $\mu_{2^\star}<0$.
If $R\geqs0$, then $E(\varphi)\geqs0$ for any $\varphi\in W^{1,2}$, so $\mu_{2^{\star}}\geqs0$.
Taking constant test functions, we infer that $R=0$ implies $\mu_{2^\star}=0$.
Now, if $R>0$ then $E(\varphi)$ defines an equivalent norm on $W^{1,2}$, and we have $1=\|\varphi\|_{2^{\star}}\leqs C\|\varphi\|_{1,2}$ for $\varphi\in B_{2^\star}$,
so $\mu_{2^\star}>0$.

Next, we will prove that there is a metric in $[h]$ with continuous scalar curvature with the same sign as that of $\mu_{2^{\star}}$.
To this end, for any $q\in[2,2^\star)$, we shall show that the sign of $\mu_{2^\star}$ determines the sign of $\mu_{q}$, so that the proof is completed by Theorem \ref{t:subcrit}.
If $\mu_{2^\star}<0$, then $E(\varphi)<0$ for some $\varphi\in B_{2^{\star}}$, and since $E(k\varphi)=k^2E(\varphi)$ for $k\in\R$,
there is some $k\varphi\in B_{q}$ such that $E(k\varphi)<0$, so $\mu_q<0$.
If $\mu_{s^\star}\geqs0$, then $E(\varphi)\geqs0$ for all $\varphi\in B_{2^\star}$, and for any $\psi\in B_q$ there is $k$ such that $k\psi\in B_{2^\star}$,
so $\mu_q\geqs0$.
All such $k$ are uniformly bounded since $k=1/\|\psi\|_{2^\star}\leqs C/\|\psi\|_q=C$ by the continuity estimate $\|\cdot\|_{1}\leqs C\|\cdot\|_{2^{\star}}$.
From this we have for all $\psi\in B_q$, $E(\psi)=E(k\psi)/k^2\geqs\mu_{2^\star}/k^2\geqs\mu_{2^\star}/C^2$, meaning that $\mu_{2^\star}>0$ implies $\mu_q>0$.
A similar scaling argument gives that if $\mu_{2^\star}=0$ then $\mu_q=0$.
\end{proof}

\subsection{Conformal covariance of the Hamiltonian constraint}
\label{sec:conf-inv}

Let $\cM$ be a smooth, closed, connected $n$-dimensional manifold equipped with a Riemannian metric $h\in W^{s,p}$, where we assume throughout this section that $p\in(1,\infty)$, $s\in(\frac{n}p,\infty)\cap[1,\infty)$ and that $n\geqs3$.
We consider the Hamiltonian constraint
\begin{equation*}\textstyle
H(\phi):=-\Delta\phi+\frac{1}{r(n-1)}R\phi+a_{\tau}\phi^{r+1}-a_{\biw}\phi^{-r-3}-a_{\rho}\phi^{-t}=0,
\end{equation*}
where $r=\frac4{n-2}$, $t\in\R$ are constants, $R\in W^{s-2,p}$ is the scalar curvature of the metric $h$, and the other coefficients satisfy $a_{\tau},a_{\biw},a_{\rho}\in W^{s-2,p}_{+}$.
In this appendix, we will be interested in the transformation properties of $H$ under the conformal change $\tilde{h}=\theta^rh$ of the metric with the conformal factor $\theta\in W^{s,p}$ satisfying $\theta>0$.
To this end, we consider
\begin{equation*}\textstyle
\tilde{H}(\psi):=-\tilde{\Delta}\psi+\frac{1}{r(n-1)}\tilde{R}\psi+\tilde{a}_{\tau}\psi^{r+1}-\tilde{a}_{\biw}\psi^{-r-3}-\tilde{a}_{\rho}\psi^{-t}=0,
\end{equation*}
where $\tilde{\Delta}$ is the Laplace-Beltrami operator associated to the metric $\tilde{h}$,
$\tilde{R}\in W^{s-2,p}$ is the scalar curvature of $\tilde{h}$,
and at the moment we do not impose any conditions on the remaining coefficients other than that they satisfy $\tilde{a}_{\tau},\tilde{a}_{\biw},\tilde{a}_{\rho}\in W^{s-2,p}_{+}$.
One can derive the following relations
\begin{equation*}
\begin{split}
\tilde{R}=\theta^{-r}R-r(n-1)\theta^{-r-1}\Delta\theta,\\
\tilde{\Delta}\psi=\theta^{-r}\Delta\psi+2\theta^{-r-1}\nabla^a\theta\nabla_a\psi.\\
\end{split}
\end{equation*}
Combining these relations with
\begin{equation*}
\Delta(\theta\psi)=\theta\Delta\psi+\psi\Delta\theta+2\nabla^a\theta\nabla_a\psi,
\end{equation*}
we obtain
\begin{equation*}\textstyle
-\tilde{\Delta}\psi+\frac{1}{r(n-1)}\tilde{R}\psi
=\theta^{-r-1}
\left(-\Delta(\theta\psi)+\frac{1}{r(n-1)}R\theta\psi\right),
\end{equation*}
which in turn implies that
\begin{equation*}\textstyle
\tilde{H}(\psi)=\theta^{-r-1}H(\theta\psi),
\end{equation*}
provided in the definition of $\tilde{H}$ that 
$\tilde{a}_{\tau}=a_{\tau}$,
$\tilde{a}_{\biw}=\theta^{-2r-4}a_{\biw}$, and
$\tilde{a}_{\rho}=\theta^{-t-r-1}a_{\rho}$.
We have proved the following well known result.

\begin{lemma}\label{l:conf-inv}
Assume the above setting, so in particular,
$\tilde{a}_{\tau}=a_{\tau}$,
$\tilde{a}_{\tbw}=\theta^{-2r-4}a_{\tbw}$, and
$\tilde{a}_{\rho}=\theta^{-t-r-1}a_{\rho}$.
Then we have
\begin{equation*}
\begin{split}
\tilde{H}(\psi)=0
\quad\Leftrightarrow\quad
H(\theta\psi)=0,\\
\tilde{H}(\psi)\geqs0
\quad\Leftrightarrow\quad
H(\theta\psi)\geqs0,\\
\tilde{H}(\psi)\leqs0
\quad\Leftrightarrow\quad
H(\theta\psi)\leqs0.
\end{split}
\end{equation*}
\end{lemma}

\subsection{General conformal rescaling and the near-CMC condition}
\label{sec:rescaling}

In this article we focused on the standard conformal method to produce 
the particular coupled elliptic PDE system that we analyzed.
Here we examine briefly other decompositions to see if it
is possible to remove the near-CMC obstacle for non-CMC existence
that still seems to remain for the non-positive Yamabe classes 
and for the positive Yamabe class with large data.

The key question here is whether or not the standard conformal method
essentially hard-wires the near-CMC assumption into the coupled system
in order to get a domain of attraction for fixed-point iterations.
If this is the case, then there remains the possibility that one can 
reverse-engineer a formulation, different from the conformal method, 
that gives a domain of attraction (preferably a contraction so that we 
also get uniqueness) without use of near-CMC conditions.
Unfortunately, the answer appears to be negative, as we demonstrate below.
In particular, it seems that the near-CMC obstacle is present in all
possible formulations based on conformal transformations, if the
estimate~\eqref{CS-aLw-bound} is used.

To begin, recall that
the objects $(\cM, \hat h_{ab}, \hat k_{ab}, \hat \rho,\hat j_a)$ form
an $n$-dimensional {\bf initial data set} for Einstein's equations iff
$\cM$ is a $n$-dimensional smooth manifold, the tensor $\hat h_{ab}$
is a Riemannian metric on $\cM$, the tensor $\hat k_{ab}$ is a
symmetric tensor field on $\cM$, the fields $\hat\rho$ and $\hat j_a$
are a non-negative scalar and a tensor field on $\cM$, respectively,
satisfying the condition $-\hat\rho^2+\hat j_a\hat j^a <0$, and the
following equations hold:
\begin{align}
\label{GCR-hc}
\hat R + \hat k^2 - \hat k_{ab}\hat k^{ab} -2\kappa \hat \rho &=0,\\
\label{GCR-mc}
-\hat \nabla_a \hat k^{ab} + \hat\nabla^b\hat k +\kappa \hat j^b &=0,
\end{align}
where $\hat\nabla_a$ is the Levi-Civita connection of the metric $\hat
h_{ab}$, the scalar field $\hat R$ is the Ricci scalar of the
connection $\hat\nabla_a$, the scalar $\hat k =\hat k_{ab}\hat h^{ab}$
is the trace of the tensor $\hat k_{ab}$, and the constant $\kappa =
8\pi$ in units where both the gravitation constant $G$ and the speed
of light $c$ have value one. The initial data set for Einstein's
equations describe an instant of time in the physical world if we
choose the number $n=3$. Nevertheless, In the calculations that follow
we keep the number $n$ as a general positive integer.

Introduce the decomposition of the two-index tensor $k_{ab}$ into
trace-free and trace parts, as follows,
\[
\textstyle
\hat k^{ab} = \hat s^{ab} + \frac{1}{n}\, \hat k \,\hat h^{ab},
\]
where $\hat s_{ab}\hat h^{ab}=0$. Introduce the following conformal
rescaling:
\begin{equation}
\label{GCR-r}
\hat h_{ab} = \phi^r \, h_{ab},\qquad
\hat s^{ab} = \phi^s \, s^{ab},\qquad
\hat k = \phi^t \, k,
\end{equation}
where the integers $r$, $s$, and $t$ are arbitrary, and we have
introduced the Riemannian metric $h_{ab}$, a symmetric tensor
$s^{ab}$, and a scalar field $k$. Introduce $\nabla_a$ the Levi-Civita
connection of the metric $h_{ab}$, which satisfies the equation
$\nabla_ah_{bc}=0$, and denote by $R$ the Ricci scalar of this
connection $\nabla_a$. The rescaling above induces the following
equations
\[
\hat h^{ab} = \phi^{-r} \, h^{ab},\qquad
\hat s_{ab} = \phi^{(2r+s)} \, s_{ab},
\]
where $\hat h^{ab}$ is the inverse tensor of $\hat h_{ab}$, and
$h^{ab}$ is the inverse tensor of $h_{ab}$. We use the convention that
indices in all other hatted tensors are raised and lowered with the
tensors $\hat h^{ab}$ and $\hat h_{ab}$, respectively, while indices
on unhatted tensors are raised and lowered with the tensors $h^{ab}$
and $h_{ab}$, respectively. For example:
\[
\hat s_{ab} = \hat h_{ac}\hat h_{bd}\hat s^{cd}
= \phi^r h_{ac}\,\phi^r h_{bd} \,\phi^s s^{cd} 
= \phi^{(2r+s)} s_{ab}.
\]
The rescaling introduced in Eq.~(\ref{GCR-r})
implies that the tensor field $\hat k^{ab}$ transforms as follows
\[
\textstyle
\hat k^{ab} = \phi^s\, s^{ab} + \frac{1}{n} \, \phi^{(t-r)} \, k h^{ab}
\LRI
\hat k_{ab} = \phi^{(2r+s)}s_{ab} +\frac{1}{n}\,\phi^{(t+r)}\,k h_{ab}.
\]
The connections $\hat\nabla_a$ and $\nabla_a$ differ in a tensor field
$C_{ab}{}^c$, in the sense that for any tensor field $v_a$ holds
\[
\hat\nabla_av_b = \nabla_av_b - C_{ab}{}^cv_c.
\]
The tensor field $C_{ab}{}^c$ depends on the scalar field $\phi$ and
the number $r$ as follows,
\begin{equation}
\label{GCR-c}
\textstyle
C_{ab}{}^c = r \,\delta_{(a}{}^c\nabla_{b)}\ln(\phi) 
- \frac{r}{2} \,h_{ab}h^{cd}\nabla_d\ln(\phi).
\end{equation}
This expression implies the contractions
\[
\textstyle
h^{ab}C_{ab}{}^c = -\frac{r}{2}(n-2) h^{cd}\nabla_d\ln(\phi),\qquad
C_{ab}{}^b = \frac{nr}{2}\nabla_a\ln(\phi).
\]
Given any two connections $\hat\nabla_a$ and $\nabla_a$ related by a
tensor field $C_{ab}{}^c$, the Riemann, Ricci, and Ricci scalar fields
associated with these two connections are related by the following
expressions
\begin{align*}
\hat R_{abc}{}^d &= R_{abc}{}^d - 2 \nabla_{[a}C_{b]c}{}^d
+ 2C_{c[a}{}^eC_{b]e}{}^d,\\
\hat R_{ac} &= R_{ac} - \nabla_aC_{cb}{}^b + \nabla_bC_{ac}{}^b
+ C_{ca}{}^e C_{eb}{}^b - C_{cb}{}^e C_{ae}{}^b,\\
\hat R &= \phi^{-r} \bigl[ R  - \nabla^aC_{ab}{}^b 
+\nabla_b(h^{ac}C_{ac}{}^b)  + h^{ac}C_{ca}{}^e C_{eb}{}^b 
- h^{ac} C_{cb}{}^e C_{ae}{}^b\bigr],
\end{align*}
where indices between square brackets mean anti-symmetrization, that
is, given any tensor $u_{ab}$ we define $u_{[ab]} := (u_{ab}
-u_{ba})/2$. In the case that the tensor $C_{ab}{}^c$ is given by
Eq.~(\ref{GCR-c}), the Ricci scalars $\hat R$ and $R$ satisfy the
equation
\[
\textstyle
\hat R = \phi^{-(r+1)} \Bigl[ \phi R -r(n-1) \Delta \phi 
-\frac{r}{4\phi}(n-1)[r(n-2)-4](\nabla_a\phi)(\nabla^a\phi)
\Bigr].
\]
Introduce the Hamiltonian and momentum fields
\begin{align*}
\hat H &:= \hat R+ \hat k^2-\hat k_{ab} \hat k^{ab},\\
\hat M^b &:=  -\hat \nabla_a\hat k^{ab} + \hat\nabla^b \hat k,
\end{align*}
then the conformal rescaling given in Eq.~(\ref{GCR-r}) implies the
following equations
\begin{align*}
\hat H &= 
\textstyle
\phi^{-(r+1)} \Bigl[ \phi R -r(n-1) \Delta \phi 
-\frac{r}{4\phi}(n-1)[r(n-2)-4](\nabla_a\phi)(\nabla^a\phi)
\Bigr] \\
&
\textstyle
\quad + \frac{n-1}{n} \phi^{2t} \, k^2 - \phi^{2(r+s)} s_{ab}s^{ab},\\
\hat M_b &= 
\textstyle
- \phi^{(r+s)}\nabla_as_b{}^a 
+ \frac{n-1}{n} \phi^t\, \nabla_b k 
-\Bigl(\frac{rn}{2}+r+s\Bigr) \phi^{(r+s)} s_b{}^a\nabla_a \ln(\phi)\\
\textstyle
&
\textstyle
\quad + \frac{n-1}{n} \,t\,\phi^t\,k\nabla_b\ln(\phi).
\end{align*}
It is convenient to reorder the terms in these equations in such a way
that the equation for the Hamiltonian field is given by
\begin{gather*}
\textstyle
-r(n-1)\Delta\phi 
-\frac{r}{4\phi}(n-1)[r(n-2)-4](\nabla_a\phi)(\nabla^a\phi)\\
\textstyle
+R\phi +  \frac{(n-1)}{n}  \, k^2\,\phi^{(2t+r+1)} 
-  s_{ab}s^{ab}\, \phi^{(3r+2s+1)}
= \phi^{(r+1)}\hat H,
\end{gather*}
and the equation for the momentum field is given by
\begin{gather*}
\textstyle
-\nabla_as_b{}^a
-\Bigl(\frac{(n+2)}{2}\,r+s\Bigr) s_b{}^a\nabla_a \ln(\phi) \\
\textstyle
= \phi^{-(r+s)}\hat M_b -\frac{(n-1)}{n} \phi^{(t-r-s)} \nabla_b k
-\frac{(n-1)}{n}\,t\phi^{(t-r-s-1)}k\nabla_b \phi.
\end{gather*}

There are many interesting particular cases of the equations
above. The first case is to keep the dimension $n\geqs 3$ arbitrary,
and choose:
\[
\textstyle
r = \frac{4}{n-2},\qquad
s = - \frac{(n+2)}{2}\, r,\qquad
t=0,
\] 
then, introducing the number $2^{*}:= 2n/(n-2)$, we conclude
that the $n$-dimensional vacuum Einstein constraint equations ($H=0$,
$M_b=0$) can be written as follows,
\[
\begin{gathered}
\textstyle
-\frac{4(n-1)}{(n-2)}\Delta\phi 
+R\phi + \frac{(n-1)}{n} \, k^2\,\phi^{(2^{*}-1)}  
-  s_{ab}s^{ab}\,\phi^{-(2^{*}+1)} = 0,\\
\textstyle
-\nabla_as_b{}^a
+\frac{(n-1)}{n}\, \phi^{2^{*}} \nabla_b k=0.
\end{gathered}
\]

In the case that the manifold $\cM$ is 3-dimensional, we have the
number $2^{*}=6$, and the equation for the Hamiltonian field is given
by
\begin{gather}
\textstyle
\nonumber
-2r\Delta\phi 
-\frac{r}{2\phi}(r-4)(\nabla_a\phi)(\nabla^a\phi)\\
\textstyle
\label{GCR-rh}
+R\phi + \frac{2}{3} \, k^2\, \phi^{(2t+r+1)}  
-  s_{ab}s^{ab}\,\phi^{(3r+2s+1)}
= \phi^{(r+1)}\hat H,
\end{gather}
and the equation for the momentum field is given by
\begin{gather}
\textstyle
\nonumber
- \nabla_as_b{}^a
-\Bigl(\frac{3r}{2}+r+s\Bigr) s_b{}^a\nabla_a \ln(\phi) \\
\textstyle
= \phi^{-(r+s)}\hat M_b -\frac{2}{3} \phi^{(t-r-s)} \nabla_b k
-\frac{2}{3}\,t\,\phi^{(t-r-s-1)}k\nabla_b \phi.
\label{GCR-rm}
\end{gather}
The {\bf semi-decoupling decomposition} in the case of the vacuum
Einstein constraint equations ($H=0$, $M_b=0$) is obtained from
Eqs.~(\ref{GCR-rh})-(\ref{GCR-rm}) in the particular case of $r=4$,
$s=-10$, and $t=0$, that is,
\[
\begin{gathered}
\textstyle
-8\Delta \phi +R\phi +\frac{2}{3}k^2\phi^{(2^{*}-1)} 
- s_{ab}s^{ab}\phi^{-(2^{*}+1)}  = 0,\\
\textstyle
-\nabla_as_b{}^a +\frac{2}{3}\phi^{2^{*}}\nabla_b k=0.
\end{gathered}
\]
The {\bf conformally covariant decomposition}, in the case of the vacuum
Einstein constraint equations ($H=0$, $M_b=0$) and in the case that
the transverse, traceless part of the tensor $k_{ab}$ vanishes, is
obtained from Eqs.~(\ref{GCR-rh})-(\ref{GCR-rm}) with the particular
choice of $r=4$, $s=-4$, and $t=0$, that is,
\[
\begin{gathered}
\textstyle
-8\Delta \phi +R\phi 
+ \Bigl(\frac{2}{3}k^2 - s_{ab}s^{ab}\Bigr) \,\phi^{(2^{*}-1)} 
= 0,\\
\textstyle
-\nabla_as_b{}^a -6\,s_b{}^a\nabla_a\ln(\phi) 
+\frac{2}{3} \nabla_b k=0.
\end{gathered}
\]
As a final example, it is interesting to write down the rescaled
equations above in the case $r=4$, $s=-10$, $t$ arbitrary:
\begin{gather*}
\textstyle
-8\Delta\phi + R\phi + \frac{2}{3} \phi^{(2t+5)} \, k^2 
- \phi^{-7} s_{ab}s^{ab} = \phi^{5}\hat H,\\
\textstyle
- \nabla_as_b{}^a = \phi^{6}\hat M_b 
-\frac{2}{3} \phi^{(t+6)} \nabla_b k
-\frac{2}{3}\,t\,\phi^{(t+5)}k\,\nabla_b \phi.
\end{gather*}

Since the leading power in each equation scales exactly as the
conformal method, the same argument leading to the negative result for 
the conformal method in Lemma~\ref{L:global-super-limit} will apply here.
Therefore, it appears that the different conformal rescalings produce
coupled systems leading to precisely the same form of the
near-CMC condition to establish non-CMC existence, in the case 
of both the non-positive Yamabe classes and the positive Yamabe 
class for large data.

\bibliographystyle{plain}
\bibliography{../bib/ref-gn}

\end{document}